\documentclass[11pt]{article}
\usepackage{template}
\usepackage{mathrsfs,mathtools}
\usepackage{enumitem}
\mathtoolsset{showonlyrefs}
\usepackage{newtxtext,newtxmath}

\title{Finding Gaussian Structure in Bosonic States}
\author{Alvan Arulandu\thanks{Harvard University, aarulandu@college.harvard.edu.} \and Sitan Chen\thanks{Harvard University, sitan@seas.harvard.edu} \and Ziyun Chen\thanks{University of Washington, ziyuncc@cs.washington.edu}  \and Jerry Li\thanks{University of Washington, jerryzli@cs.washington.edu} \and Eric Ma\thanks{Stanford University, eryma@stanford.edu}}

\newcommand{\sitan}[1]{}
\newcommand{\ziyun}[1]{}
\newcommand{\eric}[1]{}
\newcommand{\jerry}[1]{}
\renewcommand{\alvantodo}[1]{}

\begin{document}

\maketitle

\begin{abstract}
We study agnostic tomography of pure bosonic Gaussian states: given copies of an arbitrary
$n$-mode bosonic state $\rho$, the goal is to output a pure Gaussian state whose infidelity with $\rho$ is at most $\opt + \eps$, where
$\opt$ is the minimum infidelity achievable by any pure Gaussian state. 

We give efficient protocols achieving this in both the high and low fidelity regimes. When $\mathrm{opt}$ is below some universal constant, our protocol has runtime and copy complexity which is strongly polynomial in $n, 1/\eps$ and $\log \log E$, where $E$ is the energy of the closest pure Gaussian state. For arbitrary $\opt$, our protocol uses $(n+1)^{\mathrm{poly}(1/\epsilon)} \mathrm{poly}\left(1+\log\log(E)\right)$ copies and runtime. As a corollary, we obtain the first truly \emph{tolerant} Gaussianity testing protocol for distinguishing whether $\opt > c + \epsilon$ or $\opt < c - \epsilon$, for any threshold $c\in(0,1)$. We also prove $\mathrm{poly}(n,1/\epsilon)$ runtime is impossible, unless $\mathrm{NP}\subseteq\mathrm{BQP}$.

Our protocols follow a shared paradigm: first, we iteratively use general Gaussian measurements combined with techniques from classical robust statistics to obtain a good warm start estimate, then we leverage non-Gaussian measurements to refine this warm start using convex and non-convex optimization methods.
Interestingly, we prove that non-Gaussian measurements are necessary to match the strong agnostic guarantees we obtain, and in fact these guarantees are provably superior to what is possible for robustly estimating \emph{classical} Gaussians.

\end{abstract}

\thispagestyle{empty}

\newpage

{\footnotesize
\tableofcontents}
\setcounter{page}{0}
\thispagestyle{empty}

\newpage

\section{Introduction}
Gaussian states are the backbone of bosonic quantum information. They are the states generated by the workhorse operations of quantum optics such as displacement, squeezing, and passive linear optics, and they are what experiments prepare when they inject squeezed light into gravitational-wave detectors~\cite{grote2013first,tse2019quantum,acernese2019increasing,ganapathy2023broadband} and into microwave haloscopes searching for axion dark matter~\cite{malnou2019squeezed,backes2021quantum}, or when they send squeezed light through a large interferometer to perform Gaussian boson sampling~\cite{hamilton2017gaussian,zhong2020quantum,madsen2022quantum}. Gaussian states also play the role of mean-field theory for bosons: just as product states and Slater determinants are the mean-field ans\"atze for qubits and fermions, Gaussian states are the variational family behind Bogoliubov theory and its modern descendants~\cite{bogoliubov1947theory,shi2018variational,hackl2020geometry}. 
%And they are exceptionally well understood. A Gaussian state is specified by a mean vector and a covariance matrix, and an exact calculus describes its properties and its evolution under quadratic Hamiltonians~\cite{weedbrook2012gaussian,serafini2023quantum}.

%Non-Gaussianity is therefore the resource behind essentially every route to bosonic quantum advantage and fault tolerance, from cat qubits~\cite{vlastakis2013deterministically,lescanne2020exponential,grimm2020stabilization,putterman2025hardware} and Gottesman--Kitaev--Preskill (GKP) codes~\cite{gottesman2001encoding,fluhmann2019encoding,campagne2020quantum,sivak2023real,konno2024logical,larsen2025integrated} to modular photonic quantum computers~\cite{aghaeerad2025scaling}; see~\cite{walschaers2021nongaussian,lvovsky2020production} for reviews.

Yet, the bosonic states that matter most for quantum technology are typically not Gaussian.
However, many important states that arise both in physical experiments and in the design of quantum computers are often close to Gaussian.
This can be due to noise during state preparation: for instance, fluctuations of the relative phase between a squeezer and its reference turn a squeezed state into a mixture of rotated squeezed states~\cite{franzen2006experimental,dwyer2013squeezed,kwee2014decoherence}, and in a superconducting cavity the self-Kerr nonlinearity inherited from a dispersively coupled transmon continually bends displaced states away from Gaussianity~\cite{kirchmair2013observation}. 
Alternatively, when Gaussian states are used in mean-field approximations, the real underlying state is understood to be only well-approximable by Gaussian states.
Indeed, non-Gaussianity is a crucial concept in bosonic quantum information, underlying essentially every route to bosonic quantum advantage and fault tolerance, from cat qubits~\cite{vlastakis2013deterministically,lescanne2020exponential,grimm2020stabilization,putterman2025hardware} and Gottesman-Kitaev-Preskill (GKP) codes~\cite{gottesman2001encoding,fluhmann2019encoding,campagne2020quantum,sivak2023real,konno2024logical,larsen2025integrated} to modular photonic quantum computers~\cite{aghaeerad2025scaling}; see e.g.~\cite{walschaers2021nongaussian,lvovsky2020production} for reviews.
%In other cases, states are intentionally non-Gaussian.
%Cat and Gottesman-Kitaev-Preskill (GKP) states are manufactured from Gaussian resources together with a non-Gaussian ingredient such as photon detection or a nonlinearity~\cite{ourjoumtsev2006generating,neergaard2006generation,vlastakis2013deterministically,konno2024logical,larsen2025integrated}, and they retain conspicuous Gaussian structure: a cat state is a superposition of coherent states, and a GKP state is a superposition of squeezed states.

In both situations, it is natural to ask what Gaussian structure the state contains, and how much of the state this structure accounts for.
To this end, we study the task of \emph{agnostic tomography of pure bosonic Gaussian states}.
We consider the following problem: given copies of an $n$-mode bosonic state $\rho$, we would like to find a classical description of a pure Gaussian state $\ket{\psi}$ maximizing the fidelity $F(\rho,\psi) = \bra{\psi}\rho\ket{\psi}$.
More formally, if we let $\c G$ denote the set of pure Gaussian states, for a given state $\rho$, we can let
\begin{equation}
\opt = 1 - \sup_{\ket{\psi} \in \c G} \bra{\psi}\rho\ket{\psi} \label{eq:optdef}
\end{equation}
denote the best infidelity to $\rho$ achievable by any Gaussian state.
A (proper) agnostic tomography algorithm with error parameter $\eps > 0$ is one which, given copies of $\rho$, outputs a $\ket*{\hat \psi} \in \c G$ so that $1 - \bra*{\hat \psi} \rho \ket*{\hat \psi} \leq \opt + \eps$.
Our goal is to achieve such an algorithm that is efficient, both in terms of the number of copies of $\rho$ it requires, as well as in its overall runtime.

An algorithm for this task would have a number of natural applications.
First, it would immediately imply an algorithm for estimating $\opt$ to error $\eps$, by simply measuring the overlap of $\rho$ with $\ket*{\hat \psi}$.
This is a natural fidelity-based measure of non-Gaussianity, in the spirit of resource theories of non-Gaussianity~\cite{takagi2018convex,albarelli2018resource}: it is the bosonic analogue of the stabilizer fidelity~\cite{chen2025stabilizer} and of the geometric measure of entanglement~\cite{bakshi2025learning}. It also differs from the distance-based measures of~\cite{genoni2007measure,genoni2008quantifying,ghiu2013measures} which compare $\rho$ the Gaussian state with the same first and second moments, rather than against all possible Gaussian states. Estimating it would allow us to measure the utility of mean-field approximations in real-world settings.
The same quantity sets the threshold in fidelity-based witnesses of stellar rank~\cite{chabaud2020stellar,chabaud2021certification,fadrny2024experimental}: in those works, $\opt$ is explicitly computed classically for a known non-Gaussian target state known target state, and any experimental state whose fidelity with the target exceeds it is thus certified to be non-Gaussian. In contrast, our algorithm estimates the Gaussian fidelity of the unknown state directly, with no reference state.

But agnostic tomography would allow us to do much more than just estimating the fidelity.
Writing $\ket*{\hat \psi} = G\ket{0}$ for a Gaussian unitary $G$, it singles out the rotated Fock basis $\{G\ket{m}\}$, in which $\rho$ is as close to the vacuum as possible. 
For a state which is close to a Gaussian state, this is the state that the device is effectively trying to prepare, and thus serves as a natural Gaussian ``sanitization'' of the state. 
For more general states, one can interpret this basis as a Gaussian skeleton: in this basis, most of the remaining non-Gaussian structure sits at low photon number, where it is most economical to describe, in the spirit of the stellar representation of non-Gaussian states~\cite{chabaud2020stellar,fadrny2024experimental}.

\subsection{Results}
We give the first protocols for agnostic tomography of pure bosonic Gaussian states in (in)fidelity in both the high-fidelity and low-fidelity regimes. Let $d_{\rm IF}(\rho,\sigma)\equiv 1-F(\rho,\sigma)$ denote the infidelity between two quantum states, where $F(\rho,\sigma)\equiv (\tr \sqrt{\sqrt{\rho}\sigma\sqrt{\rho}})^2$ is the fidelity. Further, let $\c G $ denote the class of pure bosonic Gaussian states. Informally, our results are as follows:
\begin{theorem}[Agnostic tomography in high-fidelity regime (informal, see Theorem~\ref{thm:hf})]
     There is a universal constant $\overline{\opt} \in (0,1)$ such that the following holds. Suppose $\rho$ is an arbitrary $n$-mode bosonic state, and assume that $\opt\leq {\overline{\opt}}$, where $\opt$ is defined in Eq.~\eqref{eq:optdef}. Let $E$ be the energy of some closest state in $\c G$. Then, there exists an algorithm which takes an $\epsilon > 0$ along with $\poly(n,\epsilon^{-1},\log \log E)$ copies and, using only Gaussian unitaries, heterodyne measurements, and photon counting, outputs a $\ket*{\hat \psi}\in \c G$ satisfying
     $$d_{\rm IF}(\rho,\ket*{\hat \psi})\leq \opt +\epsilon $$
     with probability at least $0.99$.
\end{theorem}

\noindent To our knowledge, this is the first robust Gaussian state learning guarantee for any value of $\opt$ which is bounded away from $0$. Additionally, whereas the above guarantee applies in the ``high-fidelity'' regime where the unknown state is sufficiently close to Gaussian, we also devise an algorithm in the ``low-fidelity'' regime where the unknown state is arbitrary:

\begin{theorem}[Agnostic tomography in low-fidelity regime (informal, see Theorem~\ref{thm:lf-agnostic})]\label{thm:lowfidelity_informal}
    Suppose $\rho$ is an arbitrary $n$-mode bosonic state. Let $E$ be the energy of some closest state in $\c G$. Then, there exists an algorithm which takes an $\epsilon > 0$ along with $\poly(n^{1/\epsilon},\epsilon^{-1},\log \log E)$ copies and, using only Gaussian unitaries, heterodyne measurements, and photon counting, outputs a $\ket*{\hat \psi}\in \c G$ satisfying
     $$d_{\rm IF}(\rho,\ket*{\hat \psi})\leq \opt +\epsilon$$
     with probability at least $0.99$.
\end{theorem}

\noindent Thus, for any arbitrarily small constant fidelity, we can approximate the closest Gaussian state in time and samples polynomial in the system size. As an immediate corollary, our protocol yields the first \emph{tolerant tester} for Gaussianity that can distinguish between states for which $\opt > c + \epsilon$ versus $\opt < c - \epsilon$ for any threshold $c$. This is the first guarantee of its kind for bosonic states. Recent algorithms for Gaussianity testing~\cite{girardi2025gaussian,iyer2026optimal,yaghubirad2026testing} only tolerate small constant-infidelity deviations from Gaussianity, enabling one to detect when a bosonic state is either \emph{close} or \emph{far} from being Gaussian, but not to distinguish between fidelity above or below a specified threshold, even if that threshold is, say, merely $1/2$. In other words, ours is the first guarantee for bosonic states, whether for learning or for testing, in the \emph{low-fidelity} regime, where the best Gaussian fidelity is an arbitrary constant bounded away from one. Analogous low-fidelity guarantees were previously known only for \emph{discrete variable} systems~\cite{grewal2024improved,chen2025stabilizer,arunachalam2026marton,bakshi2025learning,bao2026fully}.

Finally, we pause to remark upon two notable features of our results. First, the guarantee of Theorem~\ref{thm:lowfidelity_informal} scales exponentially in the target precision $\epsilon$. While one might hope for a $\mathrm{poly}(n,\epsilon^{-1})$ bound, in Section~\ref{sec:clique-lb} we prove that this is impossible:

\begin{theorem}[Fully polynomial runtime is impossible (informal, see Theorem~\ref{thm:lowfidelity_informal})]
    Agnostic tomography of pure Gaussian states is at least as hard as finding maximal cliques in graphs. Thus, unless $\mathrm{NP}\subseteq\mathrm{BQP}$, there is no algorithm for agnostic tomography of pure Gaussian states which achieves runtime $\mathrm{poly}(n,1/\epsilon)$.
\end{theorem}

\noindent In addition, our protocols crucially rely on \emph{non-Gaussian measurements}, namely photon counting. In Section~\ref{sec:gaussian-measurements}, we show that this is in fact necessary, giving an intriguing new separation between Gaussian measurements and non-Gaussian measurements:

\begin{theorem}[Gaussian measurements are provably weaker than non-Gaussian (informal, see Theorems~\ref{thm:hf-g} and~\ref{thm:gaussian-measurement-hardness})]
    In the high-fidelity regime, there is a time-efficient protocol for agnostic tomography of pure Gaussian states which only uses Gaussian measurements and achieves infidelity $O(\opt) + \epsilon$ using $\tilde{O}(n^2/\epsilon)$ copies. However, there exist settings of $\opt$ such that any protocol in this regime that only uses Gaussian measurements cannot achieve infidelity $\opt + \epsilon$ without using $\exp(\Omega(n^{1/3}))$ copies.
\end{theorem}

% Finally, our algorithms accrue a $\log\log E$ energy overhead from a Gaussian pre-processing step which iteratively unsqueezes the state, following the approach of~\cite{bittel2025energy}. This appears in prior Gaussian state tomography algorithms \citep{bittel2025energy,chen2026towards} and is known to be necessary when using Gaussian measurements for the task \citep{chen2026log}. Since both algorithms use non-Gaussian measurements despite being unable to remove the energy scaling, it is an open question whether this can be removed entirely in either regime via a more clever use of non-Gaussianity.

\subsection{Our Techniques}

Our algorithms arise from a shared paradigm, which we believe may be generalizable to other settings.
At its heart, our algorithms follow a two-step procedure, which is then possibly iterated depending on the setting:
\begin{enumerate}
    \item First, use \emph{heterodyne measurements} to reduce the problem to an instance of a \emph{classical} problem from \emph{high-dimensional robust statistics}, and leverage existing techniques from that literature to obtain a good warm-start estimate of the best Gaussian approximation.
    \item Leverage this warm-start to rotate the states so that in this new basis, the best Gaussian approximation approximately lives in a low photon number subspace. Finding the closest Gaussian state then amounts to solving a certain tensor optimization problem in this subspace.
\end{enumerate}
We expand upon these two points below.

\paragraph{Heterodyne Gaussian measurement and robust statistics}
Gaussian measurements are a powerful, experimentally feasible, class of measurements one can perform on a bosonic state.
Crucially, when performed on a Gaussian state, the output distribution of the measurement is a sample from a classical high-dimensional Gaussian, whose parameters correspond directly to the parameters of the underlying state and the specific Gaussian measurement applied (see Section~\ref{sec:prelims} for formal definitions).
%Additionally, it turns out that the trace distance between two bosonic states is proportional (at least, for energy-bounded Gaussian states) to the total variation distance between these Gaussians.

A key implication of all of this for our application is that if we apply Gaussian measurements to our (potentially non-Gaussian) state $\rho$, the result of this measurement is a draw from a classical distribution which is close in total variation distance to a Gaussian whose underlying parameters would recover the best bosonic Gaussian state approximation to $\rho$.
This is a well-studied problem in the field of \emph{robust statistics}, and by now, many efficient algorithms are known for this problem in a variety of settings.
This is precisely the oblivious contamination model of robust statistics~\cite{Huber64,DKKLMS16}, and the agnostic problem inherits all of the difficulties of that setting.
Fortunately, a long line of work in algorithmic robust statistics, starting with~\cite{DKKLMS16,LaiRV16}, has produced computationally efficient estimators for high-dimensional distributions whose error does not grow with the dimension; see~\cite{diakonikolas2019recent} for a textbook treatment.
In particular, given samples from a distribution $D$ with total variation distance at most $\eta$ from a Gaussian, this Gaussian can be estimated in polynomial time to within $\widetilde{O}(\eta)$ in total variation distance~\cite{DKKLMS16}.

There are two key settings that are relevant for our problem.
When the infidelity of the optimal solution is low (that is, the state $\rho$ is close to being Gaussian), this becomes the standard problem of robustly learning a Gaussian mentioned above.
For this problem, the computational landscape is well understood, and there are by now well-established polynomial time algorithms for this problem with strong error guarantees.
It turns out that this is already sufficient to obtain a good warm start for the problem: by iteratively combining robust Gaussian estimators with techniques and bounds from~\cite{bittel2025optimalfermionic,chen2026towards}, this yields a pure Gaussian state whose infidelity with $\rho$ is $\widetilde{O}(\opt)$, using only Gaussian measurements and $\poly(n, \log\log E)$ copies (Theorem~\ref{thm:reit}), where $E$ is an upper bound on the energy of the best Gaussian state approximation (see Section~\ref{sec:prelims}).

The low-fidelity regime turns out to be considerably more challenging, since here the fidelity $\alpha = 1 - \opt$ of $\rho$ with its best Gaussian approximation $\pi_*$ may be small.
It is then information-theoretically impossible to output a single good estimate.
However, this is the setting of \emph{list-decodable learning}, introduced in~\cite{balcan2008clustering,charikar2017learning}, which has since yielded efficient list-decodable algorithms for mean estimation~\cite{charikar2017learning,DiakonikolasKS18} and, most relevant to us, for estimation of general Gaussians~\cite{ivkov2022list}.
We use the list-decodable Gaussian learner of~\cite{ivkov2022list} on the heterodyne outcomes to obtain a list, of size depending only on $\alpha$, of candidate means and covariances, at least one of which approximates the parameters of $\pi_*$ up to factors depending only on $\alpha$.
Unlike in the high-fidelity regime, such a coarse estimate is not by itself a good Gaussian approximation, but it is good enough to ``unsqueeze'' the state, that is, it allows us to find a rotated basis in which the true state has smaller energy.
We develop a list-decodable variant of iterative unsqueezing which follows every candidate in the list through the corresponding change of basis, and then prunes the resulting bases using a photon-number-based notion of distance between bases, so that only $O(1/\alpha)$ bases survive each round.
After $O(\log\log E)$ rounds, this produces a short list of rotated bases, in at least one of which $\pi_*$ has mean photon number at most $\poly(1/\alpha)$, independent of both $n$ and $E$ (see Section~\ref{ssec:lf-robust-unsqueezing}).

While these methods already yield non-trivial guarantees, they unfortunately run into some fundamental bottlenecks.
In particular, even in the simpler low-infidelity setting, classical robust statistics techniques can only ever hope to achieve ``semi-agnostic'' guarantees, that is, error $\widetilde{O}(\opt)$, as opposed to the strong agnostic guarantee of $\opt + \eps$ that we would ideally hope for.
It is known that such bottlenecks are inherent in classical robust statistics, at least for large classes of algorithms such as SQ algorithms~\cite{diakonikolas2017statistical}.
In fact, in Section~\ref{sec:gaussian-measurements}, we show that no algorithm using only Gaussian measurements can achieve strong agnostic guarantees for this problem.
The situation is even worse in the high infidelity setting: here, if we only used the guarantees known from the robust statistics literature, we would achieve final error guarantees which are almost trivial.

\paragraph{Low photon number projections and tensor optimization.}
To rectify this, we introduce a second algorithmic step which utilizes the prior methodology as a warm start, and uses it to reduce the problem to a optimization problem over a manageable, low-dimensional space.
We now explain this step in more detail.

Rather than directly interpreting the output of a first step as the final output of the algorithm, we view it as a rotated basis $G$ (or, in the low-fidelity regime, a short list of candidate bases) in which the best Gaussian approximation is close to the vacuum state.
Therefore, in such a basis, the problem collapses onto a low photon number subspace (see Section~\ref{sec:prelims} for formal definitions).
Indeed, any pure Gaussian state with mean photon number at most $h$ places all but an $\epsilon$-fraction of its mass on Fock states with at most $K = O\big((h+1)^3 + (h+1)\log(1/\epsilon)\big)$ total photons (Lemma~\ref{lemma:lf-uniform-cutoff}).
Consequently, its fidelity with $\rho$ is determined, up to small error, by the projection of $\rho$ onto this subspace, even though $\rho$ itself may have arbitrarily large photon number.
This is in the spirit of the stellar representation of non-Gaussian states~\cite{CMG20}: in the right rotated basis, the non-Gaussian structure that matters sits at low photon number.
It can also be thought of as a natural bosonic analog of the key structural statement utilized by~\cite{bakshi2025learning} (see e.g. Lemma 5.17 in that paper).

In the high-fidelity regime, we show that the warm-start is sufficiently good so that, in the obtained basis, only the vacuum, one-photon, and two-photon sectors matter.
Moreover, we demonstrate that the infidelity is locally smooth and strongly convex in a neighborhood of the vacuum.
Therefore, we can directly run gradient descent on $f_\rho$, re-centering the basis at the current iterate after every step so that we only ever need first-order information at the vacuum.
This by doing so, we are able to refine the warm start to infidelity $\opt + \epsilon$ using $\widetilde{O}(n^2/\epsilon)$ additional copies.
We refer the reader to Section~\ref{ssec:gd} and Theorem~\ref{thm:hf} for a full description and analysis of our algorithm.

In the low-fidelity regime, no such local convexity is available, and the optimization is genuinely global.
However, our warm start still allows us to project the problem onto the subspace spanned by all Fock states with at most $K$ photons, which has dimension $(n + 1)^K$.
This allows us to develop the bosonic analog of the tensor optimization procedure used in~\cite{bakshi2025learning} for outputting a good product state approximation given a sufficiently good warm start estimate.
At a high level, we demonstrate how to reduce this problem to an approximate tensor optimization problem on a tensor with bounded Frobenius norm.
We can then use the tensor optimization routine developed in~\cite{bakshi2025learning} to obtain our final solution.
Interestingly, our reduction is not black-box: this is because a primary difficulty in the reduction is to deal with the fact that the states must be normalized naively renders the problem into a rational optimization problem, rather than a polynomial optimization problem.
In~\cite{bakshi2025learning}, they deal with this problem by simply brute-forcing enough quantities to allow them to ignore this constraint.
We cannot do this: instead, we show that this problem can be, in large part, ignored.
As in~\cite{bakshi2025learning}, we show that the unnormalized problem can be reduced to an enumeration problem over a small list of possible answers, where each answer corresponds to a set of candidate solutions, projected down to a low-dimensional subspace.
We show that the task of associating each projected solution back to normalized Gaussian state with as small norm as possible is in fact a convex optimization problem.
Hence, we can efficiently solve this problem, and thus directly compute the answer of the normalized problem for each candidate solution the unnormalized one.
We defer the reader to Section~\ref{ssec:lf-bounded-energy} for a full algorithm description and analysis.

\subsection{Related work}
\label{sec:relatedwork}

\paragraph{Agnostic tomography.} Agnostic tomography was proposed in recent work~\cite{buadescu2021improved,anshu2024survey,grewal2024agnostic} and has been studied for stabilizer product states~\cite{grewal2024agnostic}, stabilizer states and states of high stabilizer dimension~\cite{grewal2024improved,chen2025stabilizer,arunachalam2026marton}, product states~\cite{bakshi2025learning,bao2026fully}, and Slater determinants~\cite{paul2026learning}. \cite{arulandu2025agnostic} gave the first guarantees for a class of mixed states, namely product mixed states, through a connection to robust learning of binary product distributions. The present work continues this connection between agnostic tomography and robust statistics, establishing the first link to the rich literature on classical robust statistics over Euclidean space (see below). Related notions include agnostic process tomography~\cite{wadhwa2025agnostic} and learning from adversarially corrupted copies~\cite{aliakbarpour2025adversarially}. To our knowledge, ours is the first work to study agnostic tomography in the continuous-variable setting.

\paragraph{Learning bosonic Gaussian states.} The theory of bosonic quantum learning has developed rapidly and we refer to the recent survey of \cite{mele2026advances} for a comprehensive overview. Mele et al.~\cite{mele2025learning} showed that general bosonic states require a number of copies exponential in the number of modes, while Gaussian states and states with few non-Gaussian gates can be learned efficiently. Subsequent work on Gaussian state learning reduced the energy dependence to $\log\log E$ through adaptive unsqueezing~\cite{bittel2025energy,chen2026towards}, which was subsequently shown to be unavoidable for Gaussian seed measurements~\cite{chen2026log}. The optimal sample complexity of Gaussian state tomography was determined in the pure state case by~\cite{chen2026log} and in the general case by~\cite{chen2026optimal}. Recent work also characterized when Gaussian states are as easy to learn as classical Gaussians~\cite{chen2026classical}. All of these results pertain to the setting where the unknown state is \emph{exactly} Gaussian, whereas our work is the first to consider the agnostic setting.

\paragraph{Testing and quantifying non-Gaussianity.} There are multiple different ways of quantifying non-Gaussianity that have been considered in the literature, either through distances to Gaussian reference states~\cite{genoni2007measure,genoni2008quantifying,ghiu2013measures}, through resource theories~\cite{takagi2018convex,albarelli2018resource}, through the stellar rank~\cite{chabaud2020stellar}, through the symplectic rank~\cite{mele2025symplectic}, or through correlations generated at a beam splitter~\cite{hahn2025measuring}.
There are a number of proposed ways to certify non-Gaussianity experimentally, either through photon-counting statistics~\cite{filip2011detecting,jezek2011experimental,lachman2019faithful}, fidelity with \emph{specific} non-Gaussian target states whose maximal Gaussian fidelity is computed classically~\cite{chabaud2021certification,fadrny2024experimental}, and fidelities with \emph{specific} Gaussian targets using Gaussian measurements~\cite{aolita2015reliable}. In contrast, our work is the first to show that it is possible to provably estimate the maximum fidelity with respect to \emph{all} Gaussian states.

On the algorithmic side, \cite{girardi2025gaussian} introduced symmetry-based tests that accept every pure Gaussian state with certainty, showed that they reject states $\epsilon$-far from Gaussian using a number of copies polynomial in the number of modes and the energy, gave a somewhat tolerant variant, and proved that testing whether a mixed state is Gaussian requires exponentially many copies. Two concurrent works~\cite{iyer2026optimal,yaghubirad2026testing} give testers for pure states using $\Theta(1/\epsilon^2)$ copies. The resulting testers distinguish between states within trace distance $\epsilon_1$ of Gaussian from states at trace distance at least $\epsilon_2$, but only for restrictive settings of $\epsilon_1, \epsilon_2$. In~\cite{iyer2026optimal}, they require $\epsilon_2 > \beta\epsilon_1$ for large constant $\beta$; in~\cite{yaghubirad2026testing} they require $\epsilon_2 = (1+\eta)\epsilon_1$ with $\epsilon_1 \le c\eta$ and use joint measurements on $O(1/\eta)$ copies~\cite{yaghubirad2026testing}; in~\cite{girardi2025gaussian} they require $\epsilon_1 \le C\min\{\epsilon_2^2, (nE)^{-4}\}$. In all of these cases, $\epsilon_1$ must thus be at most a small constant, so that the testers certify fidelity close to one. Our algorithm can instead be used to distinguish $\opt \le c - \epsilon$ from $\opt \ge c + \epsilon$ for \emph{any} constant $c \in (0,1)$, answering an open question of~\cite{girardi2025gaussian}. Additionally, whereas \cite{iyer2026optimal,yaghubirad2026testing} assume pure inputs, our protocol works on arbitrary and possibly mixed inputs. Nevertheless, our guarantee is incomparable as we require copies and time polynomial in the number of modes, as we are solving a harder \emph{learning} task.

\paragraph{Classical robust statistics and list-decodable estimation.} Classical robust statistics dates back to Tukey and Huber~\cite{tukey1960survey,Huber64}. Computationally efficient robust estimators for high-dimensional Gaussians were first developed in~\cite{DKKLMS16,LaiRV16}. We use the sum-of-squares estimator of~\cite{kothari2022polynomial} and the bounded-covariance mean estimator of~\cite{diakonikolas2020outlier}. List-decodable learning was introduced in~\cite{charikar2017learning}, and list-decodable estimation of Gaussian means and covariances was developed in~\cite{DiakonikolasKS18,ivkov2022list}. The classical literature on algorithmic robust statistic is too vast to do justice here, and we refer the reader to~\cite{diakonikolas2023algorithmic} for a textbook treatment.
%Bousquet, Kane, and Moran~\cite{bousquet2019optimal} showed that agnostic density estimation in total variation must, in general, lose a constant factor over the best approximation in the class.

\paragraph{Gaussian versus non-Gaussian measurements.} One notable feature of our results is a separation between what can be done for agnostic tomography of Gaussian states under Gaussian measurements versus non-Gaussian measurements. Such separations for other tasks has a long history of study within quantum optics, for example the gap between Gaussian receivers and the Helstrom bound for discriminating coherent states~\cite{takeoka2008discrimination}, which non-Gaussian receivers based on photon counting and feedback can overcome~\cite{cook2007optical}. Other examples include bosonic data hiding against Gaussian operations and classical communication~\cite{sabapathy2021bosonic,lami2021quantum}, the dependence on number of modes for learning mixed Gaussian states~\cite{chen2026towards}, and the dependence on the energy for learning pure Gaussian states~\cite{chen2026log}.

% \alvantodo{maybe say more about classical connections to our learners?}

\section{Preliminaries}
\label{sec:prelims}
Throughout the paper, $C$ and $c$ denote sufficiently large and sufficiently small positive universal constants, respectively; their values may change between occurrences.

Let $D$ denote the Fr\`echet derivative such that for function $f:V \to \R $, $D^k f(\theta)$ denotes the $k$-linear form where $D^k f(\theta)[h_1,\dots,h_k]$ is the $k$-th derivative evaluated in the directions $h_1,\dots,h_k\in V$. For example, $Df(\theta)[h]\equiv D_hf(\theta)$ is the directional derivative and $D^2 f(\theta)[h_1,h_2]=\langle h_1,\nabla^2 f(\theta)h_2\rangle_{V}$ is a quadratic form involving the Hessian. 

\subsection{Matrix decompositions}
We say a matrix $A \in \R^{2n \times 2n}$ is real symplectic if $A \Omega_n A^\top = \Omega_n$, where $\Omega_n$ is the standard symplectic form:
\[
\Omega_n = \begin{bmatrix}
    0 & I_n \\
    -I_n & 0 
\end{bmatrix} \; .
\]
\paragraph{Williamson decomposition}
A positive definite matrix $A \in \R^{2n \times 2n}$ can be decomposed as $A = S(I_2 \otimes D) S^\top$, where $S$ is symplectic and $D$ is diagonal.
The diagonal entries of $D$ are called the \emph{symplectic eigenvalues of $A$}, and are always positive.

\paragraph{Takagi factorization.}
Every complex symmetric matrix $A=A^\top$ admits a factorization
\begin{equation}\label{eq:takagi-factorization}
    A=USU^\top,\qquad U^\dagger U=I,\qquad
    S=\operatorname{diag}(s_1,\ldots,s_n),\quad s_j\geq0.
\end{equation}
The $s_j$ are the singular values of $A$: indeed, $U^\top\bar U=I$ gives $AA^\dagger=US^2U^\dagger$.

\subsection{Bosonic quantum systems}
Here, we briefly review some of the basics of bosonic quantum information.
We will assume that the reader has some familiarity with this concept throughout this paper, and we refer the interested reader to~\cite{eisert2003introduction,serafini2023quantum} for more complete introduction to the topic.

\paragraph{Fock space}
An $n$-mode bosonic system is a system on an infinite-dimensional Hilbert space.
The Hilbert space of a single mode is $\bH_1 = L^2 (\R)$, and the standard basis for this is the \emph{Fock basis} $\c F_1 = \{\ket{m}: m \in \N \}$.
The Hilbert space of $n$ modes $\bH_n$ is the tensor product of $n$ single-mode Hilbert spaces, with associated Fock basis $\c F_n = \{\ket{m}: m \in \N^n\}$.
Each mode has an associated \emph{annihilation} and \emph{creation} operator $\hat{a}_i,\hat{a}_i^\dagger$, which can be defined by their action on the Fock basis:
\[
\hat{a}_i \ket{m} = \sqrt{m_i} \ket{m - e_i} \; , \qquad \hat{a}_i^\dagger \ket{m} = \sqrt{m_i + 1} \ket{m + e_i} \; ,
\]
where $e_i$ is the $i$-th standard basis vector.
These operators satisfy the commutation relations $[\hat{a}_i,\hat{a}_j^\dagger]=\delta_{ij}I$.
These define the \emph{position} and \emph{momentum} quadrature field operators, respectively, given by
\[
\hat{\bf R}=(\hat q_1,\dots,\hat q_n,\hat p_1,\dots,\hat p_n)^\top 
\]
by $\hat q_i=\frac{\hat{a}_i+\hat{a}_i^\dagger}{\sqrt 2}$ and $\hat p_i=\frac{\hat{a}_i-\hat{a}_i^\dagger}{i\sqrt 2}$. 
Then, we have that $[\hat {\bf R}_i,\hat {\bf R}_j]=i\Omega_{n,ij}$. 
Also define the number operator $N=\sum_{i=1}^n \hat{a}_i^\dagger \hat{a}_i$.

For a general $n$-mode bosonic state $\rho$, we can define its \emph{Wigner function} $W_{\rho}: \R^{2n} \to \R$ by
\[
W_\rho (q, p) = \frac{1}{(2 \pi)^n} \int_{\R^n} e^{-i (p \cdot y)} \bra{q + \frac{y}{2}} \rho \ket{q - \frac{y}{2}} d^n y \; .
\]
This is a quasi-probability distribution. Its mean is ${\bf m}$ and its covariance is $V/2$, where
\[
{\bf m} = {\bf m}(\rho)=\tr(\rho\hat{\bf R}) \; , \qquad V = V(\rho)_{ij}=\tr(\rho\{\hat{\bf R}_i -{\bf m}(\rho)_i,\hat{\bf R}_j -{\bf m}(\rho)_j\}) \; .
\]
We refer to $\bf m$ as the mean or displacement vector of $\rho$, and $V$ is its \emph{covariance matrix}; with this normalization, the vacuum has $V=I$ and Wigner function $\cN(0,I/2)$.
$V$ is a real, symmetric matrix, and moreover always satisfies the following uncertainty relation 
\begin{equation}\label{eq:V-psd}
    V+i\Omega_{n}\succeq 0 \; ,
\end{equation}
meaning that not all PSD covariance matrices are realizable.

The energy operator and energy of a state can be written as 
\begin{equation}\label{eq:energy}
    \hat E_n\equiv \frac 1 2 \sum_{j=1}^n(\hat q_j^2+\hat p_j^2)=\frac 1 2 \hat {\bf R}^\top \hat {\bf R},\quad E(\rho)\equiv \tr( \rho \hat E_n)=\frac 1 4 \tr V(\rho)+\frac 1 2 \|{\bf m}(\rho)\|_2^2
\end{equation}
where we note that $\hat E_n=\sum_{j=1}^n \hat{a}_j^\dagger \hat{a}_j +\frac n 2 I $, meaning $E(\rho)\geq n/2$, with equality for the vacuum. 
We also let $\nu(\rho)=\tr(\rho N)=E(\rho)-n/2$ denote the mean photon number of $\rho$.

\paragraph{Gaussian states} A Gaussian state is the class of states whose Wigner function is a Gaussian density. Such states are uniquely determined by their mean and covariance. 
We will use ${\bf m}, V$ to denote the mean and covariance of a quantum Gaussian state.
A Gaussian state is pure if and only if $V = SS^\top$ for some symplectic matrix $S$ in the Williamson decomposition; that is, all of its symplectic eigenvalues are $1$.
One consequence is that the eigenvalues of the covariance matrix of a pure Gaussian state come in pairs $(\lambda, 1 / \lambda)$.

We now record here a helpful bound for converting between various distances for bosonic Gaussian states and their parameters. 

\begin{fact}[Theorem 8 \citep{bittel2025energy}]\label{fact:trd-to-param}
    Let $\pi_1({\bf m}_1,V_1)$ and $\pi_2({\bf m}_2,V_2)$ be two bosonic Gaussian states. Then, 
    \begin{equation}
        \trd(\pi_1,\pi_2)\lesssim \|V_1^{-1/2}({\bf m}_1-{\bf m}_2)\|_2+\tr((V_1^{-1}+V_2^{-1})|V_1-V_2|) \; .
    \end{equation}
\end{fact}

\paragraph{Gaussian unitary} A Gaussian unitary is a unitary operator that maps every bosonic Gaussian state to another Gaussian state.
Formally, a unitary $U$ is Gaussian if and only if $U^\dagger \hat{\bf R}_j U = \sum_{k = 1}^{2n} S_{jk} \hat{\bf R_k} + d_j I$, where $I$ denotes the identity operator, $S \in \R^{2n \times 2n}$ is symplectic, and $d \in \R^{2n}$ is the displacement.
Gaussian unitaries are generated by Hamiltonians which are quadratic in the quadrature operators, and thus physically realizable by standard techniques in quantum optics.
Thus, throughout this paper, we will assume that given a symplectic matrix $S \in \R^{2n \times 2n}$ and displacement $d$, one can efficiently apply its associated Gaussian unitary $U$ in polynomial time.

For our purposes, the most important operational property of Gaussian unitaries is that if $U$ is a Gaussian unitary with associated symplectic matrix $S$ and displacement $d$, and $\rho$ is an arbitrary bosonic state with Wigner function $W_\rho$, then the Wigner function of $\rho' = U \rho U^\dagger$ is given by
\[
W_{\rho'} (x) = W_\rho (S^{-1} (x - d)) \; .
\]
In particular, this means a Gaussian state with mean ${\bf m}$ and covariance $V$ gets mapped to a new Gaussian state with mean $S {\bf m} + d$ and covariance $S V S^\top$.

\paragraph{Gaussian measurements}
Gaussian measurements are a powerful and implementable class of measurements on bosonic systems.
These are POVMs of the form 
\[
M (d^{2n} y) = \frac{1}{(2 \pi)^n} D(y) \tau D(y)^\dagger d^{2n} y \; ,
\]
where $\tau$ is a mean-zero Gaussian state known as the \emph{seed state}, and its covariance $W$ is known as the \emph{seed covariance}.
Here $D(y)$ is a displacement operator satisfying $D(y)^\dagger \hat{\bf R}_j D(y) =\hat{\bf R}_j + y_j I$ for all $j = 1, \ldots, 2n$.

Their key operational definition for our purposes is that when applied to a Gaussian state, the result is a draw from a classical Gaussian distribution.
In particular, if we measure a Gaussian state with mean $\bf m$ and covariance $V$, the result is a draw from $\cN ({\bf m}, \tfrac{V + W}{2})$.
When $W=I_{2n}$, this is called a heterodyne measurement. When $W=K\,{\rm diag}(a,\dots,a,a^{-1},\dots,a^{-1})K^\top$ for an orthogonal symplectic matrix $K$ and $a>0$, it is a general-dyne measurement; in the limit $a\to0$, its outcomes reduce to those of a homodyne measurement of the commuting quadratures $k_1^\top\hat{\bf R},\dots,k_n^\top\hat{\bf R}$, where $k_j$ is the $j$-th column of $K$. 
Both of these are experimentally feasible, and so we will assume throughout this paper that we can implement such a measurement in polynomial time.
More generally, by first applying a Gaussian unitary, and then applying a heterodyne measurement, we obtain a general Gaussian measurement, which we assume can also be implemented in polynomial time.

\paragraph{Energy and covariance.}
% \label{sssec:energy-covariance}

The operator norm of the covariance can be bounded in terms of the mean photon number as follows from~\cite{mele2025achievable}.

\begin{fact}[Lemma~37 of~\cite{mele2025achievable}]\label{fact:energy-covariance}
Let $\rho$ be an $n$-mode bosonic state with covariance matrix $V$ and finite mean photon number $\nu=\nu(\rho)$. Then
\[
    \|V\|_{\rm op}\leq 1+2\nu+2\sqrt{\nu(\nu+1)}\leq4\nu+2.
\]
In particular, for a pure Gaussian state $\pi$ with covariance $V$,
\begin{equation}\label{eq:lfpr-energy-condition}
    \|V^{-1}\|_{\rm op}=\|V\|_{\rm op}\leq4\nu(\pi)+2.
\end{equation}
\end{fact}

\paragraph{Bargmann-Fock space.}

The $n$-mode bosonic Fock space can be realized by holomorphic functions in $n$ variables. 
\begin{definition}[Bargmann-Fock Space]\label{def:bf-space}
    Define the following inner product for entire complex functions in $n$ variables:
    \begin{align}
       \langle f,g\rangle _{\c B_n}  =\pi^{-n}\int_{\C^n }\bar f (z)g(z)e^{-\|z\|^2}dz
    \end{align}
    Then, $\mathcal{B}_n \equiv \{f:\C^n \to \C: \|f\|_{\c B_n}<\infty\}$ is a Hilbert space known as the Bargmann-Fock space. In particular, $\mathcal{B}_n $ can be identified with the $n$-mode bosonic Fock space $\c F_n $ via the mappings:
    \begin{align}
        \ket{k_1k_2\dots k_n}\leftrightarrow\frac{z_1^{k_1}z_2^{k_2}\cdots z_n^{k_n}}{\sqrt{k_1!k_2!\cdots k_n!}},\quad \hat{a}_i^\dagger \leftrightarrow z_i,\quad \hat{a}_i\leftrightarrow \p_{z_i} 
    \end{align}
\end{definition}

\noindent Importantly, this notation gives a unique parameterization for pure Gaussian states. 

\begin{fact}[Bargmann Form]
Every pure Gaussian state $\ket{\psi}$ can be uniquely written as
\begin{equation}
f_\psi(z)=\frac 1 {\sqrt{Z(b,A)}}\exp(b^\top z +\frac 1 2 z^\top Az)
\end{equation}
with $\theta=(b,A)\in \mathcal{P}_n=\{(b,A):b\in \C^n,A\in \C^{n\times n},A=A^\top\}$ and $\|A\|_{\rm op}<1$. The global phase is typically fixed so that the vacuum coefficient is positive. 
In Fock space, these coefficients parameterize the state as follows:
\begin{align}\label{eq:barg-state}
        \ket{\psi(b,A)}&=\frac{1}{\sqrt{Z(b,A)}}\underbrace{\exp\bigg(\overbrace{\sum_{i=1}^n b_i \hat{a}_i^\dagger}^{L_b} +\overbrace{\frac 1 2 \sum_{i,j=1}^n A_{ij}\hat{a}_i^\dagger \hat{a}_j^\dagger}^{Q_A}\bigg)\ket{0}}_{\ket{\Phi(b,A)}}\\
        \label{eq:barg-state-expand}
        &=\frac{1}{\sqrt{Z(b,A)}}\bigg(1+\sum_{i=1}^n b_i \hat{a}_i^\dagger+\frac 1 2 \sum_{i,j=1}^n (A_{ij}+b_ib_j)\hat{a}_i^\dagger \hat{a}_j^\dagger +\cdots \bigg)\ket{0}
    \end{align}
    where $Z(b,A)$ is the normalization factor that gives $\ket{\psi(b,A)}$ unit norm. We also define $K_{\theta}\equiv L_b+Q_A$ for $\theta=(b,A)$. 
\end{fact}

\noindent We first express the displacement and photon number directly in terms of the Bargmann parameters. In particular, the linear parameter $b$ is generally different from the physical displacement.

\begin{lemma}[Displacement and photon number in Bargmann coordinates]\label{lemma:lf-bargmann-moments}
Let $A=A^\top$ with $\|A\|_{\rm op}<1$, let $\ket\psi=\ket{\psi(b,A)}$ and $\pi_{b,A}=\ketbra\psi$, and set $\beta_j=\bra\psi \hat{a}_j\ket\psi$. If $s_1,\ldots,s_n$ are the singular values of $A$, then
\begin{align}
    b&=\beta-A\bar\beta,
    &\beta&=(I-A\bar A)^{-1}(b+A\bar b),\label{eq:lf-displacement-relation}\\
    (1-\|A\|_{\rm op})\|\beta\|_2
    &\leq\|b\|_2
    \leq(1+\|A\|_{\rm op})\|\beta\|_2,\label{eq:lf-displacement-bounds}\\
    \nu(\pi_{b,A})&=\|\beta\|_2^2+
        \sum_{j=1}^n\frac{s_j^2}{1-s_j^2}.\label{eq:lf-bargmann-energy}
\end{align}
\end{lemma}
\noindent We defer the proof to Appendix~\ref{app:proof-lf-bargmann-moments}.

We next record a formula for the normalization $Z(b,A)=\|\Phi(b,A)\|^2$ in \eqref{eq:barg-state}. The unnormalized vector $\ket{\Phi(b,A)}$ corresponds to the Bargmann function $f(z)=e^{b^\top z+z^\top Az/2}$. Thus its squared norm is $\pi^{-n}\int_{\C^n}|f(z)|^2e^{-\|z\|_2^2}\,d^{2n}z$. Using $|e^t|^2=e^{2\Re t}$ gives
\begin{equation}\label{eq:lf-normalization-integral}
    Z(b,A)=\pi^{-n}\int_{\C^n}
       e^{-\|z\|_2^2+2\Re(b^\top z)+\Re(z^\top Az)}\,d^{2n}z.
\end{equation}
The following lemma evaluates this integral, which lets us compute $Z$ and its derivatives using real matrices.

\begin{lemma}[Gaussian integral for the normalization]\label{lemma:lf-normalization-formula}
For $A=A^\top$ with $\|A\|_{\rm op}<1$, define
\[
    M(A)=\begin{pmatrix}I-\Re A&\Im A\\ \Im A&I+\Re A\end{pmatrix},
    \qquad u(b)=\begin{pmatrix}\Re b\\-\Im b\end{pmatrix}.
\]
If $s_1,\ldots,s_n$ are the singular values of $A$, then $M(A)$ has eigenvalues $1-s_j$ and $1+s_j$, and in particular is positive definite. Moreover,
\begin{equation}\label{eq:lf-normalization-formula}
    \log Z(b,A)=-\frac12\log\det M(A)+u(b)^\top M(A)^{-1}u(b).
\end{equation}
\end{lemma}
\noindent We defer the proof to Appendix~\ref{app:proof-lf-normalization-formula}.

\paragraph{Parameter versus trace distance.} We will view $\c P_n $ as a real Hilbert space with dimension $n^2+3n$ equipped with the inner product:
\begin{align}
    \langle (b,A),(b',A')\rangle_{\c P_n }=\Re\left\{\langle b,b'\rangle +\frac 1 2 \tr(A^*A')\right\}
\end{align}
such that $\|(b,A)\|_{\c P}^2=\|b\|_2^2+\frac 1 2\|A\|_F^2$, where we drop the subscript when the mode count is implied. While this parameter distance is invariant to passive Gaussian unitaries, i.e. unitary changes of modes, this is not the case for general unitaries which mix creation and annihilation operators. For similar reasons, the parameter distance between two pure Gaussian states close in trace distance can be arbitrarily large. However, for states close to vacuum, parameter distance can be related to trace distance.
\begin{lemma}\label{lemma:barg-param-L2-small}
   Let $\ket{\psi(b,A)}$ denote a pure Gaussian state which is near-vacuum, in the sense that $\trd(\ketbra{\psi},\ketbra*{0})\leq \xi_0$ for some universal constant $\xi_0$. Then, 
   \begin{align}
       \|(b,A)\|_{\c P} \lesssim \trd(\ketbra{\psi},\ketbra*{0}) \; .
   \end{align}
\end{lemma}
\begin{proof}
Define $\xi\equiv \trd(\ketbra{\psi},\ketbra{0})$ such that $\xi=\sqrt{1-\lvert\braket{0}{\psi(b,A)}\rvert^2}=\sqrt{1-1/Z}$ from \eqref{eq:barg-state-expand}. Taking $\xi_0\leq 1/\sqrt 2$, expanding with \eqref{eq:barg-state-expand} gives:
\begin{align}
    \|(b,A+bb^\top)\|_{\c P }^2 =\|b\|_2^2+\frac 12  \|A+bb^\top \|_{F}^2\leq Z-1=\frac{\xi^2}{1-\xi^2}\leq 2\xi^2  
\end{align}
where the left hand side is the $\ket{1},\ket{2}$ contribution to the norm in $Z$ and where the right hand side removes the $\ket{0}$ contribution. Then, 
\begin{align}
    \|(b,A)\|_{\c P}\leq \|(b,A+bb^\top)\|_{\c P}+\|(0,bb^\top )\|_{\c P}\leq \sqrt 2\; \xi  +\frac{\|b\|_2^2}{\sqrt 2 }\leq \sqrt 2 \; \xi +\sqrt 2 \; \xi^2\lesssim \xi 
\end{align}
which proves the claim. 
\end{proof}

\noindent We will use these near-vacuum properties in Section~\ref{ssec:gd} for the iterative refinement step of the high-fidelity algorithm.

% In particular, $\|A\|_{\rm op}=\max_j s_j$ and $\|A\|_F^2=\sum_j s_j^2$. 
% Writing $u_j$ for the columns of $U$, we will use the equivalent relations
% \[
%     A\bar u_j=s_j u_j,\qquad u_j^\dagger u_k=\delta_{jk},
%     \qquad U^\dagger A=SU^\top.
% \]
% The last identity turns a relation $d=Ad^\dagger$ into $c=Sc^\dagger$ under the change of modes $c=U^\dagger d$, where $d^\dagger$ and $c^\dagger$ denote columns of adjoints. These identities also apply when the operator relations hold only after acting on a fixed state.

\paragraph{Unsqueezing}
% \label{ssec:prelim-unsqueezing}

We state the deterministic Williamson estimate used in the iterative unsqueezing analysis of \cite{bittel2025energy}.

\begin{lemma}[{\cite[Lemma~12]{bittel2025energy}}]\label{lemma:williamson-unsqueezing}
Let $V$ be the covariance matrix of an $n$-mode Gaussian state, and let $\hat V$ be a real symmetric matrix satisfying
\begin{equation}\label{eq:williamson-sandwich}
    0\preceq\hat V-V\preceq\frac{2\zeta}{1-\zeta}(V+I)
\end{equation}
for some $\zeta\in(0,1)$. Take a Williamson decomposition
$\hat V=S\widehat D S^\top$, where $S$ is symplectic and
$\widehat D=\operatorname{diag}(d_1,\ldots,d_n,\\ d_1,\ldots,d_n)$ with $d_j>0$.
Then the Gaussian unitary associated with $S^{-1}$ transforms $V$ into
$V'=S^{-1}VS^{-\top}$ satisfying
\begin{equation}\label{eq:williamson-unsqueezing}
    \|(V')^{-1}\|_{\rm op}
       \leq\sqrt{1+\frac{2\zeta}{1-\zeta}
          \bigl(\|V^{-1}\|_{\rm op}+1\bigr)}.
\end{equation}
% In particular, for $\zeta=1/4$, the hypothesis is
% $0\preceq\hat V-V\preceq\frac23(V+I)$ and the bound becomes
% \[
%     \|(V')^{-1}\|_{\rm op}
%        \leq\sqrt{\frac53+\frac23\|V^{-1}\|_{\rm op}}.
% \]
\end{lemma}

\subsection{Primitives from robust statistics}
We will also require two primitives from the literature on robust statistics; we again refer the reader to~\cite{diakonikolas2023algorithmic} for a more detailed exposition on this topic.

\paragraph{Robustly learning a Gaussian}
Our first key primitive is the following guarantee for robustly learning a Gaussian, given samples from a distribution which is close in total variation distance to a Gaussian:
\begin{proposition}[Robustly Learning Gaussians, \citep{DKKLMS16,diakonikolas2023algorithmic}]\label{prop:robust-gaussian-learner}
    Let $P$ be a distribution over $\R^d$ and $Q=\cN(\mu,\Sigma)$ be a Gaussian  such that $\opt_{\rm TV}\equiv \tvd(P,Q)\leq \overline{\opt}_{\rm TV}$ for a sufficiently small universal constant $\overline{\opt}_{\rm TV}$. Then, there exists an algorithm which given $N\geq O(d^2 \log(1/\delta)/{\opt_{\rm TV}}^2)$ samples and $\poly(N)$ time outputs estimates $\hat \mu$ and $\hat \Sigma$, such that with probability at least $1-\delta$:
   \begin{equation}\label{eq:robust-gaussian-learner-relative-guarantee-F}
   \|\Sigma^{-1/2}(\hat \mu-\mu)\|_2\leq \zeta \quad {\rm and}\quad  \|\Sigma^{-1/2}\hat \Sigma \Sigma^{-1/2}-I\|_F\leq \zeta
   \end{equation}
   for $\zeta=\zeta({\opt_{\rm TV}})= O (\opt_{\rm TV} \log (1/\opt_{\rm TV}))$. This implies that $\tvd(\cN(\mu,\Sigma),\cN(\hat \mu,\hat \Sigma))\leq O (\opt_{\rm TV} \log (1/\opt_{\rm TV}))$.
\end{proposition}

\paragraph{List-decodable Gaussian estimation}
In the low fidelity regime, we will also need to learn non-trivial Gaussian approximations to classical distributions, even when the quality of the best Gaussian approximation is quite poor.
In such cases, the best Gaussian approximation is no longer unique, and rather, the standard procedure is to output a small list of candidate parameters so that one of them is close to the true Gaussian.
For us, we will require the results from \cite{ivkov2022list} for list-decodable mean and covariance estimation, restated below. 
\begin{lemma}[{\cite[Theorem~8.6]{ivkov2022list}}]\label{lemma:list-decodable-gaussian}
There are universal constants $C_{\rm dec},r_{\rm dec}>1$ such that the following holds. For $\beta\in(0,1]$ and $d\geq1$, set
\begin{equation}\label{eq:list-decodable-parameters}
    T_\beta=C_{\rm dec}(2/\beta)^{r_{\rm dec}},\qquad
    m_\beta(d)=\left\lceil e^{T_\beta}(d+1)^{T_\beta}\right\rceil,
    \qquad \ell_\beta=\lceil e^{T_\beta}\rceil.
\end{equation}
Given $\beta$ and $m_\beta(d)$ independent samples from any distribution $P$ on $\R^d$, a randomized algorithm with $\poly(m_\beta(d))$ time returns a list of at most $\ell_\beta$ pairs $(\hat\mu,\hat\Sigma)$ with $\hat\Sigma\succ0$. For every fixed Gaussian $Q=\c N(\mu,\Sigma)$ with $\Sigma\succ0$ and
\[
    \tvd(P,Q)\leq1-\beta,
\]
with probability at least $9/10$ the list contains a pair satisfying all three bounds
\begin{align}
    T_\beta^{-1}\Sigma&\preceq\hat\Sigma\preceq T_\beta\Sigma,
        \label{eq:list-decodable-spectral}\\
    \bigl\|\Sigma^{-1/2}(\hat\Sigma-\Sigma)\Sigma^{-1/2}\bigr\|_F
        &\leq T_\beta^2,\label{eq:list-decodable-frobenius}\\
    \bigl\|\Sigma^{-1/2}(\hat\mu-\mu)\bigr\|_2
        &\leq\sqrt{T_\beta(1+T_\beta)}.\label{eq:list-decodable-mean}
\end{align}
\end{lemma}

\section{Agnostic learning in the high-fidelity regime}\label{sec:hf}
We begin in the high-fidelity regime, where the given state, $\rho$, is sufficiently close to a Gaussian state such that $\opt \leq {\overline{\opt}}$ for some universal constant threshold ${\overline{\opt}}$. As an instructive exercise, we start by introducing the naive algorithm of heterodyne tomography and ask ourselves ``how can we make this robust?".

\subsection{Robust heterodyne tomography}
Heterodyne tomography, as the name suggests, applies a heterodyne measurement to $N$ copies of a Gaussian state $\pi$ with the hope of estimating parameters $({\bf m},V)$ using the resulting $N$ samples from $\c N ({\bf m},(V+I_{2n})/2)$. As a first pass, this seems to follow immediately from re-arranging the empirical mean, $\hat \mu_h$, and empirical covariance, $\hat \Sigma_h$, to construct $\hat{\bf m}=\hat \mu_h$ and $\hat V=2\hat \Sigma_h-I$. However, this may fail to correspond with a valid Gaussian state since the covariance error $\hat V-V$ is two-sided, causing a potential violation of the symplectic condition in \eqref{eq:V-psd}. Thus, we modify our covariance estimator to first scale the empirical covariance according to its relative error guarantee such that $\hat V-V$ is one-sided and the condition is satisfied. 

\begin{proposition}[Heterodyne Tomography, Theorem 10 of \citep{bittel2025energy}]
   Let $\{y_i\}_{i\in [N]}$ be the $N$ samples obtained from a heterodyne measurement on $\pi$, and let $\hat \mu_h,\hat \Sigma_h$ be the empirical mean and covariance respectively such that the latter has relative error $(1-\zeta)\Sigma_h\preceq \hat \Sigma_h\preceq (1+\zeta)\Sigma_h$. Construct the following estimators:
   \begin{equation}\label{eq:h-tomo-estimator}
       \hat{\bf m}\equiv \hat \mu_h ,\quad \hat{V}\equiv \frac{2\hat \Sigma_h}{1-\zeta }-I,\quad \hat\pi\equiv \pi(\hat {\bf m}, \hat V)
   \end{equation}
   For universal constant $C$, take $\zeta=C\sqrt{\frac{n+\log 1/\delta}{N}}$ such that $\zeta \leq 1/4$ for suitably large $N$. Then, with probability at least $1-\delta$, the relative error is at most
   \begin{equation}\label{eq:ht-complexity}
       \trd(\pi,\hat\pi)\lesssim (2n+\tr V^{-1})\sqrt{\frac{n+\log 1/\delta }{N}}\leq \epsilon \text{ for }N=O(n(n+\tr V^{-1})^2\epsilon^{-2}\log 1/\delta)
   \end{equation}
   or alternatively $N=O(nE(\pi)^2\epsilon^{-2}\log 1/\delta)$.
\end{proposition}

\noindent Unfortunately, heterodyne tomography is not robust. To illustrate this, suppose we receive copies of $\rho=(1-\eta )\ketbra{0}+\eta \ketbra{\xi}$ where $\ket{\xi}=(\hat{a}_1^\dagger)^\ell\ket{0}$ for some small $\eta \leq {\overline{\opt}}$ and positive integer $\ell$. Then, $\opt\leq {\overline{\opt}}$ since $\trd(\rho,\ketbra{0})\leq \eta$, but the outcome distribution of applying a heterodyne measurement to $\rho$ is $(1-\eta){\c N}(0,I)+\eta \c N (\ell e_1,I)$, whose mean $\eta \ell e_1$ explodes as $\ell\to\infty$. This failure comes from the well-known fact that the mean can be heavily skewed by outliers.

Luckily, this issue is precisely the origin of robust statistics, a developed literature in classical statistics focused on building estimators which can withstand small corruptions. These tools give us a clear path toward a robust version of heterodyne tomography: perform a heterodyne measurement, use off-the-shelf methods to obtain robust mean and covariance estimates, and then substitute these into \eqref{eq:h-tomo-estimator}. We can then obtain the following robust tomography algorithm.

\begin{theorem}[Robust Heterodyne Tomography]\label{thm:robust-heterodyne}
   Let $\rho$  be an $n$-mode state where $\trd(\rho,\pi)\leq {\opt_{\rm tr}}\leq {\overline{\opt}_{\rm tr}}$ for some Gaussian state $\pi$ with mean ${\bf m}$ and covariance $V$. Then, there exists an algorithm which given $\tilde O(n^2 \log^5(1/\delta)/{\opt_{\rm tr}}^2)$ copies of $\rho$ constructs a Gaussian state $\hat \pi $ with the semi-agnostic guarantee:
   \begin{equation}\label{eq:rht-bound}
           \trd(\rho, \hat \pi)\lesssim (n+\tr V^{-1})\cdot \Otilde( {\opt_{\rm tr}})\lesssim E(\pi)\cdot \Otilde ({\opt_{\rm tr}}).
   \end{equation}
\end{theorem}

\noindent We defer the proof to Appendix~\ref{app:proof-robust-heterodyne}.

Unfortunately, the semi-agnostic guarantee in \eqref{eq:rht-bound} is quite disappointing since the error scales with the energy of the closest Gaussian state, which is at least linear in the number of modes, $n$. Thus, Theorem~\ref{thm:robust-heterodyne} is generally only useful when we are okay with mode factors inflating the original corruption size. 

\subsection{Robust iterative unsqueezing}
While heterodyne tomography is simple, we incur a quadratic energy dependence in the sample complexity $O(nE^2\epsilon^{-2}\log 1/\delta)$, which is very poor for squeezed states. Fortunately, subsequent work by \citep{bittel2025energy} reduced the energy dependence to $\log \log E$ through a protocol which we call \textit{iterative unsqueezing}. This involves a series of applications of heterodyne tomography to generate a sequence of Gaussian operations that unsqueeze the state, after which the resulting state is estimated with heterodyne tomography and then rotated back through the corresponding sequence of inverse Gaussian operations to derive the final estimate. Specifically, let $\pi_0\equiv \pi $ be the given Gaussian state, and conduct the following procedure for $K=\log \log (\|V^{-1}\|_{\rm op})$ total rounds: \alvantodo{appendix: measure $E(\pi_{k+1})$ and stop when small, so we don't need to depend on $\|V^{-1}\|_{\op}$}
\begin{enumerate}
    \item Perform heterodyne tomography on $\pi_{k}$ with roughly $O(n)$ copies and failure probability $\delta/K$ to arrive at estimates $\hat \pi_k=\pi(\hat {\bf m}_k,\hat V_k)$.
    \item Compute the Williamson decomposition of the covariance estimate $ \hat V_k=\hat S_k\hat D_k\hat S^\top_k $. 
    \item Form $\pi_{k+1} =U_{\hat S_k^{-1}}\rho U_{\hat S_k^{-1}}^\dagger$ which has the same mean as $\rho$ but covariance $S^{-1}V S^{-1\top}$. 
    \item Repeat the previous steps for $k\leftarrow k+1$
\end{enumerate}
The intuition is that each step unsqueezes the incoming state until $\pi_K$ has a covariance matrix whose inverse is bounded: $\|V^{-1}_K\|_{\rm op}\leq 2$. Since no particular direction is exponentially small due to squeezing, this tames the $\tr V^{-1}$ contribution in \eqref{eq:ht-complexity}, meaning we can perform heterodyne tomography on $\pi_K$ with $O(n^3 \epsilon^{-2}\log 1/\delta)$ copies:
\begin{enumerate}[start=4]
    \item Perform heterodyne tomography on $\pi_K $ with roughly $O(n^3/\epsilon^2)$ copies to obtain $\hat \pi_{K} \equiv \pi(\hat {\bf m}_{K} ,\hat V_{K})$. 
    \item Form $\hat \pi=U_{\hat S_1^{-1}}^\dagger \cdots U_{\hat S_K^{-1}}^\dagger\hat \pi_{K}  U_{\hat S_K^{-1}}\cdots U_{\hat S_1^{-1}}$ as the final estimate. 
\end{enumerate}
  The only energy dependence in the sample complexity comes from the number of iterations $K=\log \log \|V^{-1}\|_{\rm op}\leq \log \log E(\pi)$, with $O(n)$ copies used for the heterodyne tomography of step 1 in each round. This gives the following net copy complexity. 

\begin{theorem}[``Energy-Independent" Tomography via Iterative Unsqueezing, Theorem 3 \citep{bittel2025energy}]\label{thm:eit}
   Let $\pi$ be a Gaussian state with unknown covariance $V$. Then, there exists an algorithm which takes 
    $$N=O(((n+\log\log\log \|V^{-1}\|_{\rm op})\log \log \|V^{-1}\|_{\rm op}+n^3\epsilon^{-2})\log 1/\delta)$$
    copies of $\pi$ and uses exclusively Gaussian measurements to ouptut a Gaussian state $\hat \pi$ such that $\trd(\pi,\hat \pi)\leq \epsilon$ with probability at least $1-\delta$.
\end{theorem}

\noindent Again, we ask if this ``energy-independent" Gaussian state tomography algorithm can be robustified. The natural first suggestion is to run the algorithm on the given copies of $\rho$, constructing $\hat \rho_k, \hat{\bf m}_k, \hat V_k, \hat S_k, \hat D_k$ accordingly for $k\in [K]$, by replacing every occurrence of heterodyne tomography in the algorithm with the robust primitive in Theorem~\ref{thm:robust-heterodyne}. At first, this seems to immediately fail because the error of robust heterodyne tomography in step 1 for the first iteration blows up by a linear factor in $E(\pi_0)$, which may be arbitrarily large. Fortunately, this is not an issue because heterodyne tomography is only used for a covariance estimate with constant relative error, $\zeta \leq 1/4$, during the iterative unsqueezing procedure. In fact, this is why each iteration uses only $O(n)$ copies as opposed $O(nE(\pi)^2 \epsilon^{-2}\log K/\delta)$. Thus, by substituting the robust heterodyne tomography primitive, we can make this ``energy-independent" protocol robust. 

\begin{theorem}[Robust ``Energy-Independent" Tomography via Iterative Unsqueezing]\label{thm:reit}
   Let $\rho$ be a $n$-mode state where $\trd(\rho,\pi)\leq {\opt_{\rm tr}}\leq {\overline{\opt}_{\rm tr}}$ for Gaussian state $\pi$ with covariance $V$. Then, there exists an algorithm which takes $O(n^3{\opt_{\rm tr}}^{-2}\log\log\|V^{-1}\|_{\rm op}\log^5(1/\delta))$ copies of $\rho$ and uses exclusively Gaussian measurements to output a Gaussian state $\hat \pi$ satisfying
   \begin{equation}\label{eq:reit-bound}
       \trd(\rho,\hat \pi)\leq n\cdot \Otilde({\opt_{\rm tr}}).
   \end{equation}
\end{theorem}
\noindent We defer the proof to Appendix~\ref{app:proof-reit}.

While this robust iterative unsqueezing successfully removes the $\tr V^{-1}$ factor in \eqref{eq:rht-bound}, the linear mode factor in \eqref{eq:reit-bound} still gives only a semi-agnostic guarantee for arbitrary Gaussian states. 

\subsection{Robust pure state projection}

\sitan{I'm not sure the mixed state results in the preceding subsections are that worth spelling out in gory detail. We can certainly mention the robust statistics setup in Proposition 3.2, and we can even devote some appendix to stating and proving Theorem 3.3 and Theorem 3.4, but probably we want to work up to Theorem 3.7 as efficiently as possible. otherwise there are too many ``warm up'' results before we get to the heart of this section, namely theorem 3.8}

 Now, we do not expect to do better for mixed Gaussian states since the trace distance between such states is only bounded by parameter distances up to a necessary $\sqrt{n}$ factor \citep{chen2026towards}. \alvantodo{prove tight?} However, this mode factor disappears for pure Gaussian states, enabling \citep{chen2026towards} to optimally resolve tomography for this class; the algorithm is identical to Theorem~\ref{thm:eit}, except the final unsqueezed state is projected to pure by taking the raw, uninflated covariance estimate and replacing all symplectic eigenvalues with $1$. Specifically, if $\hat{\bf m}_K$ and $\hat V_{K}=\frac{2\hat \Sigma_{h}}{1-\zeta}-I$ are the estimates from step 5 of Theorem~\ref{thm:eit}, we take the Williamson decomposition of $\hat V_{\textrm{raw}}\equiv 2\hat \Sigma_{h}-I=\hat S \hat D\hat S^\top $ and form $\pi(\hat {\bf m},\hat V)$ where $\hat {\bf m}\equiv \hat {\bf m}_K$ and $\hat V\equiv \hat S \hat S^\top $, before rotating back in step 6. This gives the following guarantee. 

 \begin{theorem}[Near-Optimal Tomography for Pure Gaussian States, Theorem S7 \citep{chen2026towards}]\label{thm:near-optimal-pgs-tomo}
     Let $\ket{\psi}$ be a pure Gaussian state with covariance $V$. Then, there exists an algorithm which takes $O(n^2/\epsilon^2 +(n+\log\log\log\|V\|_{\rm op})\log \log \|V\|_{\rm op})$ copies and uses exclusively Gaussian measurements to output a pure Gaussian state $\ket*{\hat \psi}$ such that, with probability $2/3$, $\trd(\ketbra{\psi},\ketbra*{\hat \psi})\leq \epsilon$.
 \end{theorem}

\noindent Again, we consider what this means for the robust setting: does appending pure state projection to Theorem~\ref{thm:reit} remove the $n$ factor in \eqref{eq:reit-bound} to give a truly semi-agnostic guarantee? We respond in the affirmative with the following result.

\begin{theorem}\label{thm:riuproj}
 Let $\rho$ be an $n$-mode state where $\trd(\rho,\ketbra{\psi})\leq {\opt_{\rm tr}}\leq {\overline{\opt}_{\rm tr}}$ for pure Gaussian state $\ket{\psi}$ with covariance $V$. Then, given ${\epsilon>0}$ and $\delta\in(0,1/2)$, there exists an algorithm which takes $\Otilde(n^2 \epsilon^{-2}\log \log \|V^{-1}\|_{\rm op}\allowbreak\log^5(1/\delta))$ copies of $\rho$ and uses exclusively Gaussian measurements to output a pure Gaussian state $\ket*{\hat \psi}$ satisfying
 \begin{equation}\label{eq:riuproj-bound}
     \trd(\rho,\ketbra*{\hat \psi})\leq \Otilde({\opt_{\rm tr}}) +\epsilon.
 \end{equation}
 with probability at least $1-\delta$.
\end{theorem}

\noindent From Theorem~\ref{thm:reit}, we have that $\|V^{-1}_K\|_{\rm op}\leq 2$. Since trace distance is unitarily invariant, WLOG it suffices to just assume that $\|V^{-1}\|_{\rm op}\leq 2$ already and prove the guarantee using notation for the original state. We extract the bulk of this proof into the following lemma, whose proof closely follows that of Theorem S7 in \citep{chen2026towards}, with additional care to remove the strong operator norm assumption in (S107) of \citep{chen2026towards} and instead control the relative Frobenius norm of the estimates.

\begin{lemma}\label{lemma:frob-h-to-trd}
    Suppose $\|V^{-1}\|_\op \leq 2$ and that there exist estimators $(\hat \mu_h,\hat \Sigma_h)$ which satisfy 
    \begin{equation}\label{eq:frob-h-to-trd-hypoth}
        \|\hat \mu_h-\mu_h\|_2\leq \zeta,\quad \|\hat \Sigma_h -\Sigma_h\|_F \leq \zeta 
    \end{equation}
    where $\mu_h={\bf m}$ and $\Sigma_h=\frac{I+V}{2}$ denote the parameters of the heterodyne outcome distribution for pure Gaussian state $\ketbra{\psi}=\pi({\bf m},V)$. Define the estimate $\ketbra*{\hat \psi}\equiv \pi(\hat \mu_h,\hat V)$ where $\hat V=\hat S\hat S^\top$ from the decomposition of $\hat V_{\rm raw}\equiv 2\hat \Sigma_h-I=\hat  S \hat D\hat S^\top $. Then, 
    \begin{equation}
        \trd(\ketbra*{\hat \psi},\ketbra{\psi})\lesssim  \zeta 
    \end{equation}
    provided that $\zeta$ is sufficiently small. 
\end{lemma}
\noindent We defer the proof to Appendix~\ref{app:proof-frob-h-to-trd}.

\begin{proof}[Proof of Theorem~\ref{thm:riuproj}]
Run robust heterodyne tomography to obtain $(\hat \mu_h,\hat \Sigma_h)$ which, by Proposition~\ref{prop:robust-gaussian-learner}, satisfies
\begin{equation}\label{eq:riuproj-simgah-relf}
    \|\Sigma_h^{-1/2}(\hat \mu_h-\mu_h)\|_2\leq \zeta \quad {\rm and}\quad \|\Sigma_h^{-1/2} \hat \Sigma_h \Sigma_h^{-1/2}-I\|_F\leq \zeta
\end{equation}
for $\zeta=\Otilde({\opt_{\rm tr}})$ that satisfies $\zeta \leq 1/4$. Since $\frac 1 2 I\preceq V\preceq 2I$, $\frac 3 4 I\preceq \Sigma_h\preceq \frac 3 2 I $, meaning \eqref{eq:riuproj-simgah-relf} implies $\|\hat \mu_h-\mu_h\|_2\leq \sqrt{\frac 3 2}\zeta$ and $\|\hat \Sigma_h-\Sigma_h\|_F \leq \frac 3 2 \zeta $. 
Then, by triangle inequality,
\begin{equation}
\trd(\rho,\ketbra*{\hat\psi})\leq \trd(\rho,\ketbra{\psi})+\trd(\ketbra{\psi},\ketbra*{\hat\psi})\leq {\opt_{\rm tr}}+O(\zeta)=\Otilde({\opt_{\rm tr}})
\end{equation}
where we use the absolute guarantees to invoke Lemma~\ref{lemma:frob-h-to-trd} and bound the second term. 
\end{proof}

\subsection{Robust iterative refinement via gradient descent}\label{ssec:gd}
While the semi-agnostic guarantee in \eqref{eq:riuproj-bound} does not contain mode or energy factors, it does lose logarithmic factors in ${\opt_{\rm tr}}$. This loss traces itself through the robust heterodyne tomography primitive back to the robust learner of Proposition~\ref{prop:robust-gaussian-learner} itself. As a result, any algorithm which simply applies robust heterodyne tomography should not expect to break the $\Otilde({\opt_{\rm tr}})$ barrier.

We circumvent this nevertheless by using non-Gaussian measurements to iteratively refine our estimate in Theorem~\ref{thm:riuproj} to reach not just $O(\opt)$ but the optimal scaling of $\opt +\epsilon$. We now state our main result.

\begin{theorem}\label{thm:hf}
   Let $\rho$ be an arbitrary $n$-mode bosonic state and $\ket{\psi}$ be a pure Gaussian state with covariance $V$ such that $\infd(\rho,\ketbra{\psi})\leq \opt\leq \overline{\opt}$ for universal constant $\overline{\opt}$. Then, for given $\epsilon > 0$ and $\delta\in(0,1/2)$, there exists an algorithm which takes
   \[
       \Otilde\left(n^2(\log\log \|V^{-1}\|_{\rm op} \log^5(1/\delta)+\epsilon^{-1})\right)
   \]\alvantodo{probably state copy complexity better}
   copies and outputs a pure Gaussian state $\ket*{\hat \psi}$ satisfying
   \begin{equation}\label{thm:hf-bound}
      \infd(\rho,\ketbra*{\hat \psi})\leq \opt + \epsilon
   \end{equation}
   with probability at least $1-\delta$. 
\end{theorem}

\noindent Our algorithm begins by running the robust learner from Theorem~\ref{thm:riuproj} to obtain the pure Gaussian state $\ket*{\hat \psi_0}$. However, since Theorem~\ref{thm:riuproj} requires closeness in trace distance despite Theorem~\ref{thm:hf} assuming closeness in infidelity, we must apply the Fuchs-van de Graaf inequality. Specifically, $\trd(\rho,\ketbra{\psi})\leq \sqrt{1-F(\rho,\ketbra{\psi})}\leq \sqrt \opt $, so the guarantee of Theorem~\ref{thm:riuproj} applies with ${\opt_{\rm tr}=\sqrt{\opt}}$, meaning $\trd(\rho,\ketbra*{\hat \psi_0})\leq \gamma $ where ${\gamma=\Otilde(\sqrt{\opt})+\epsilon_0}$. Re-writing in infidelity, 
\begin{align}\label{eq:tr-to-F-transfer}
    \sqrt{1-\bra*{\hat \psi_0}\rho\ket*{\hat \psi_0}}=\|(I-\rho)^{1/2}\ket*{\hat \psi_0}\|_2 &\leq \|(I-\rho)^{1/2}\ket{\psi}\|_2+\|(I-\rho)^{1/2}(\ket*{\hat \psi_0}-\ket{\psi})\|_2\\
    &\leq \sqrt{1-\bra \psi \rho\ket \psi}+\|{\ket*{\hat\psi_0}}-\ket{\psi}\|_2\\
    &\leq \sqrt{{\opt}}+\sqrt 2\, \trd(\ketbra\psi ,\ketbra*{\hat \psi_0})\\
    &\leq \sqrt {\opt} +\sqrt 2 (\sqrt {\opt}+\gamma)\\
    \infd(\rho,\ketbra*{\hat \psi_0})&\leq {\Otilde(\opt)+O(\epsilon_0^2)}
\end{align}
where we use triangle inequality and the facts that $I-\rho\preceq I$ and $\|\ket*{\hat \psi_0}-\ket{\psi}\|_2\leq \sqrt 2 \trd(\ketbra \psi,\ketbra*{\hat \psi_0})$.\sitan{These steps are sufficiently standard that we should move this part to a section in the prelims} 

If $U_0$ is the Gaussian unitary satisfying $U_0^\dagger \ket*{\hat \psi_0}=\ket 0$, by the unitary invariance of trace distance and infidelity, we can apply $U_0^\dagger$ to any $\rho$ we use in this stage, run our refinement algorithm, and then perform \(U_{0}\hat\pi U_{0}^\dagger\) to rotate back. WLOG, it then suffices to ignore the first learner and simply assume that the given state is near-vacuum:
\begin{equation}\label{eq:vacc-close}
    \infd(\rho,\ketbra{0})\leq \Otilde(\opt)+O(\epsilon_0^2),\quad \trd(\rho,\ketbra{0})\leq \gamma=\Otilde(\sqrt{\opt})+\epsilon_0
\end{equation}
\sitan{based on Proposition~\ref{prop:hf-gd}, do we actually use the former estimate, or only the trace distance closeness?}
This allows us to leverage properties of near-vacuum Gaussian states, such as Lemma~\ref{lemma:barg-param-L2-small}, to study the geometry of the {infidelity} objective 
\begin{equation}
    f_\rho(\theta)=1-\bra{\psi(\theta)}\rho\ket{\psi(\theta)}    
\end{equation}
as a function of the Bargmann parameters expressed as reals: $\theta=(b,A)\in \mathcal P$. For $\ket{\psi(\theta)}$ near vacuum, i.e. $\theta \in B_r=\{\theta\in\mathcal P:\|\theta\|_{\mathcal P}\leq r\}$ for sufficiently small $r$,  we find that $f_\rho(\theta)$ is smooth and strongly convex. Thus, our algorithm proceeds by performing gradient descent on $f_\rho(\theta)$; we will use measurements to simulate noisy access to the objective $f_\rho(\theta)$, which we then use to simulate access to the gradient via finite differencing.

\subsubsection{Near-vacuum geometry of the infidelity}
We find that the Hessian of $f_\rho(\theta)$ can be controlled by the trace distance to the vacuum state, provided this quantity sufficiently small.  

\begin{proposition}\label{prop:f-smooth-convex}
There exist universal constants $r_0,C>0$ such that for every $n$, every
density operator $\rho$, and every $\theta\in B_{r_0}$,
\begin{equation}\label{eq:hessian-stability}
 \|\nabla^2f_\rho(\theta)-2I\|_{\rm op}
 \le C\bigl(\|\theta\|_{\c P}+\trd(\rho,\ketbra 0)\bigr).
\end{equation}
where the operator norm is induced by $\|\cdot\|_{\mathcal P}$.
In particular, there are universal constants $\gamma_0>0$ and $R\in(0,1/16)$ such that, whenever
$\trd(\rho,\ketbra 0)\le\gamma_0$, the infidelity objective satisfies
\begin{equation}\label{eq:convex-smooth}
 I\preceq\nabla^2f_\rho(\theta)\preceq3I
 \qquad (\theta\in B_R),
\end{equation}
and its unique minimizer $\theta_*$ over $B_R$ satisfies $\|\theta_*\|_{\c P}\leq R/4$.
\end{proposition}
\noindent We make headway toward \eqref{eq:hessian-stability} by first controlling higher-order derivative tensors of the state $\ket{\psi(\theta)}$ near-vacuum. We start by understanding how the operator in the exponent of \eqref{eq:barg-state}, $K_\theta$, inflates the norm of various states in Fock space. 
\begin{lemma}\label{eq:Ktheta-growth}
    Let $\Pi_{\leq m}$ be the projector onto states with total photon number at most $m$. For $\ket{\phi}\in \c F_{\leq m}\equiv \Pi_{\leq m}\c F$, we have that $\|K_\theta \ket{\phi}\|_2\leq 2(m+2)\|\theta \|_{\mathcal P}\|\ket{\phi}\|_2$. 
\end{lemma}
\begin{proof}
    WLOG, we perform a unitary change of modes such that $b=\|b\|_2e_1$. Then, 
    \begin{equation}\label{eq:Lb-growth}
        \|L_b\phi\|_2^2=\|b\|_2^2 \bra{\phi}\hat{a}_1\hat{a}_1^\dagger \ket{\phi}=\|b\|_2^2 \bra{\phi}(\hat{a}_1^\dagger \hat{a}_1+I)\ket{\phi}\leq (m+1)\|b\|_2^2\|\ket{\phi}\|_2^2
    \end{equation}
    Similarly, we claim that:
    \begin{align}\label{eq:QA-growth}
        \|Q_A\ket{\phi}\|_2\leq \frac 1 2 \|A\|_F\sqrt{(m+1)(m+2)}\|\ket{\phi}\|_2
    \end{align}
    Since every term in $Q_A$ has two creation operators and $\ket{\phi}\in \c F_{\leq m}$, we have that $Q_A\ket{\phi}\in \c F_{\leq m+2}$. Considering the variational characterization of the left hand side for any normalized $\ket{\varphi}\in \c F_{\leq m+2}$, 
    \begin{align}
        |\bra{\varphi}Q_A\ket{\phi}|=\frac 1 2 \left|\sum_{i,j=1}^n A_{ij}\bra{\varphi}\hat{a}_i^\dagger \hat{a}_j^\dagger \ket{\phi}\right|&\leq \frac 1 2 \|A\|_F\left(\sum_{i,j=1}^n|\bra{\varphi}\hat{a}_i^\dagger \hat{a}_j^\dagger \ket{\phi}|^2\right)^{1/2}\\
        &\leq \frac 1 2 \|A\|_F \|\ket{\phi}\|_2 \left(\sum_{i,j=1}^n \|\hat{a}_j\hat{a}_i \ket{\varphi}\|_2^2\right)^{1/2}\label{eq:QA-growth-1}
    \end{align}
    where we apply Cauchy-Schwarz twice. Since the ladder operators for different modes commute, for $i\neq j$, $\hat{a}_i^\dagger \hat{a}_j^\dagger \hat{a}_j \hat{a}_i=N_i N_j$ where $N_i=\hat{a}_i^\dagger \hat{a}_i $ is the number operator for the $i$-th mode, while for $i=j$, $\hat{a}_i^\dagger \hat{a}_j^\dagger \hat{a}_j \hat{a}_i=\hat{a}_i^\dagger (\hat{a}_i\hat{a}_i^\dagger-1)\hat{a}_i=N_i(N_i-1)$. Then, 
    \begin{align}
    \sum_{i,j=1}^n\|\hat{a}_j\hat{a}_i\ket{\varphi}\|_2^2 =\bra{\varphi}\left(\sum_{i,j=1}^na_i^\dagger \hat{a}_j^\dagger \hat{a}_j \hat{a}_i \right)\ket{\varphi }=\bra{\varphi}\left(\sum_{i\neq j}^n N_iN_j+\sum_{i=1}^n N_i(N_i-1)\right)\ket{\varphi}&=\bra{\varphi}N(N-1)\ket{\varphi}\\
    &\leq (m+2)(m+1)
    \end{align}
    where $N=\sum_{i=1}^n N_i$ is the total number operator and where we recall that $\ket{\varphi}\in \c F_{\leq m+2}$ in the last equality. Substituting this into \eqref{eq:QA-growth-1} then proves \eqref{eq:QA-growth}. Thus, by triangle inequality,
    \begin{align}
        \|K_\theta \ket{\phi}\|_2\leq \|L_b \ket{\phi}\|_2+\|Q_A\ket{\phi}\|_2&\leq \left(\|b\|_2\sqrt{m+1}+\frac 1 2 \|A\|_F\sqrt{(m+1)(m+2)}\right)\|\ket{\phi}\|_2\\
        &\leq 2(m+2)\|\theta\|_{ \c P}\|\ket{\phi}\|_2 \label{eq:L_b+Q_A-0}
    \end{align}
    which proves the bound. 
\end{proof}

\noindent Control on the norm growth that $K_\theta$ induces allows us to bound the magnitudes of higher-order derivatives of $\ket{\psi(\theta)}$ by carefully differentiating and tracking the total photon number of each intermediate state. 

\begin{lemma}\label{lemma:barg-derivatives}
There is a universal constant $C_D$ such that, for every $k\geq 1$, $\|\theta\|_{\c P}\leq 1/8$, and directions $h_1,\dots,h_k\in\mathcal P$,
\begin{equation}\label{eq:barg-derivatives}
    \|\ket{D^k\psi(\theta)[h_1,\dots,h_k]}\|_2
    \leq C_D^k k!\prod_{i=1}^k\|h_i\|_{\mathcal P}\,.
\end{equation}
% where $D^k$ denotes real Fr\'echet derivatives. 
\end{lemma}
\noindent We defer the proof to Appendix~\ref{app:proof-barg-derivatives}.

This higher-order derivative control on $\ket{\psi(\theta)}$ is sufficient to prove Proposition~\ref{prop:f-smooth-convex}. 

\begin{proof}[Proof of Proposition~\ref{prop:f-smooth-convex}]
Let $r_0=1/8$, and for brevity, denote $g_h(\theta)\equiv Dg(\theta)[h]$ and $g_{h_1h_2}(\theta)\equiv D^2 g(\theta)[h_1,h_2]$, where $g$ may be $Z, \Phi,$ or $\psi$. Further, let $C_j\equiv C_D^j j! $. In broad strokes, we first compute the Hessian of $f_{\ketbra{0}}(\theta)$ at $\theta=0$. Then, we use third derivative control to control the Hessian of $f_{\ketbra{0}}(\theta) $ for $\|\theta \|_{\c P}\leq r_0$, and finally, we use second derivative control to handle the Hessian of $f_{\rho}(\theta)$. 

We begin with the Hessian for $\rho=\ketbra{0}$ at $\theta=0$. Since $K_\theta$ is linear in $\theta$ and all creation operators commute, $e^{K_{\theta+th}}=e^{K_\theta}e^{tK_h}$, so $D\Phi(\theta)[h]=K_he^{K_\theta}\ket0$. With \eqref{eq:barg-state}, this gives $\ket{\Phi(0)}=\ket 0$ and $D\Phi(0)[h]=K_h\ket{0}$, which is orthogonal to $\ket 0$. Then, by product rule, 
\begin{align}
    Z_h(0)&=2\Re\braket{\Phi(0)}{\Phi_h(0)}=0\\
    \ket{\psi_h(0)}&=\cancel{-\frac 1 2 Z(0)^{-3/2}Z_h(0)\ket{\Phi(0)}}+Z(0)^{-1/2}\ket{\Phi_h(0)}=K_h\ket 0
\end{align}
Applying product rule again, 
\begin{align}
   f_{\ketbra 0}{(0)} &=0\\
   Df_{\ketbra 0}(0)[h]&=-2\Re\braket{\psi(0)}{0}\cancel{\braket{0}{\psi_h(0)}}=0
   \label{eq:D1-vac-theta0}
   \\
   D^2  f_{\ketbra 0}(0)[h,h]&=-2\Re\braket{\psi(0)}{0}\braket{0}{\psi_{hh}(0)}-\cancel{2\|\braket{0}{\psi_h(0)}\|^2_2}=2\|\ket{\psi_h(0)}\|_2^2=2\|K_h\ket 0\|^2_2 = 2\|h\|_{\c P }^2 \label{eq:D2-vac-theta0}
\end{align}
where the second equality in the last line uses the following manipulation of the normalization $\|\ket{\psi(\theta)}\|_2=1$: 
\begin{align}
0=D^2\|\ket{\psi(\theta)}\|_2^2[h,h]=2\Re\braket{\psi(\theta)}{\psi_{hh}(\theta)}+2\|\ket{\psi_h(\theta)}\|_2^2
\end{align}

With the vacuum Hessian at $\theta=0$ computed, we continue to control the change in the Hessian for $\|\theta\|_{\c P }\leq r_0$ and then for $\rho$ beyond the vacuum state. Writing $s^{(H)}(\theta)=\bra{\psi(\theta)}H\ket{\psi(\theta)}$ for any bounded self-adjoint $H$, product rule yields:
\begin{align}
   s^{(H)}_{hh} (\theta)&=2\Re\bra{ \psi_{hh}(\theta)}H\ket{\psi(\theta)}+2\bra{\psi_h(\theta)}H\ket{\psi_h(\theta)}\\
   s^{(H)}_{uhh}(\theta)&=2\Re\bra{\psi_{uhh}(\theta)}H\ket{\psi(\theta)}+2\Re\bra{\psi_{hh}(\theta)}H\ket{\psi_u(\theta)}+4\Re\bra{\psi_{uh}(\theta)}H\ket{\psi_h(\theta)}
\end{align}
Since $\|\ket{\psi(\theta)}\|_2=1$ and Lemma~\ref{lemma:barg-derivatives} holds for $\|\theta\|_{\c P}\leq 1/8=r_0$, by triangle inequality and these bounds,
\begin{align}
    {|s^{(H)}_{hh}(\theta)|}&\leq2(C_2+C_1^2)\|H\|_{\rm op}\|h\|_{\mathcal P}^2
    \label{eq:D2-sH-hh-size}
    \\
    {|s^{(H)}_{uhh}(\theta)|}&\leq(2C_3+6C_1C_2)\|H\|_{\rm op}\|u\|_{\mathcal P}\|h\|_{\mathcal P}^2\label{eq:D3-sH-uhh-size}
\end{align}
Then, by the fundamental theorem of calculus, \eqref{eq:D2-vac-theta0}, and \eqref{eq:D3-sH-uhh-size},
\begin{align}
    |D^2 f_{\ketbra 0}(\theta)[h,h]-2\|h\|_{\c P}^2|=\left|\int_0^1D^3 f_{\ketbra 0}(t\theta)[\theta,h,h]dt\right|\leq (2C_3+6C_1 C_2)\|\theta\|_{\c P }\|h\|_{\c P }^2 \label{eq:D2-vac-0-to-vac-theta}
\end{align}
where we notice $f_{\ketbra{0}}(\theta)=1-s^{(H)}(\theta)$ for $H=\ketbra 0$ which has $\|H\|_{\rm op}=1$. Now all that is left is controlling the Hessian when $\Delta \equiv \rho-\ketbra 0\neq 0$. 
\begin{align}
    |D^2 f_\rho(\theta)[h,h]-D^2f_{\ketbra 0 }(\theta)[h,h]|=|s_{hh}^{(\Delta)}(\theta)|\leq 2(C_2+C_1^2)\|\Delta \|_{\rm op}\|h\|_{\c P }^2\label{eq:D2-vac-th-to-rho-th}
\end{align}
Noting that $\|\Delta \|_{\rm op}\leq 2\trd(\rho,\ketbra{0})$, by triangle inequality, \eqref{eq:D2-vac-0-to-vac-theta}, \eqref{eq:D2-vac-th-to-rho-th},
\begin{align}
    |D^2 f_\rho(\theta)[h,h]-2\|h\|_{\c P}^2|&\leq (2C_3+6C_1 C_2)\|\theta\|_{\c P }\|h\|_{\c P }^2+2(C_2+C_1^2)\|\Delta \|_{\rm op}\|h\|_{\c P }^2\\
    &\leq C(\|\theta\|_{\c P}+\trd(\rho,\ketbra 0))\|h\|_{\c P}^2
\end{align}
where we choose $C$ according to $C_1,C_2,C_3$. Since this holds for all unit $h$, we have proven \eqref{eq:hessian-stability}. Then, taking
\[
    R=\min\{r_0/4,1/(2C)\},\qquad
    \gamma_0=\min\{1/(2C),R^2/64\},
\]
we have $\|\nabla^2f_\rho(\theta)-2I\|_{\rm op}\leq C(R+\gamma_0)\leq1$ for $\theta\in B_R$, proving \eqref{eq:convex-smooth}. Further, $\|A\|_{\rm op}\leq\sqrt2\|\theta\|_{\c P}\leq\sqrt2R<1$ for $R<1/16$, so every point of $B_R$ describes a valid pure Gaussian state.

To show that the last part about the minimizer $\theta_*$ over $B_R$, we use the fundamental theorem of calculus twice and recall that $f_{\ketbra{0}}(0)=Df_{\ketbra{0}}(0)[\theta]=0$ from \eqref{eq:D1-vac-theta0}: 
\begin{align}
f_{\ketbra 0}(\theta)=\int_0^1 \int_0^s D^2 f_{\ketbra 0}(t\theta)[\theta,\theta] dt ds=\int_0^1 (1-t)D^2 f_{\ketbra 0}(t\theta)[\theta,\theta] dt\geq \int_0^1 (1-t)\|\theta\|_{\c P}^2 dt \geq \frac 1 2 \|\theta\|_{\c P}^2 
\end{align}
where we use the lower bound of \eqref{eq:convex-smooth}. Then, for $\theta$ on the boundary of $B_R$, by triangle inequality, ${\gamma_0\leq R^2/64}$, and $\|\Delta \|_{\rm op}\leq 2\gamma_0$,
\begin{align}
    f_\rho(\theta)\geq f_{\ketbra 0}(\theta)-|f_\rho(\theta)-f_{\ketbra 0}(\theta)|\geq R^2/2-\|\Delta\|_{\rm op}\geq R^2/2-2\gamma_0\geq 3R^2/8{>}2\gamma_0\geq f_\rho(0)
\end{align}
where we use the fact that $f_\rho(0)=|\bra{0}(\ketbra 0 -\rho)\ket 0 |\leq \|\Delta \|_{\rm op}\leq 2\gamma_0$. Then, by continuity of $f_\rho$, compactness of $B_R$, and the fact that the boundary is strictly worse than the origin, it follows that all minimizers are in the interior. Letting one minimizer be $\theta_*$ such that $\nabla f_\rho(\theta_*)=0$, by the fundamental theorem of calculus and \eqref{eq:convex-smooth} in the same manner as before,
\begin{align}
    f_\rho(\theta)\geq f_\rho(\theta_*)+\cancel{Df_\rho(\theta_*)[\theta-\theta_*]}+\int_0^1 (1-t)D^2 f_{\rho}(\theta_* (1-t)+t\theta)[\theta-\theta_*,\theta-\theta_*]dt \geq f_\rho(\theta_*)+\frac 1 2 \|\theta-\theta_*\|_{\c P}^2
\end{align}
which implies that any $\theta \neq \theta_* $ has $f_\rho(\theta)>f_\rho(\theta_*)$, meaning the minimizer is unique. Finally, taking $\theta=0$, $\frac 1 2 \|\theta_*\|_{\c P}^2 \leq f_\rho(0)\leq 2\gamma_0\leq R^2/32$ proves the $\|\theta_*\|_{\c P}\leq R/4$ claim. 
\end{proof}

\subsubsection{Noisy objective estimate}
Given that the objective function $f_\rho(\theta) $ is both smooth and strongly convex from Proposition~\ref{prop:f-smooth-convex}, all we need for gradient descent is to give a protocol for estimating both $f_\rho(\theta)$ and $\nabla f_\rho(\theta)$ given $\theta$ and copies of $\rho$. We do this with a single Gaussian unitary followed by an on-off photon detector as follows. \sitan{might as well just directly say here and in the proof that we measure with $\{\ketbra{\psi(\theta)}, I - \ketbra{\psi(\theta)}\}$, without referencing $U_\theta$}

\begin{lemma}\label{lemma:hf-value-oracle}
For $\theta\in B_R$, we can estimate $f_\rho(\theta)$ to additive error $\kappa$ with probability at least $1-\delta$ using $O(\kappa^{-2}\log(2/\delta))$ copies of $\rho$.
\end{lemma}
\begin{proof}
    \sitan{eventually move to appendix}
Let $U_\theta$ be the Gaussian unitary which prepares $\ket{\psi(\theta)}=U_\theta \ket 0$. For each copy, we apply $U_\theta^\dagger$ and then measure the two-outcome POVM $\{\ketbra 0,I-\ketbra 0\}$. The latter outcome has the probability:
\begin{align}
    p_{I-\ketbra 0} =\tr \left((I-\ketbra 0)U_\theta^\dagger \rho U_\theta\right)=\infd(\rho,\ketbra{\psi(\theta)})
\end{align}
which is exactly the estimand. Thus, an empirical average with $O(\kappa^{-2}\log(2/\delta))$ copies is sufficient by Hoeffding's inequality.
\end{proof}

\noindent We use this estimation protocol to construct a zeroth-order estimate for $\nabla f_\rho(\theta)$ via finite differencing.
\begin{lemma}\label{lemma:hf-gradient-error}
   If $\trd(\rho,\ketbra 0)\leq \gamma_0$, $\theta \in B_{R/2}$, and $\kappa\in (0,R/4]$, we can estimate $\nabla f_\rho(\theta)$ to additive error $\kappa$ in $\|\cdot \|_{\c P}$ with probability at least $1-\delta$ using $O(n^6 \kappa^{-4}\log(n/\delta))$. 
\end{lemma}
\noindent We defer the proof to Appendix~\ref{app:proof-hf-gradient-error}.

\subsubsection{Gradient descent}
We now present and analyze the gradient descent algorithm for near-vacuum states. 

\begin{proposition}\label{prop:hf-gd}\label{lemma:hf-pgd}
    Suppose $\trd(\rho,\ketbra 0)\leq\gamma_0$ and let $\opt\equiv\inf_{\ket\psi\in\c G}\infd(\rho,\ketbra\psi)$. Then, given $\epsilon\in (0,1)$ and $\delta\in(0,1/2)$, there exists an algorithm which takes $\Otilde(n^6\epsilon^{-2}\log(1/\delta))$ copies of $\rho$ and outputs a pure Gaussian state $\ket*{\hat\psi}$ such that
    \begin{align}
        \infd(\rho,\ketbra*{\hat\psi})\leq\opt+\epsilon
    \end{align}
    with probability at least $1-\delta$.
\end{proposition}
\begin{proof}
\sitan{we could probably get rid of Lemma 3.13 and the proof below by citing a standard result about noisy zero-th order gradient descent here}
We first show that $\opt =\min_{\theta\in B_R}\infd(\rho,\ketbra{\psi(\theta)})$, i.e. the minimizer is near-vacuum. Let $\ket{\psi}$ be any pure Gaussian state with $\infd(\rho,\ketbra{\psi})\leq \infd(\rho,\ketbra{0})$. Then, since $\infd(\rho,\ketbra{0})\leq 2\trd(\rho,\ketbra{0})$,
\begin{align}
    \xi\equiv \trd(\ketbra{\psi},\ketbra{0})\leq \trd(\rho,\ketbra{\psi})+\trd(\rho,\ketbra{0})\leq \sqrt{\infd(\rho,\ketbra{0})}+\gamma_0 \leq \sqrt{2\gamma_0}+\gamma_0\leq R/4
\end{align}
Then, provided $R\leq 4\xi_0$, the proof of Lemma~\ref{lemma:barg-param-L2-small} gives that $\|\theta\|_{\c P}\leq\sqrt2(\xi+\xi^2)\leq2\xi\leq R/2$. This means that every pure Gaussian state with infidelity at least as good as the vacuum, including the minimizer, must lie in $B_R$. Thus, 
\begin{equation}\label{eq:hf-local-opt}
    \opt=\min_{\theta\in B_R}f_\rho(\theta)=f_\rho(\theta_*),
\end{equation}
We begin our algorithm by taking $\kappa =\min\{R/4,\sqrt{\epsilon/6}\}$ and $T=\max\{1,\log_{3/2}(R/(4\kappa))\}$. Starting with $\theta_0=0$, we then iterate the following gradient descent step for $t\in [T]$:
\begin{align}\label{eq:hf-gd-step}
    \theta_{t}=\theta_{t-1}-\frac 1 3 \widehat{\nabla f_{\rho}}(\theta_t)
\end{align}
where the step size is $1/3$ and $\widehat{\nabla f_{\rho}}(\theta_t)$ is the zeroth-order estimator from Lemma~\ref{lemma:hf-gradient-error}. After the last step, we simply output $\ket*{\hat \psi}=\ket{\psi(\theta_T)}$.

We first use our Hessian bounds to show that the iterates approach $\theta_*$. Since $\nabla f_\rho(\theta_*)=0$, the fundamental theorem of calculus gives:
\begin{align}\label{eq:grad-theta-t}
    \nabla f_\rho(\theta_t)=H_{\theta_t}(\theta_t-\theta_*),\quad H_{\theta_t}=\int_0^1 \nabla^2 f_\rho((1-s)\theta_*+s\theta_t)ds
\end{align}
where \eqref{eq:convex-smooth} implies $I\preceq H_{\theta_t}\preceq 3I $ and thus $\|I-H_{\theta_t}/3\|_{\rm op}\leq 2/3$. Then, by this fact, triangle inequality, \eqref{eq:hf-gd-step}, \eqref{eq:grad-theta-t}, and the additive error in \eqref{lemma:hf-gradient-error}, 
\begin{align}
    \|\theta_{t}-\theta_*\|_{\c P }&\leq \left\|\theta_{t-1}-\theta_* -\frac 1 3 \nabla f_\rho(\theta_{t-1})\right\|_{\c P }+\frac 1 3 \|\nabla f_\rho(\theta_{t-1})-\widehat{\nabla f_\rho}(\theta_{t-1})\|\leq \|(I-H_{\theta_{t-1}})(\theta_{t-1}-\theta_*)\|_{\c P}+\frac \kappa 3 \\
    &\leq \frac 2 3 \|\theta_{t-1}-\theta_*\|_{\c P}+\frac \kappa 3 \label{eq:thetat-conv-step}
\end{align}
Since $\|\theta_0-\theta_*\|_{\c P}\leq R/4$ by Proposition~\ref{prop:f-smooth-convex}\sitan{doesn't the first paragraph of this proof already establish, via Lemma~\ref{lemma:barg-param-l2-small}, that the minimizer is in the interior ($B_{R/2}$)? shouldn't this be sufficient? if so, we can cut out the last part of Proposition~\ref{prop:f-smooth-convex}}, if $\|\theta_{t-1}-\theta_*\|_{\c P}\leq R/4$, $\theta_t\in B_{R/2}$. Substituting into \eqref{eq:thetat-conv-step} with $\kappa\leq R/4$, we have that $\|\theta_{t}-\theta_*\|_{\c P}\leq R/4$ as well, meaning all iterates remain in $B_{R/2}$ by induction. Then, recurring \eqref{eq:thetat-conv-step} within itself, 
\begin{align}
    \|\theta_{T}-\theta_*\|_{\c P }&\leq \frac 2 3 \|\theta_{T-1}-\theta_*\|_{\c P}+\frac \kappa 3 \leq \frac 2 3 \left(\frac 2 3 \|\theta_{t-2}-\theta_*\|_{\c P}+\frac \kappa 3 \right)+\frac \kappa 3 \leq \cdots\leq (2/3)^T \frac R 4 +\frac \kappa 3 \sum_{j=0}^{T-1}(2/3)^j \\
    &\leq 2\kappa\label{eq:thetat-conv}
\end{align}
Then, by the fundamental theorem of calculus and \eqref{eq:convex-smooth}, since $\nabla f_\rho(\theta_*)=0$,
\begin{align}
    f_\rho(\theta_T)-f_\rho(\theta_*)&=\int_0^1 (1-t) D^2 f_\rho((1-t)\theta_*+t\theta_T)[\theta_T-\theta_*,\theta_T-\theta_*]dt\leq 3\|\theta_T-\theta_*\|_{\c P }^2 \int_0^1 (1-t)dt \\
    &\leq \frac 3 2 \|\theta_T-\theta_*\|_{\c P }^2
\end{align}
Applying \eqref{eq:thetat-conv}, this gives $f_{\rho}(\theta_T)-f_{\rho}(\theta_*)\leq 6\kappa^2 \leq \epsilon $ which combines with \eqref{eq:hf-local-opt} to prove the main guarantee. By Lemma~\ref{lemma:hf-gradient-error}, each gradient estimate uses $O(n^6\kappa^{-4}\log(n/\delta))$ copies by Lemma~\ref{lemma:hf-gradient-error}. Substituting the values for $\kappa, R, T$, this reproduces the claimed copy complexity. 
\end{proof}

Proposition~\ref{prop:hf-gd} is already poly-time and its degraded $\Otilde (n^6\epsilon^{-2})$ scaling is a technical artifact from using a zeroth-order method for gradient estimation.

\subsubsection{Direct gradient estimation}
\label{sssec:hf-direct-gradient}
We show that this artifact can be removed by directly estimating the gradient with a more complicated non-Gaussian measurement, achieving optimal copy complexity $\Otilde(n^2\epsilon^{-1})$. 
\begin{proposition}\label{prop:hf-gd-n2}\label{lemma:hf-pgd-n2}
    Suppose $\trd(\rho,\ketbra 0)\leq\gamma_0$ and let $\opt\equiv\inf_{\ket\psi\in\c G}\infd(\rho,\ketbra\psi)$. Then, given $\epsilon\in (0,1)$ and $\delta\in(0,1/2)$, there exists an algorithm which takes $\Otilde(n^2\epsilon^{-1}\log(1/\delta))$ copies of $\rho$ and outputs a pure Gaussian state $\ket*{\hat\psi}$ such that
    \begin{align}
        \infd(\rho,\ketbra*{\hat\psi})\leq\opt+\epsilon
    \end{align}
    with probability at least $1-\delta$.
\end{proposition}
We begin by giving a direct gradient estimation protocol for the gradient at the vacuum state. 
\begin{lemma}[Direct gradient estimation]\label{lemma:hf-direct-gradient-est}
For state $\sigma$, $\kappa\in(0,1)$, and $\delta\in(0,1/2)$, we can
estimate $\nabla f_\sigma(0)$ to error $\kappa$ in $\|\cdot\|_{\c P}$ using
$O(n^2\kappa^{-2}\log(2/\delta))$ copies and general single-copy measurements.
\end{lemma}
\begin{proof}
    Let $\Pi$ be the projector onto $\c F_{\leq 2}$, and define $P=\Pi_{\leq 2}-\ketbra 0$ and $d_{\leq 2}=\dim \c F_{\leq 2}=1+n+n(n+1)/2$. Then, from the proof of Proposition~\ref{prop:f-smooth-convex}, 
    \begin{align}
    \langle h,\nabla f_\sigma (0)\rangle_{\c P}=D f_\sigma(0)[h]=-2\Re \bra{0}K_h^\dagger \sigma \ket{0}=\Re\langle K_h\ket 0,-2P\sigma \ket 0 \rangle 
\end{align}
where we use the fact that $K_h\ket 0\in \c F_{\leq 2}$ is orthogonal to $\ket{0}$ giving $PK_h\ket 0=K_h\ket 0$. Since we know that $\langle h,g\rangle_{\c P}=\Re\langle K_h\ket{0},K_g\ket{0}\rangle$, it follows that $K_{\nabla f_\sigma(0)}\ket 0 =-2P\sigma \ket{0}$ and that it suffices to estimate this quantity to $\ell_2$ error $O(\kappa)$, since we can simply read off the Bargmann parameters from the Fock coefficients of the state. 

    We do this using the POVM $\c M =\{I-\Pi, d_{\leq 2}\ketbra{v} d\mu(v)\}$ where $\mu$ is the uniform measure on the unit sphere of $\c F_{\leq 2}$. For outcome $v$, we record $X=(d_{\leq 2}+1)P\ket{v}\braket{v}{0}$, and for $I-\Pi$, we record $X=0$. Then, 
    \begin{align}
        \E[X]&=d_{\leq 2}(d_{\leq 2}+1)\int P\ket{v}\braket{v}{0}\cdot \bra{v}\sigma \ket{v}d\mu(v)\\
        &=(d_{\leq 2}+1)P\left(d_{\leq 2}\int \bra{v}\sigma\ket{v}\ketbra{v} d\mu(v)\right)\ket{0}\\
        &=(d_{\leq 2}+1)P\cdot \frac{\Pi\sigma\Pi +\tr(\Pi\sigma)\Pi }{d_{\leq 2}+1}\cdot \ket{0}\\
        &=P\sigma \ket{0}
    \end{align}
    where we use the standard Haar second-moment identity and use $P$ to annihilate the second term. For outcome $v$, $\|X\|_2^2\leq (d_{\leq 2}+1)^2\lvert \braket{v}{0}\rvert^2 $ since $\|P\ket{v}\|_2^2\leq 1$. Then, 
    \begin{align}
        \E[\|X\|_2^2]&\leq d_{\leq 2}(d_{\leq 2}+1)^2\int  \lvert \braket{v}{0}\rvert^2\bra{v}\sigma \ket{v}d\mu(v)\leq (d_{\leq 2}+1)^2\bra{0}\left(d_{\leq 2}\int \bra{v}\sigma\ket{v}\ketbra{v} d\mu(v)\right)\ket{0}\\
        &\leq (d_{\leq 2}+1)^2\bra{0}\cdot \frac{\Pi\sigma\Pi +\tr(\Pi\sigma)\Pi }{d_{\leq 2}+1}\cdot \ket{0}\\
        &\leq (d_{\leq 2}+1)(\bra{0}\sigma \ket{0}+\tr(\Pi\sigma))\\
        &\leq 2 (d_{\leq 2}+1)
    \end{align}
    A standard median-of-means estimator obtains error $\kappa$ with $O(d_{\leq 2}\kappa^{-2}\log(1/\delta))$ copies. \alvantodo{check}
\end{proof}

Since Lemma~\ref{lemma:hf-direct-gradient-est} provides an estimator for the gradient at vacuum, we can perform gradient descent in the same fashion as Proposition~\ref{prop:hf-gd} except that we conjugate by a Gaussian unitary to send our current iterate to vacuum, then estimate the gradient, then take the step, and finally undo all the Gaussian unitaries we've accumulated. A similar analysis to the proof of Proposition~\ref{prop:hf-gd} with this estimator then gives Theorem~\ref{thm:hf}, using the smoothness and strong convexity to handle these additional rotations. 

\section{Agnostic learning in the low-fidelity regime}\label{sec:lf}

In this section we provide our agnostic learner for the low-fidelity regime.
As mentioned above, the algorithm will work in two parts: first, we give an iterative unsqueezing procedure based on repeated calls to list-decodable learning that yields a short list of candidate bases, with the guarantee that the best Gaussian state approximator is approximately contained in the low photon number Fock subspace in one of these rotated bases.
Then, we will show how to use tensor optimization to find this Gaussian approximator in this candidate basis.

A Gaussian unitary $G$ defines the \emph{rotated basis} $\{G^\dagger\ket m\}$ of Fock space~\cite{bittel2025energy}, which we identify with $G$; written in this basis, a state $\sigma$ becomes $G\sigma G^\dagger$. 
Let $\pi_*$ be a closest pure Gaussian state to $\rho$, let $V_*$ denote its covariance, and let $E_0$ be an upper bound on its mean photon number $\nu(\pi_*)$, defined in Section~\ref{sec:prelims}.
Throughout this section, we will assume that $F(\rho, \pi_*) \geq \alpha$.

Throughout this section, we will often work in different Gaussian bases.
Thus, for conciseness, when we say we perform a measurement in a basis $G$, we mean applying $G$ to fresh copies of $\rho$ before measuring them.
Note that such a change of basis preserves fidelity, but can change the mean photon number. 
The unrotated Fock basis is $G=I$, and a state obtained in the basis $G$ is rotated back by applying $G^\dagger$.

\subsection{List-decodable iterative unsqueezing}
\label{ssec:lf-robust-unsqueezing}

This subsection gives the algorithm for constructing a short list of rotated bases, in one of which $\pi_*$ has low photon number. %The procedure has the same structure as ordinary iterative unsqueezing~\cite{bittel2025energy}: estimate the covariance in the current basis by heterodyne tomography, change to the basis given by a Williamson decomposition of the estimate, which reduces the squeezing, and repeat. The difference is that the heterodyne outcomes no longer follow a single Gaussian distribution. When $F(\rho,\pi_*)\geq\alpha$, the outcome distribution of $\rho$ is only guaranteed to overlap with that of $\pi_*$ on a fraction $\alpha/2$ of its mass, and the rest is arbitrary. 
%We will repeatedly use list-decodable estimation, which returns a short \emph{list} of candidate estimates, at least one of which is accurate. Since we cannot tell which candidate is accurate, we follow all of them for a fixed number of rounds, and then prune the resulting bases so that the list stays short. The invariant is that the list always contains a \emph{useful} basis, in which $\pi_*$ has bounded squeezing. At the end, $\pi_*$ also has bounded photon number in one such basis.

\paragraph{Notation.}
%Throughout this subsection, $\pi_*$ is a fixed pure Gaussian state with $F(\rho,\pi_*)\geq\alpha$, which the statements below call the \emph{comparator}; it need not be a closest Gaussian state to $\rho$. We write $V_*$ for its covariance. 
For a rotated basis $G$ and a state $\sigma$, let
\[
    N_G=G^\dagger NG,\qquad \nu_G(\sigma)=\tr(\sigma N_G)=\nu(G\sigma G^\dagger)
\]
be the photon-number operator and the mean photon number in the basis $G$. When we are working in a basis $G$, we write ${\bf m}_G$ and $V_G$ for the mean and covariance of $G\pi_*G^\dagger$, and measure its squeezing by $\kappa(V)=\|V\|_{\rm op}=\|V^{-1}\|_{\rm op}$. 
When the basis is clear, we will often drop the subscripts.

We call $G$ \emph{useful at level} $A$ if $\kappa(V)\leq A$. Heterodyne measurement of $G\pi_*G^\dagger$ produces samples from $\c N({\bf m},\Sigma)$, where $\Sigma=(V+I)/2$ is its \emph{heterodyne covariance}. For a symplectic matrix $S$, we write $U_S$ for the Gaussian unitary with no displacement that maps a state with moments $({\bf m},V)$ to one with moments $(S{\bf m},SVS^\top)$. Finally, let
\[
    t_\alpha=C_{\rm dec}(4/\alpha)^{r_{\rm dec}},
\]
where $C_{\rm dec},r_{\rm dec}>1$ are the universal constants of Lemma~\ref{lemma:list-decodable-gaussian}. 
The main output of this section is the following theorem:
%This is the parameter $T_{\alpha/2}$ of that lemma, which governs the accuracy, list size and sample size of the decoder when a fraction $\alpha/2$ of the samples is good.

\begin{theorem}[List-decodable unsqueezing]\label{thm:lfpr-unsqueezing}
Fix $\alpha\in(0,1]$, $A_0\geq1$, and $\eta\in(0,1/2]$. There is an algorithm which, given copies of an arbitrary $n$-mode state $\rho$, returns a list of at most $W=\lceil8/(3\alpha)\rceil$ Gaussian bases, or reports failure. For every fixed best Gaussian approximator $\pi_*$ satisfying
\[
    F(\rho,\pi_*)\geq\alpha,\qquad \kappa(V_*)\leq A_0,
\]
with probability at least $1-\eta$ it returns a list containing a basis $G$ such that
\begin{equation}\label{eq:lfpr-final-frame}
    \nu(G\pi_*G^\dagger)\leq H_\alpha,
\end{equation}
where we may take $H_\alpha=e\mathsf K_\alpha^2\asymp\alpha^{-s}$ for a
universal exponent $s\asymp r_{\rm dec}$, with $\mathsf K_\alpha$ defined
in \eqref{eq:lfpr-contraction-constant}. The number of copies and total runtime of the algorithm are both at most
\begin{equation}\label{eq:lfpr-preprocessing-cost}
    \exp (C t_\alpha) n^{t_\alpha} \poly\bigl(\log(1/\eta),\,\log\log(e+A_0)\bigr),
\end{equation}
where $C_{\rm dec},r_{\rm dec}$ are as in Lemma~\ref{lemma:list-decodable-gaussian}.
% \sitan{what is the purpose of defining $t_\alpha$ if it is just $t_\alpha$?}
% Here and below $C_{\rm dec}$ is chosen large enough that $t_\alpha$ also dominates the universal exponent of the matrix computations. No bound on the displacement of $\pi_*$ or the energy of $\rho$ is required. As in Theorem~\ref{thm:lf-outline-target}, decoder numerical runtime and measurement synthesis costs are separate.
\end{theorem}

\noindent The algorithm runs in blocks. Each block consists of $q$ rounds, followed by a fit and a pruning step. In each round, every basis in the current list is replaced by its updated bases, one for each candidate in a list estimate (Section~\ref{sssec:lfpr-list}). Along a branch that always follows an accurate candidate, each round roughly takes the square root of the squeezing of $\pi_*$ (Section~\ref{sssec:lfpr-unsqueeze}). After $q$ rounds, a final fit changes to bases in which estimates of $\pi_*$ become the vacuum; this controls the photon number of $\pi_*$, not only its squeezing. The block then prunes the list back to at most $W$ bases, using a photon-cutoff test and pairwise comparisons of bases (Section~\ref{sssec:lfpr-prune}). Section~\ref{sssec:lfpr-algorithm} states the algorithm and proves Theorem~\ref{thm:lfpr-unsqueezing}.

\subsubsection{List estimates}\label{sssec:lfpr-list}
% Robust list-decoding and unsqueezing primitives for Section 4.

\paragraph{Heterodyne list estimates.}
We first record what one run of list-decodable estimation guarantees in a single basis. The next lemma applies Lemma~\ref{lemma:list-decodable-gaussian} to heterodyne outcomes.

% Fix $\alpha\in(0,1]$ and a pure Gaussian comparator \jerry{don't use comparator} $\pi$ with
% $F(\rho,\pi)\geq\alpha$. Write $\Omega=\Omega_n$.
% In the current frame $G$, write $({\bf m},V)$ for the mean and covariance of $G\pi G^\dagger$, and $\Sigma=(V+I)/2$ for its heterodyne covariance.
% % \sitan{definition of frame here is repeated from beginning of the section}

% Below, we state the following consequence of classical results for list-decodable estimation, applied to heterodyne readouts.

\begin{lemma}[List estimates]\label{lemma:lfpr-list}
%Set
%\begin{equation}\label{eq:lfpr-decoder-constants}
% L_\alpha=\lceil e^{t_\alpha}\rceil.
%\end{equation}
% \sitan{we should avoid introducing too many new variable names, as this is a hallmark of GPT writing. in Eq. (52), why not just put the expressions directly on the RHS? in particular, when two variables are equal to the same value, we really should collapse them to a single variable name (e.g., $\Lambda = t_\alpha$)}
In any adaptively
chosen basis, a classical list decoding algorithm run on
$s=O \left(  e^{t_\alpha}(3n)^{t_\alpha} \right)$ fresh heterodyne outcomes returns
at most $O(e^{t_\alpha})$ pairs, each with positive definite covariance. 
With probability at least $9/10$, one pair $(\hat\mu,\hat\Sigma)$ satisfies
all three bounds
\begin{align}\label{eq:lfpr-list-guarantees}
 t_\alpha^{-1}\Sigma\preceq\hat\Sigma\preceq t_\alpha\Sigma,\qquad
 \bigl\|\Sigma^{-1/2}(\hat\Sigma-\Sigma)\Sigma^{-1/2}\bigr\|_F
 &\leq t_\alpha^2,\\
 \bigl\|\Sigma^{-1/2}(\hat\mu-{\bf m})\bigr\|_2&\leq\sqrt{t_\alpha(1+t_\alpha)}.\notag
\end{align}
% The union of $b$ independent such lists contains at most $L_\alpha b$
% pairs and fails to contain a joint estimate satisfying
% \eqref{eq:lfpr-list-guarantees} with conditional probability at most
% $e^{-b}$.
\end{lemma}

\begin{proof}
Let $P$ and $Q=\c N({\bf m},\Sigma)$ be the heterodyne laws of the input
and $\pi_*$ in the chosen basis. Measurement contractivity and the
pure-state fidelity bound imply
\[
 1-\tvd(P,Q)\geq1-\trd(\rho,\pi)
 \geq1-\sqrt{1-\alpha}\geq\alpha/2.
\]
The result then directly follows from Lemma~\ref{lemma:list-decodable-gaussian}.
% samples and decoder randomness in the $b$ repetitions bounds their joint
% failure probability by $0.02^b\leq e^{-b}$.
\end{proof}

\paragraph{Good pairs.}
We call a pair $(\hat\mu,\hat\Sigma)$ \emph{good} for a basis $G$ if it satisfies the three bounds \eqref{eq:lfpr-list-guarantees}, where ${\bf m}$ and $\Sigma$ are the mean and heterodyne covariance of $\pi_*$ in the basis $G$. By Lemma~\ref{lemma:lfpr-list}, a single list contains a good pair with probability at least $9/10$. To boost this, the algorithm runs the decoder $b$ times on fresh copies and takes the union of the resulting lists, and so it fails to contain a good pair with probability at most $10^{-b}$. 
% The algorithm never needs to know which pairs are good; only the analysis follows them.

\paragraph{Updated bases.}
Each pair $(\hat\mu,\hat\Sigma)$ in a list determines a new basis. Inflate the covariance estimate and take a Williamson decomposition,
\begin{equation}\label{eq:lfpr-inflated-covariance}
    \hat V=2\Lambda\hat\Sigma=S\widehat D S^\top,
    \qquad V+I\preceq\hat V\preceq\Lambda^2(V+I),
\end{equation}
where $S$ is symplectic, $\widehat D=\operatorname{diag}(d_1,\ldots,d_n,d_1,\ldots,d_n)$ with $d_j>0$, and the inequalities hold whenever the pair is good, by the spectral bound in \eqref{eq:lfpr-list-guarantees}. If the current basis is $G$, the \emph{updated basis} of the pair is $U_{S^{-1}}G$. In the updated basis, a covariance $V$ becomes $S^{-1}VS^{-\top}$; in particular, $\hat V$ becomes the diagonal matrix $\widehat D$. The inflation by $2\Lambda$ makes $\hat V$ an overestimate of $V$, which is what the unsqueezing lemma requires. 
%The mean estimate $\hat\mu$ is not used in this step.

% We call a candidate pair $(\hat\mu,\hat\Sigma)$ \emph{good} if it satisfies all three bounds in \eqref{eq:lfpr-list-guarantees} for the fixed state $\pi$ in the current frame. By Lemma~\ref{lemma:lfpr-list}, each list contains a good candidate with probability at least $9/10$.

% % We discard non-positive-definite covariance candidates. 
% For each candidate pair $(\hat\mu,\hat\Sigma)$, we use a Williamson decomposition of the matrix
% \begin{equation}\label{eq:lfpr-inflated-covariance}
%  W=2t_\alpha\hat\Sigma=S\widehat D S^\top,
%  \qquad V+I\preceq W\preceq t_\alpha^2(V+I)
% \end{equation}
% to construct a new ``child frame.''
% % with $S$ symplectic and $\widehat D=\operatorname{diag}(d_1,\ldots,d_n,d_1,\ldots,d_n)$, where $d_j>0$ are the symplectic eigenvalues of $W$. 
% For a good candidate, the inequalities follow from the spectral bound in \eqref{eq:lfpr-list-guarantees} and $\Sigma=(V+I)/2$.
% % \sitan{spell out more clearly what a good candidate is}
% If the current frame is $G$, the child frame is $G'=U_{S^{-1}}G$, where $U_{S^{-1}}$ is the Gaussian unitary with no displacement that maps first moments ${\bf m}$ to $S^{-1}{\bf m}$. It transforms a covariance $V$ to $S^{-1}VS^{-\top}$; in particular, $W$ gets mapped to $S^{-1}WS^{-\top}=\widehat D$. This is the change of frame specified by the Williamson decomposition.
% Inflation is used only for this step: its relative Frobenius error can
% grow with $n$, so the terminal fit below uses the raw estimate.\sitan{this last sentence is very jargon-ey..}

\subsubsection{Unsqueezing and the fit}\label{sssec:lfpr-unsqueeze}

\paragraph{The unsqueezing step.}

We now show that changing to the updated basis of a good pair roughly takes the square root of the squeezing of $\pi_*$. Let $V$ be the covariance of $\pi_*$ in the current basis, and let $\kappa=\kappa(V)=\|V^{-1}\|_{\rm op}\geq1$ be its squeezing. For a good pair, \eqref{eq:lfpr-inflated-covariance} shows that $\hat V$ overestimates $V$ by at most a constant multiple of $V+I$. Indeed, $\kappa V\succeq I$ gives $V\succeq(V+I)/(\kappa+1)$, so
\[
    0\preceq\hat V-V\preceq\Bigl(\Lambda^2-\frac{1}{\kappa+1}\Bigr)(V+I).
\]
The coefficient on the right is positive, so this is the hypothesis \eqref{eq:williamson-sandwich} of Lemma~\ref{lemma:williamson-unsqueezing} with $\zeta\in(0,1)$ chosen so that $2\zeta/(1-\zeta)=\Lambda^2-1/(\kappa+1)$. The algorithm never needs $\zeta$, which depends on the unknown $V$; the updated basis is computed from $\hat V$ alone. The lemma bounds the squeezing of $\pi_*$ in the updated basis, where its covariance is $V'=S^{-1}VS^{-\top}$:
\begin{equation}\label{eq:lfpr-unsqueeze}
    \|(V')^{-1}\|_{\rm op}
    \leq\sqrt{1+\Bigl(\Lambda^2-\frac{1}{\kappa+1}\Bigr)(\kappa+1)}
    =\Lambda\sqrt{\kappa+1}
    \leq\sqrt2\,\Lambda\sqrt\kappa,
\end{equation}
where the last step uses $\kappa\geq1$. Thus one round replaces the squeezing $\kappa$ by at most $\sqrt2\Lambda\sqrt\kappa$.

Now follow good pairs for $q$ consecutive rounds, and let $\kappa_0,\kappa_1,\ldots,\kappa_q$ be the squeezing of $\pi_*$ in the successive bases, with $\kappa_0\leq A$. Since $\kappa_{j+1}\leq\sqrt2\Lambda\sqrt{\kappa_j}$ for each $j$, unrolling the recursion gives
\begin{equation}\label{eq:lfpr-block-unsqueeze}
    \kappa_q\leq(\sqrt2\Lambda)^{1+\frac12+\cdots+2^{-(q-1)}}A^{2^{-q}}
    =(\sqrt2\Lambda)^{2(1-2^{-q})}A^{2^{-q}}
    \leq2\Lambda^2A^{2^{-q}},
\end{equation}
where the last step uses $\sqrt2\Lambda\geq1$. The constant factors do not accumulate: their exponents form a geometric series that sums to less than $2$. So however many rounds we take, the overhead is at most $2\Lambda^2$, while the exponent of $A$ shrinks by a factor of $2^q$.

% For a good list estimate, apply Lemma~\ref{lemma:williamson-unsqueezing}
% with $\hat V=W$ and $\hat S=S$.
% Since $V$ is pure, $a=\|V^{-1}\|_{\rm op}\geq1$ and $I\preceq aV$.
% Set $c_a=t_\alpha^2-1/(a+1)>0$ and $\zeta_a=c_a/(c_a+2)\in(0,1)$.
% The inequality $V\succeq(V+I)/(a+1)$ and
% \eqref{eq:lfpr-inflated-covariance} imply
% \[
%     0\preceq W-V\preceq t_\alpha^2(V+I)-V
%        \preceq c_a(V+I)
%        =\frac{2\zeta_a}{1-\zeta_a}(V+I).
% \]
% Thus the covariance estimate satisfies the lemma's hypothesis with
% $\zeta=\zeta_a$. This parameter is used only in the analysis;
% the frame update is computed from $W$.
% For the child covariance $V'=S^{-1}VS^{-\top}$, the lemma gives
% \begin{equation}\label{eq:lfpr-unsqueeze}
%  \|(V')^{-1}\|_{\rm op}\leq\sqrt{1+c_a(a+1)}=t_\alpha\sqrt{a+1}\leq{}\sqrt2t_\alpha\sqrt a.
% \end{equation}
% Consequently, along $q$ good child choices, the inverse-covariance norms
% $a_j$ obey $a_{j+1}\leq\sqrt2t_\alpha\sqrt{a_j}$.
% Starting from $a_0\leq A$, iteration gives
% \begin{equation}\label{eq:lfpr-block-unsqueeze}
%  a_q\leq(\sqrt2t_\alpha)^{2(1-2^{-q})}A^{2^{-q}}\leq{} 2t_\alpha^2 A^{2^{-q}},
% \end{equation}
% with a constant independent of $q$. \sitan{this paragraph should be rewritten}

\paragraph{The final output.}
After $q$ rounds, the squeezing of $\pi_*$ is bounded independently of $n$ along a good branch. This alone does not bound its photon number independently of $n$: a product of $n$ single-mode squeezed states, each with squeezing $A$, has photon number $n(A+A^{-1}-2)/4$. To remove the dependence on $n$, we make one final change of basis, to a basis in which an estimate of $\pi_*$ becomes the vacuum. Its analysis uses the Frobenius and mean bounds in \eqref{eq:lfpr-list-guarantees}, which control the total estimation error across all modes, and not only the spectral bound, which controls the error in each direction separately. For the same reason, the fit starts from the raw estimate $\hat\Sigma$ rather than from the inflated matrix $\hat V$: inflating by a factor of $\Lambda$ adds a relative error close to $(\Lambda-1)I$, whose Frobenius norm grows like $\sqrt n$.

The estimate $\hat\Sigma$ need not be the heterodyne covariance of a pure state, so we first round it to one. This is easiest in terms of the heterodyne precision $K=\Sigma^{-1}=2(V+I)^{-1}$, because purity is a linear constraint on $K$. With $\Omega=\Omega_n$, a matrix $0\prec K\prec2I$ is the precision of a pure Gaussian state if and only if
\begin{equation}\label{eq:lfpr-pure-precision}
    K+\Omega K\Omega^\top=2I.
\end{equation}
Indeed, if $V$ is pure then $\Omega V\Omega^\top=V^{-1}$, so $\Omega K\Omega^\top=2(V^{-1}+I)^{-1}=2V(V+I)^{-1}$, and adding $K=2(V+I)^{-1}$ gives $2I$. Conversely, if $K$ satisfies \eqref{eq:lfpr-pure-precision}, then $V=2K^{-1}-I$ is positive definite and $\Omega V\Omega^\top=2(2I-K)^{-1}-I=K(2I-K)^{-1}=V^{-1}$, so $V$ is the covariance of a pure state.

We also keep the eigenvalue bounds that the true precision satisfies. If $A^{-1}I\preceq V\preceq AI$, then the eigenvalues of $K$ lie in $[2/(A+1),\,2A/(A+1)]$, so $K$ lies in the set
\begin{equation}\label{eq:lfpr-precision-domain}
    \mathcal K_A=\Bigl\{K=K^\top:\ K+\Omega K\Omega^\top=2I,\ \ \tfrac{1}{A+1}I\preceq K\preceq\bigl(2-\tfrac{1}{A+1}\bigr)I\Bigr\}.
\end{equation}
This set is closed, convex and nonempty, since it contains the precision $I$ of the vacuum. So every symmetric matrix $\hat K$ has a unique closest point in $\mathcal K_A$ in Frobenius norm, and it has a closed form. First project onto the affine constraint, $\hat K\mapsto I+X$ with $X=(\hat K-\Omega\hat K\Omega^\top)/2$; then clip the eigenvalues of $X$ to $[-\tfrac{A}{A+1},\tfrac{A}{A+1}]$. Clipping preserves the affine constraint, because $\Omega X\Omega^\top=-X$ and clipping is an odd function of $X$. Since the first map is an orthogonal projection and clipping is the Frobenius projection onto the eigenvalue constraint, their composition is the Frobenius projection onto $\mathcal K_A$.

The fit takes a \emph{terminal pair}, meaning a pair $(\hat\mu,\hat\Sigma)$ from a list computed in the final basis $G$ of a block. Let $K_0$ be the point of $\mathcal K_A$ closest to $\hat\Sigma^{-1}$. By the equivalence above, $V_0=2K_0^{-1}-I$ is a pure covariance; write $V_0=S_0S_0^\top$ with $S_0$ symplectic. The \emph{fitted basis} is obtained from $G$ by the displacement by $-\hat\mu$ followed by $U_{S_0^{-1}}$. In it, the pure Gaussian state with mean $\hat\mu$ and covariance $V_0$, which is our estimate of $\pi_*$, becomes the vacuum. The algorithm fits every terminal pair. In the next lemma, $V$ is the covariance of $\pi_*$ in the basis $G$ where the terminal list was computed.

% Note that bounded squeezing alone does not prevent that total photon number is proportional to $n$.
% We now use the Frobenius and mean guarantees
% to remove that dependence. For a pure covariance, its heterodyne
% precision $K(V)=2(V+I)^{-1}$ satisfies
% \begin{equation}\label{eq:lfpr-pure-precision}
%  K(V)+\Omega K(V)\Omega^\top=2I.
% \end{equation}
% Indeed, the left side equals
% $2(V+I)^{-1}+2(V^{-1}+I)^{-1}=2I$. Conversely, if
% $0\prec K\prec2I$ obeys this identity, then $V=2K^{-1}-I$ is a pure
% covariance: $\Omega V\Omega^\top=2(2I-K)^{-1}-I=V^{-1}$.

% Given $A\geq1$, define
% \begin{equation}\label{eq:lfpr-precision-domain}
%  \ell_A=(A+1)^{-1},\qquad
%  \mathcal K_A=\{K=K^\top:K+\Omega K\Omega^\top=2I,
%        \ \ell_A I\preceq K\preceq(2-\ell_A)I\}.
% \end{equation}
% For each pair in a fresh terminal list, project the raw precision
% $\hat K=\hat\Sigma^{-1}$ onto $\mathcal K_A$ in Frobenius norm,
% obtaining $K_0$, and set $V_0=2K_0^{-1}-I=S_0S_0^\top$ with $S_0$
% symplectic. Displace by $-\hat\mu$ and then apply $S_0^{-1}$. This
% maps the fitted pure Gaussian state to vacuum.

\begin{lemma}[the fitted basis]\label{lemma:lfpr-warm-frame}
If $A^{-1}I\preceq V\preceq AI$ and the terminal pair satisfies
\eqref{eq:lfpr-list-guarantees}, then the fitted basis maps the
best Gaussian approximator to a state with total mean photon number at most
\begin{equation}\label{eq:lfpr-warm-energy}
 h_\alpha(A)=u_\alpha A^8,\qquad
 u_\alpha=1200\bigl(t_\alpha^6+t_\alpha(1+t_\alpha)\bigr)\asymp t_\alpha^6.
\end{equation}
% Thus the universal exponent is $p=8$, and
% $u_\alpha\asymp\alpha^{-6r_{\rm dec}}$.
\end{lemma}
\noindent We defer the proof to Appendix~\ref{app:proof-lfpr-warm-frame}.

% The precision projection is explicit. Its affine part sends $\hat K$
% to $I+X$, where $X=(\hat K-\Omega\hat K\Omega^\top)/2$. Clip the
% eigenvalues of $X$ to $[-(1-\ell_A),1-\ell_A]$ and add $I$.
% Because $\Omega X\Omega^\top=-X$ and the clipping function is odd,
% this preserves the affine constraint. Spectral clipping is the
% Frobenius projection onto the spectral interval; orthogonality of the
% first projection proves that their composition projects onto
% $\mathcal K_A$.

% \paragraph{Cost convention.}
% The sample bound in Lemma~\ref{lemma:lfpr-list} is independent of covariance
% scale. The numerical runtime of its classical decoder also depends on
% bit parameters; see \cite[Section~3.1 and Theorem~8.6]{ivkov2022list}.
% We therefore count decoder calls separately from the matrix computations
% above. A scale-independent copy bound does not assert a
% scale-independent numerical implementation of that decoder.

\subsubsection{Pruning the list}\label{sssec:lfpr-prune}
% Included in the robust list-decodable unsqueezing subsection of main.tex.

We prune the bases in two ways. First, in a useful basis $\rho$ has substantial weight at low photon number, so we discard bases in which $\rho$ has little weight below a photon cutoff. Second, among the remaining bases, we discard any basis that is close to one we keep. We measure the closeness of two bases $G$ and $H$ by
\begin{equation}\label{eq:lfpr-compatibility}
    \Delta(G,H)=\lambda_{\min}(N_G+N_H),
\end{equation}
the smallest possible total photon number $\nu_G(\sigma)+\nu_H(\sigma)$ of a state $\sigma$ in the two bases. It is symmetric, and $\Delta(G,G)=0$. The next lemma computes $\Delta$ from the relative Gaussian transformation between the two bases, without manipulating operators on the full Fock space.
% \paragraph{Pruning Gaussian frames.}
% For a Gaussian unitary $G$, write
% \[
%     N_G=G^\dagger NG,
%     \qquad \nu_G(\sigma)=\tr(\sigma N_G).
% \]
% for the total photon-number operator and the mean photon number of a state $\sigma$ under frame $G$. 
% % A useful frame need not give the input appreciable vacuum probability directly.
% % We therefore test the whole low-photon subspace: applying $G$ and measuring
% % $\{\Pi_{\leq k},I-\Pi_{\leq k}\}$ estimates
% % $w_G=\tr(\rho P_G(k))$ to additive error $a$ using
% % $O(a^{-2}\log(2/\eta))$ fresh copies, with failure probability at most $\eta$. 

% To decide which frames may be pruned, we use the symmetric quantity following as a distance metric between two frames:
% \begin{equation}\label{eq:lfpr-compatibility}
%     \Delta(G,H)=\lambda_{\min}(N_G+N_H).
% \end{equation}
% It is the smallest achievable sum of photon numbers in the two frames.
% The following formula makes it computable without manipulating operators
% on the full Fock space.

\begin{lemma}\label{lemma:lfpr-compatibility}
Let $U=GH^\dagger$, and let $U\ket0$ have mean $\mathbf d$, covariance
$C$, and squeezing parameters $r_1,\ldots,r_n\geq0$. Then
\begin{equation}\label{eq:lfpr-compatibility-formula}
    \Delta(G,H)=\sum_{j=1}^n(\cosh r_j-1)
       +\frac12\mathbf d^\top(I+C)^{-1}\mathbf d.
\end{equation}
Thus $\Delta$ is computable from the relative Gaussian transformation
using polynomially many matrix operations.
\end{lemma}
\noindent We defer the proof to Appendix~\ref{app:proof-lfpr-compatibility}.

\paragraph{Low-photon subspaces.}
For an integer $k\geq1$, let $\Pi_{\leq k}$ be the projection onto states with total photon number at most $k$, and let
\[
    P_G(k)=G^\dagger\Pi_{\leq k}G
\]
be the corresponding projection in the basis $G$. We call $\tr(\rho P_G(k))$ the \emph{cutoff mass} of $\rho$ in the basis $G$. Applying $G$ to a fresh copy of $\rho$ and testing whether the photon number is at most $k$ succeeds with probability exactly $\tr(\rho P_G(k))$. By Hoeffding's inequality, $O(a^{-2}\log(1/\tau))$ copies therefore estimate the cutoff mass to additive error $a$ with failure probability at most $\tau$, independently of $k$.

The next proposition gives the two properties of $\Delta$ that make pruning work. Part~(i) says that if $\Delta(G,H)$ is small, then a pure Gaussian state with few photons in one of the bases also has few photons in the other, so discarding one of the two bases loses little. Part~(ii) says that if $\Delta(G,H)$ is large, then the low-photon subspaces of $G$ and $H$ are nearly orthogonal. By the packing bound of Lemma~\ref{lemma:lfpr-packing} below, $\rho$ can then have large cutoff mass in only a few bases that are pairwise far apart. So we fix a threshold $T$ and discard a basis whenever $\Delta$ between it and some kept basis is at most $T$: by~(i) this costs only a bounded factor in the photon number, and by~(ii) few bases survive.

% Write
% \[
%     \qquad P_G(k)=G^\dagger\Pi_{\leq k}G,
% \]
% for the low-photon subspace. 

% Following we prove two key properties of our metric $\Delta$. One is saying that when $\Delta(G,H)$ is small the ratio between $\nu_G(\pi)$ and $\nu_H(\pi)$ should be bounded accordingly, the other one says that when $\Delta(G,H)$ is large, actually there is a separation for their low-photon projections $P_G(k)$, $P_H(k)$. 

% Hence for pruning process we choose a suitable threshold $T$, and pruning frame $H$ when there is another frame $G$ on the list such that $\Delta(G,H)<T$. By the first property it do not lose too much on the bounds of the total photon number, and by the second property there can not be left with too much frames on the list after pruning.

\begin{proposition}\label{prop:lfpr-frame-comparison}
Let constants $C_{\rm tr}=32$, $d=32$, and a sufficiently large universal
$C_{\rm sep}\geq1$. 
% In particular, these constants are independent of $\epsilon$, $\opt$, and $\alpha$.
\begin{enumerate}[label=(\roman*)]
    \item For every pure Gaussian $\pi$ with $\nu_G(\pi)\leq h$ and $h\geq1$,
    \begin{equation}\label{eq:lfpr-energy-transfer}
        \nu_H(\pi)\leq C_{\rm tr}(\Delta(G,H)+1)^3h.
    \end{equation}
    \item For every integer $k\geq 1$ and $0<\mu\leq1/2$, set
    \begin{equation}\label{eq:lfpr-separation-threshold}
        T(k,\mu)=C_{\rm sep}\left(\frac{k}{\mu}\right)^d.
    \end{equation}
    If $\Delta(G,H)>T(k,\mu)$, then
    $\|P_G(k)P_H(k)\|_{\mathrm{op}}\leq\mu$.
\end{enumerate}
\end{proposition}

\begin{proof}[Proof of (i)]
Work in basis $G$. Let $(\mathbf m,V)$ be the moments of $\pi_*$
and $(\mathbf d,C)$ those of the $H$-vacuum in this basis.
Writing $C=SS^\top$, the passage to basis $H$ transforms these moments
to $S^{-1}(\mathbf m-\mathbf d)$ and $S^{-1}VS^{-\top}$, up to a passive
transformation. Hence
\begin{equation}\label{eq:lfpr-relative-energy}
    \nu_H(\pi)=\frac{\tr(C^{-1}V)-2n}{4}
       +\frac12(\mathbf m-\mathbf d)^\top C^{-1}(\mathbf m-\mathbf d).
\end{equation}
Put $\Delta=\Delta(G,H)$ and $b=(\tr C-2n)/4$.
By \eqref{eq:lfpr-compatibility-formula},
\begin{align}\label{eq:lfpr-relative-moment-bounds}
    b&=\sum_j\sinh^2r_j
       =\sum_j(\cosh r_j-1)^2+2\sum_j(\cosh r_j-1)
       \leq\Delta(\Delta+2),\\
    \|C^{\pm1}\|_{\mathrm{op}}&\leq4(\Delta+1)^2,
    \qquad Q:=\mathbf d^\top(I+C)^{-1}\mathbf d\leq2\Delta.\notag
\end{align}
Comparing eigenvalues of $C^{-1}$ and $(I+C)^{-1}$ then gives
\[
    \mathbf d^\top C^{-1}\mathbf d
    \leq(1+\|C^{-1}\|_{\mathrm{op}})Q
    \leq10(\Delta+1)^3.
\]

To control the trace term without introducing a factor of $n$, note
that a pure covariance $W$ with $u=(\tr W-2n)/4$ satisfies
\begin{equation}\label{eq:lfpr-pure-frobenius}
    \|W-I\|_F^2\leq16u^2+8u\leq16(u+1)^2.
\end{equation}
Indeed, its reciprocal eigenvalue pair $(\lambda_j,\lambda_j^{-1})$
contributes $16u_j^2+8u_j$, where
$u_j=(\lambda_j+\lambda_j^{-1}-2)/4\geq0$ and $\sum_j u_j=u$.
For $A=C^{-1}-I$ and $B=V-I$, purity and $h\geq1$ give
$\tr A=4b$, $\tr B\leq4h$,
$\|A\|_F\leq4(b+1)$, and $\|B\|_F\leq5h$.
Therefore
\[
    \tr(C^{-1}V)-2n
    =\tr A+\tr B+\tr(AB)
    \leq4b+4h+20(b+1)h
    \leq28(\Delta+1)^2h.
\]
Finally $\|\mathbf m\|_2^2\leq2h$, so
\[
    (\mathbf m-\mathbf d)^\top C^{-1}(\mathbf m-\mathbf d)
    \leq2\|C^{-1}\|_{\mathrm{op}}\|\mathbf m\|_2^2
       +2\mathbf d^\top C^{-1}\mathbf d
    \leq36(\Delta+1)^3h.
\]
Substitution into \eqref{eq:lfpr-relative-energy} gives
$\nu_H(\pi)\leq25(\Delta+1)^3h\leq C_{\rm tr}(\Delta+1)^3h$,
proving (i).
\end{proof}

\begin{proof}[Proof of (ii)]
With $U=GH^\dagger$, unitary invariance gives
\[
    \|P_G(k)P_H(k)\|_{\mathrm{op}}
       =\|\Pi_{\leq k}U\Pi_{\leq k}\|_{\mathrm{op}}.
\]
Suppose this is greater than $\mu$, and choose unit vectors
$\psi,\chi\in\operatorname{Im}(\Pi_{\leq k})$ with
$|\langle\psi,U\chi\rangle|>\mu$.
Write $(\mathbf d,C)$ for the moments of $U\ket0$ and $r_j$ for its
squeezes. We first bound the largest squeeze, then use the 
definition of $\Delta(G,H)$ to control the total compatibility cost.

For every real vector $x$ and every unit vector
$\varphi\in\operatorname{Im}(\Pi_{\leq k})$,
\begin{equation}\label{eq:lfpr-quadrature-moments}
    \langle\varphi,(x^\top\hat{\mathbf R})^2\varphi\rangle
       \lesssim k\|x\|_2^2.
\end{equation}
Indeed, a passive transformation sends
$x^\top\hat{\mathbf R}$ to $\|x\|_2\hat q_1$ and preserves the cutoff;
use $\hat q_1=(\hat{a}_1+\hat{a}_1^\dagger)/\sqrt2$ and
$\|\hat q_1\Pi_{\leq k}\|_{\mathrm{op}}
\lesssim\sqrt k$, since $k\geq1$.
We also need the following density bound: every unit quadrature has
probability density at most $C_H\sqrt k$ on a cutoff-$k$ state,
for a universal $C_H$. After a passive rotation and partial trace, its
density is at most $\sum_{j=0}^k|h_j(t)|^2$, where $h_j$ are normalized
Hermite functions. Mehler's identity gives, for $0<v<1$,
\[
    \sum_{j=0}^k|h_j(t)|^2
    \leq\frac{v^{-k}}{\sqrt{\pi(1-v^2)}}
        \exp\left(-\frac{1-v}{1+v}t^2\right).
\]
Taking $v=e^{-1/k}$ proves the bound. Scaling the quadrature by $s$
divides the density bound by $s$, and translation leaves it unchanged.

Let $s=e^{r_{\max}}$ and measure a unit quadrature corresponding to the
largest eigenvalue $s^2$ of $C$. Its densities $p,q$ in $\psi,U\chi$
satisfy $\int t^2p(t)\,dt\lesssim k$ and
$\|q\|_\infty\lesssim\sqrt k/s$.
Take $R=C_0\sqrt k/\mu$ for a sufficiently large universal $C_0$,
and let $I_R=[-R,R]$. Then
\[
    p(I_R^c)\leq c\mu^2,
    \qquad q(I_R)\lesssim k/(\mu s),
\]
where $c$ is a sufficiently small universal constant.
Fidelity monotonicity under this measurement and Cauchy--Schwarz give
\[
    \mu<|\langle\psi,U\chi\rangle|
       \leq\int\sqrt{pq}
       \leq\sqrt{q(I_R)}+\sqrt{p(I_R^c)}.
\]
For $s$ larger than a sufficiently large universal multiple of
$k\mu^{-3}$, the right-hand side is at most $\mu$, a contradiction. Thus
\begin{equation}\label{eq:lfpr-largest-squeeze}
    e^{r_{\max}}\lesssim k\mu^{-3}.
\end{equation}

Now set $P=\Pi_{\leq k}$, $B=UNU^\dagger$, $Q=UPU^\dagger$, and
$\Delta=\Delta(G,H)$. By the original definition,
$N+B\succeq\Delta I$. Since $P,Q$ have finite rank, there are unit vectors
$\phi\in\operatorname{Im}P$, $\eta\in\operatorname{Im}Q$ such that
\[
    P\eta=t\phi,\qquad Q\phi=t\eta,\qquad
    t=\|PQ\|_{\mathrm{op}}=\|PUP\|_{\mathrm{op}}>\mu.
\]
The operator $B$ commutes with $Q$, so
$\langle\phi,B\eta\rangle=t\langle\eta,B\eta\rangle$.
Expanding $\eta=t\phi+(I-P)\eta$ therefore gives
\[
    t\langle\phi,B\phi\rangle
       =t\langle\eta,B\eta\rangle
          -\langle(I-P)B\phi,(I-P)\eta\rangle.
\]
Since $\langle\phi,N\phi\rangle\leq k$ and
$\langle\eta,B\eta\rangle\leq k$, we obtain
\begin{equation}\label{eq:lfpr-variational-cutoff}
    \Delta\leq\langle\phi,(N+B)\phi\rangle
       \leq2k+\frac{\|(I-P)BP\|_{\mathrm{op}}}{t}
       \leq2k+\frac{\|(I-P)BP\|_{\mathrm{op}}}{\mu}.
\end{equation}
Thus it suffices to bound the part of $B$ that crosses the photon cutoff.

Write the quadratic Hamiltonian as
\[
    B=\hat{a}^\dagger M \hat{a}+\frac12
       \bigl(\hat{a}^{\dagger\top}\beta \hat{a}^\dagger+\hat{a}^\top\overline\beta \hat{a}\bigr)
       +\hat{a}^\dagger\ell+\ell^\dagger \hat{a}+c_0I,
       \qquad M=M^\dagger,\quad\beta=\beta^\top.
\]
Only the terms creating one or two photons contribute to $(I-P)BP$.
For $v$ supported on at most $k+2$ photons, Cauchy--Schwarz gives
\[
    \left\|\frac12\hat{a}^\top\overline\beta \hat{a}\,v\right\|_2
       \leq\frac12\|\beta\|_F
          \left(\sum_{i,j}\|\hat{a}_i \hat{a}_jv\|_2^2\right)^{1/2}
       \lesssim k\|\beta\|_F\|v\|_2,
\]
where the sum equals $\langle v,N(N-1)v\rangle$ and $k\geq1$.
Taking adjoints bounds the pair-creation term on $\operatorname{Im}P$;
the one-photon term satisfies
$\|\hat{a}^\dagger\ell P\|_{\mathrm{op}}\lesssim\sqrt k\|\ell\|_2$.
Consequently,
\begin{equation}\label{eq:lfpr-quadratic-cutoff}
    \|(I-P)BP\|_{\mathrm{op}}
       \lesssim k\|\beta\|_F+\sqrt k\|\ell\|_2.
\end{equation}

Expanding
$B=\tfrac12(\hat{\mathbf R}-\mathbf d)^\top C^{-1}
(\hat{\mathbf R}-\mathbf d)-nI/2$ in ladder operators gives
\[
    \|\beta\|_F^2=\sum_j\sinh^2(2r_j),\qquad
    \|\ell\|_2^2=\frac12\mathbf d^\top C^{-2}\mathbf d.
\]
For the first identity, a passive change of modes diagonalizes $C$;
the pair coefficient in mode $j$ then has magnitude $\sinh(2r_j)$.
The second follows from the linear term
$-\mathbf d^\top C^{-1}\hat{\mathbf R}$.
By \eqref{eq:lfpr-compatibility-formula} and
$s^{-2}I\preceq C\preceq s^2I$,
\begin{align*}
    \|\beta\|_F^2
       &=4\sum_j(\cosh r_j-1)(\cosh r_j+1)\cosh^2r_j
         \lesssim s^3\Delta,\\
    \|\ell\|_2^2
       &\leq\frac{s^4+s^2}{2}
          \mathbf d^\top(I+C)^{-1}\mathbf d
         \lesssim s^4\Delta.
\end{align*}
Substituting into \eqref{eq:lfpr-quadratic-cutoff} and then
\eqref{eq:lfpr-variational-cutoff} yields
\[
    \Delta\leq2k+\frac{C_*}{\mu}
       \bigl(ks^{3/2}+\sqrt k\,s^2\bigr)\sqrt\Delta
\]
for a universal constant $C_*$. Young's inequality and
\eqref{eq:lfpr-largest-squeeze} now give
\[
    \Delta\lesssim k+\mu^{-2}(k^2s^3+ks^4)
       \lesssim k^5\mu^{-14}\leq(k/\mu)^{32}.
\]
Choosing $C_{\rm sep}$ sufficiently large proves (ii) by contradiction.
\end{proof}

\noindent Part~(ii) limits how many bases survive pruning through the following packing bound.

\begin{lemma}\label{lemma:lfpr-packing}
Let $P_1,\ldots,P_s$ be orthogonal projections satisfying
$\tr(\rho P_i)\geq w$ and
$\|P_iP_j\|_{\mathrm{op}}\leq\mu$ for $i\ne j$, where $w>\mu$.
Then
\[
    s\leq\frac{1-\mu}{w-\mu}.
\]
\end{lemma}
\begin{proof}
For $x_i\in\operatorname{Im}(P_i)$,
$|\langle x_i,x_j\rangle|\leq\mu\|x_i\|_2\|x_j\|_2$, so
\[
    \left\|\sum_i x_i\right\|_2^2
       \leq(1+(s-1)\mu)\sum_i\|x_i\|_2^2.
\]
The map $(x_1,\ldots,x_s)\mapsto\sum_i x_i$ thus has squared norm at
most $1+(s-1)\mu$. Its product with its adjoint is $\sum_iP_i$, giving
\[
    sw\leq\tr\left(\rho\sum_iP_i\right)
       \leq\left\|\sum_iP_i\right\|_{\mathrm{op}}
       \leq1+(s-1)\mu.
\]
This proves the lemma.
\end{proof}

\noindent Finally, in a basis where $\pi_*$ has bounded photons, $\rho$ has large cutoff mass, so useful bases survive the mass test.

\begin{lemma}\label{lemma:lfpr-heavy}
If $F(\rho,\pi_*)\geq\alpha$ and $\nu_G(\pi_*)\leq h$ for $h\geq1$, then
one may take $C_0=64$ in the bound
\[
    k\geq C_0h/\alpha
    \quad\Longrightarrow\quad
    \tr(\rho P_G(k))\geq3\alpha/4.
\]
\end{lemma}
\begin{proof}
Write $\pi_*=\ketbra{\psi}$ and $P=P_G(k)$.
Markov's inequality gives
$\zeta=\|(I-P)\psi\|_2^2\leq h/k\leq c\alpha$,
where $c=C_0^{-1}=1/64$.
Since $0\preceq\rho\preceq I$,
\[
    \sqrt\alpha\leq\|\sqrt\rho\psi\|_2
    \leq\|\sqrt\rho P\psi\|_2+\|\sqrt\rho(I-P)\psi\|_2
    \leq\sqrt{\tr(\rho P)}+\sqrt\zeta.
\]
Here the squared first term is at most $\tr(P\rho P)$, and the second
is at most $\zeta$. Consequently
$\tr(\rho P)\geq(\sqrt\alpha-\sqrt\zeta)^2
\geq(1-\sqrt c)^2\alpha=49\alpha/64\geq3\alpha/4$.
\end{proof}

\subsubsection{The algorithm and its analysis}\label{sssec:lfpr-algorithm}

We now put everything together.
\paragraph{Parameters for one block.}
Set $q=11$ and $\mu=\alpha/16$. Recall the photon bound $h_\alpha$ of Lemma~\ref{lemma:lfpr-warm-frame} and the constants $C_{\rm tr}$, $d$ and $C_{\rm sep}$ of Proposition~\ref{prop:lfpr-frame-comparison}. For a squeezing bound $A\geq1$ at the start of a block, define

\begin{align}
    B(A)&=2\Lambda^2 A^{2^{-q}},
    &h(A)&=h_\alpha(B(A)),\label{eq:lfpr-block-bounds}\\
    k(A)&=\left\lceil64h(A)/\alpha\right\rceil,
    &T(A)&=C_{\rm sep}\left(\frac{k(A)}{\mu}\right)^d,\notag\\
    h^+(A)&=8C_{\rm tr}T(A)^3h(A),
    &A^+(A)&=6h^+(A).\notag
\end{align}

%Each quantity makes one of the earlier results applicable along a good branch. After $q$ rounds, the squeezing of $\pi_*$ is at most $B(A)$ by \eqref{eq:lfpr-block-unsqueeze}, so the fit gives photon number at most $h(A)$ by Lemma~\ref{lemma:lfpr-warm-frame}. The cutoff $k(A)$ is large enough for Lemma~\ref{lemma:lfpr-heavy}, and $T(A)$ is the separation threshold of Proposition~\ref{prop:lfpr-frame-comparison}(ii) at cutoff $k(A)$ and overlap $\mu$. If a basis is discarded because it is within threshold $T(A)$ of a kept basis, Proposition~\ref{prop:lfpr-frame-comparison}(i) bounds the photon number in the kept basis by $C_{\rm tr}(T(A)+1)^3h(A)\leq h^+(A)$, using $T(A)\geq1$. Finally, \eqref{eq:lfpr-energy-condition} turns this photon bound into the squeezing bound $4h^+(A)+2\leq A^+(A)$ for the next block, using $h^+(A)\geq1$.

For conciseness, throughout this section, we will refer to the list-deocdable learning algorithm in Lemma~\ref{lemma:lfpr-list} as the \emph{decoder} algorithm.
For an integer $b\geq1$, let $\operatorname{List}_b(G)$ be the union of $b$ independent runs of the decoder in the basis $G$, each on fresh copies.
Write $\operatorname{Fit}_A(G,\hat\mu,\hat\Sigma)$ for the fitted basis of the terminal pair $(\hat\mu,\hat\Sigma)$ in the basis $G$, computed with $B(A)$ in place of $A$ in \eqref{eq:lfpr-precision-domain}. Every cutoff-mass estimate in a block has additive error $\alpha/16$ and failure probability $\tau$. 
With this, we can now give pseudocode (see Algorithm~\ref{alg:lfpr-block}) for processing one block.

\begin{algorithm}[H]
\caption{One block of list-decodable unsqueezing}
\label{alg:lfpr-block}
\KwIn{Copies of $\rho$, $\alpha\in(0,1]$, a basis list $\mathcal L$ with $|\mathcal L|\leq W$, a bound $A\geq1$, repetitions $b\geq1$, and $\tau\in(0,1/2]$.}
\KwOut{A list $\mathcal L^+$ of at most $W$ bases, or failure.}
Set $\mathcal T_0=\mathcal L$\;
\For{$j=0,\ldots,q-1$}{
    Set $\mathcal T_{j+1}=\varnothing$\;
    \ForEach{$G\in\mathcal T_j$}{
        \ForEach{$(\hat\mu,\hat\Sigma)\in\operatorname{List}_b(G)$}{
            % Discard covariance candidates which are not symmetric positive definite.
            Compute a Williamson decomposition $2t_\alpha\hat\Sigma=S\widehat D S^\top$\;
            Append $U_{S^{-1}}G$ to $\mathcal T_{j+1}$\;
        }
    }
}
Set $\mathcal C=\varnothing$\;
\ForEach{$G\in\mathcal T_q$}{
    \ForEach{$(\hat\mu,\hat\Sigma)\in\operatorname{List}_b(G)$ with $\hat\Sigma\succ0$}{
        Append $\operatorname{Fit}_A(G,\hat\mu,\hat\Sigma)$ to $\mathcal C$\;
    }
}
\ForEach{$H\in\mathcal C$}{
    Estimate $\tr(\rho P_H(k(A)))$ by $\hat w_H$\;
}
Set $\mathcal L^+=\varnothing$\;
\ForEach{$H\in\mathcal C$ with $\hat w_H\geq\alpha/2$, in a fixed order}{
    \If{$\Delta(H,J)>T(A)$ for every $J\in\mathcal L^+$}{
        Append $H$ to $\mathcal L^+$\;
        \If{$|\mathcal L^+|>W$}{Return failure\;}
    }
}
Return $\mathcal L^+$\;
\end{algorithm}

The next lemma shows that a block preserves the invariant: if the input list contains a basis useful at level $A$, then the output list contains a basis useful at level $A^+(A)$, and the list stays short.

\begin{lemma}[One block step]\label{lemma:lfpr-block}
Suppose $F(\rho,\pi_*)\geq\alpha$ and the input list contains a basis in which the best Gaussian approximator $\pi_*$ has squeezing at most $A$. If the estimation lists contain good pairs and the cutoff-mass estimates in Algorithm~\ref{alg:lfpr-block} are accurate, the next list has at most $8/(3\alpha)$ bases and contains a basis $J$ with
\[
    \nu_J(\pi_*)\leq h^+(A),\qquad
    \kappa(V_J)\leq A^+(A).
\]
\end{lemma}
\begin{proof}
Follow a good list element at each of the $q$ levels. Equation~\eqref{eq:lfpr-block-unsqueeze}, obtained from Lemma~\ref{lemma:williamson-unsqueezing}, bounds $\pi_*$ squeezing at the resulting leaf by $B(A)$. A good terminal pair and Lemma~\ref{lemma:lfpr-warm-frame} then give a candidate $H_*$ with $\nu_{H_*}(\pi_*)\leq h(A)$.

By Lemma~\ref{lemma:lfpr-heavy}, its true cutoff mass is at least $3\alpha/4$, and its estimated mass is therefore at least $3\alpha/4-\alpha/16>\alpha/2$. It survives the mass test. If retained, it already has the required photon bound. If discarded during merging, some retained $J$ satisfies $\Delta(H_*,J)\leq T(A)$, and Proposition~\ref{prop:lfpr-frame-comparison} gives $\nu_J(\pi_*)\leq h^+(A)$. Equation~\eqref{eq:lfpr-energy-condition} gives the squeezing bound.
% \jerry{I can't find what equation this is supposed to reference}

Every retained basis has true cutoff mass at least $w=7\alpha/16$, and distinct retained bases have projector overlap at most $\mu=\alpha/16$. Lemma~\ref{lemma:lfpr-packing} bounds their number by
\[
    \frac{1-\mu}{w-\mu}\leq\frac8{3\alpha}.
\]
Thus the failure condition is never triggered on the stated success event.
\end{proof}

\paragraph{The full algorithm.}
The full algorithm runs $R$ blocks, starting from the list that contains only the Fock basis $G=I$, and updates the squeezing bound by $A\mapsto A^+(A)$ after each block. Given $\alpha,A_0,\eta$ as in Theorem~\ref{thm:lfpr-unsqueezing}, set

\begin{equation}\label{eq:lfpr-round-count}
    R=\left\lceil\log_2(e+\log A_0)\right\rceil,
    \qquad N_\alpha=13W L_\alpha^{12},
\end{equation}

and choose

\begin{equation}\label{eq:lfpr-boost-count}
    b=\left\lceil2\log(2N_\alpha R/\eta)+800\right\rceil,
    \qquad J=N_\alpha Rb^{12}.
\end{equation}

Set $\tau=\eta/(2J)$. Here $R$ is the number of blocks, $b$ is the number of decoder runs in each union of lists, and $J$ bounds the total number of unions of lists and cutoff-mass estimates over the whole run, as shown in the proof below. Algorithm~\ref{alg:lfpr-unsqueezing} applies $R$ blocks with these choices of $b$ and $\tau$.

% \paragraph{The full algorithm.}
% Given $\alpha,A_0,\eta$ as in Theorem~\ref{thm:lfpr-unsqueezing}, set
% \begin{equation}\label{eq:lfpr-round-count}
%     R=\left\lceil\log_2(e+\log A_0)\right\rceil,
%     \qquad N_\alpha=13W L_\alpha^{12},
% \end{equation}
% and choose
% \begin{equation}\label{eq:lfpr-boost-count}
%     b=\left\lceil2\log(2N_\alpha R/\eta)+800\right\rceil,
%     \qquad J=N_\alpha Rb^{12}.
% \end{equation}
% Set $\tau=\eta/(2J)$. Starting from the identity frame,
% Algorithm~\ref{alg:lfpr-unsqueezing} applies $R$ blocks with these choices
% of $b,\tau$ and updates the squeezing bound after each block.

\begin{algorithm}[H]
\caption{List-decodable iterative unsqueezing}
\label{alg:lfpr-unsqueezing}
\KwIn{Copies of $\rho$, $\alpha\in(0,1]$, an initial bound $A_0\geq1$, and $\eta\in(0,1/2]$.}
\KwOut{A list of at most $W$ bases, or failure.}
Set $\mathcal L_0=\{I\}$\;
\For{$r=0,\ldots,R-1$}{
    Run Algorithm~\ref{alg:lfpr-block} with $(\rho,\alpha,\mathcal L_r,A_r,b,\tau)$ to obtain $\mathcal L_{r+1}$\;
    \If{the block reports failure}{Return failure\;}
    Set $A_{r+1}=A^+(A_r)$\;
}
Return $\mathcal L_R$\;
\end{algorithm}

\begin{proof}[Proof of Theorem~\ref{thm:lfpr-unsqueezing}]

We analyze Algorithm~\ref{alg:lfpr-unsqueezing} in four steps: we count the random subroutine calls, show that with probability at least $1-\eta$ all of them succeed, follow a useful basis through the blocks on that event, and bound the cost.

\emph{Counting calls.} Each union of $b$ lists has at most $L_\alpha b$ pairs, so every basis has at most $L_\alpha b$ updated bases. Starting from at most $W$ bases, round $j$ of a block holds at most $W(L_\alpha b)^j$ bases. A block computes one union of lists for each basis in rounds $0,\ldots,q$, at most $(q+1)W(L_\alpha b)^q=12W(L_\alpha b)^{11}$ in total, and one cutoff-mass estimate for each of the at most $W(L_\alpha b)^{12}$ candidates. Over $R$ blocks, the total number of unions and mass estimates is therefore at most $13RW(L_\alpha b)^{12}=N_\alpha Rb^{12}=J$, including on runs that stop early with failure.

\emph{All calls succeed.} By Lemma~\ref{lemma:lfpr-list}, each union of $b$ lists fails to contain a good pair with probability at most $10^{-b}\leq e^{-b}$, and each mass estimate fails with probability at most $\tau=\eta/(2J)$. Both bounds hold conditionally on all earlier observations, because every call uses fresh copies, so a union bound over the at most $J$ calls applies. Our choice of $b$ gives $12\log b\leq b/2$, and therefore
\[
    Je^{-b}=N_\alpha Rb^{12}e^{-b}
       \leq N_\alpha Re^{-b/2}\leq\eta/2.
\]
So with probability at least $1-\eta$, every union of lists contains a good pair and every mass estimate is accurate.

\emph{Following a useful basis.} Condition on this event. The Fock basis $G=I$ is useful at level $A_0$. By Lemma~\ref{lemma:lfpr-block} and induction, every list $\mathcal L_r$ contains a basis useful at level $A_r$, and the algorithm never reports failure. It remains to bound $A_R$. Write $p=8$ and $d=32$ for the exponents in Lemma~\ref{lemma:lfpr-warm-frame} and Proposition~\ref{prop:lfpr-frame-comparison}, and set $v_\alpha=2^pu_\alpha\Lambda^{2p}\asymp t_\alpha^{2p+6}$. Substituting $B(A)=2\Lambda^2A^{2^{-q}}$ into $h_\alpha$ gives $h(A)=v_\alpha A^{p2^{-q}}$. Since $k(A)\lesssim h(A)/\alpha$ and $\mu\asymp\alpha$, we have $T(A)\lesssim\alpha^{-2d}h(A)^d$. Since $A^+(A)=48C_{\rm tr}T(A)^3h(A)$, we obtain, for a sufficiently large universal $C$,
\begin{equation}\label{eq:lfpr-contraction-constant}
    \mathsf K_\alpha
      =C\alpha^{-6d}v_\alpha^{3d+1}
      \asymp\alpha^{-s/2},\qquad s\asymp r_{\rm dec},
\end{equation}
and
\begin{equation}\label{eq:lfpr-root-contraction}
    A^+(A)\leq\mathsf K_\alpha A^{p(3d+1)2^{-q}}
       \leq\mathsf K_\alpha\sqrt A.
\end{equation}
Here $s\asymp r_{\rm dec}$ because $p$ and $d$ are universal and $t_\alpha\asymp\alpha^{-r_{\rm dec}}$. The last inequality uses $p(3d+1)2^{-q}=776/2048\leq1/2$; this is why we take $q=11$, the smallest value for which it holds. Taking logarithms, each block maps $\log A$ to at most $\frac12\log A+\log\mathsf K_\alpha$, so
\[
    \log A_r\leq2^{-r}\log A_0+2\log\mathsf K_\alpha.
\]
By \eqref{eq:lfpr-round-count}, $2^{-R}\log A_0\leq1$, so $A_R\leq H_\alpha$ with
\begin{equation}\label{eq:lfpr-terminal-constant}
    H_\alpha=e\mathsf K_\alpha^2
      \asymp\alpha^{-s},\qquad s\asymp r_{\rm dec}.
\end{equation}
Since $R\geq1$, the last block also gives a kept basis $G$ with
\[
    \nu_G(\pi_*)\leq h^+(A_{R-1})\leq A_R\leq H_\alpha.
\]
This proves \eqref{eq:lfpr-final-frame} and the success probability.

\emph{Cost.} Each union of lists uses at most $b\lceil e^{t_\alpha}(3n)^{t_\alpha}\rceil$ copies and $b$ decoder calls. Each mass estimate uses $O(\alpha^{-2}\log(J/\eta))$ copies, independently of its photon cutoff. The remaining computations, namely Williamson decompositions, projections onto $\mathcal K_A$, and evaluations of $\Delta$ by Lemma~\ref{lemma:lfpr-compatibility}, take polynomially many arithmetic operations in $n$ per candidate, and each candidate is compared with at most $W$ kept bases. Enlarging $C_{\rm dec}$ if necessary, we may assume that $t_\alpha$ exceeds the degree of these polynomials. 
% We count decoder calls separately from arithmetic operations, because the numerical runtime of the decoder also depends on bit-complexity parameters~\cite[Section~3.1 and Theorem~8.6]{ivkov2022list}. 
Since $L_\alpha=\lceil e^{t_\alpha}\rceil$, we have $N_\alpha\leq\exp(Ct_\alpha)$ and $b\leq Ct_\alpha\log(R/\eta)$. Multiplying these costs by the call bound $J=N_\alpha Rb^{12}$ shows that the number of copies, the number of decoder calls, and the number of arithmetic operations outside the decoder are each at most
\[
    \exp(Ct_\alpha)n^{t_\alpha}R\bigl(\log(R/\eta)\bigr)^{13}.
\]
Since $R\lesssim\log\log(e+A_0)$, this proves \eqref{eq:lfpr-preprocessing-cost}.

\end{proof}

\subsection{Optimization at bounded photon number}
\label{ssec:lf-bounded-energy}

The preceding subsection supplies a list containing a basis in which a best Gaussian approximator has total mean photon number at most $h$. We now construct an optimizer that competes with every pure Gaussian state satisfying this photon bound. The input $\sigma$ is arbitrary.

% \subsection{``Tensor optimization'' (AKA bullet point 4 in the above GPT outline)}

% In the above sections, we successfully rotated into a frame such that the best pure Gaussian approximation of our target state $\sigma$ has energy at most $h$. Thus, if we can approximately find the best pure Gaussian approximation of energy at most $h$ to $\sigma$, we are done. Our goal this section will be to show how to do this (in some sense) efficiently in time and copies.

\begin{theorem}\label{thm:lf-bounded-energy}
Let $\sigma$ be an $n$-mode state, $h\geq1$, $\epsilon \in(0,1/2]$, and $\eta\in(0,1/2]$. There is an algorithm which outputs a pure Gaussian state $\hat\pi$ satisfying
\begin{equation}\label{eq:lf-optimizer-guarantee}
    F(\sigma,\hat\pi)\geq\sup_{\pi\in\cG:\nu(\pi)\leq h}F(\sigma,\pi)-\epsilon
\end{equation}
with probability at least $1-\eta$, using at most
% \sitan{given we are assuming $h\ge 1$ and using $O(\cdot)$, can just write $h$ instead of $h+1$}
\begin{equation}\label{eq:lf-optimizer-cost}
    C_{h,\epsilon}n^{O(h^3+h\log(1/\epsilon))}\log(2/\eta)
\end{equation}
copies and runtime.
% \sitan{maybe just ``copies and runtime''? also, $C_{h,\epsilon}$ is hard to parse here as it references $C_{h,K,\delta}$ which appears a few pages after this}
Here $C_{h,\epsilon}$ is a constant depends only on $h$ and $\epsilon$.
% \sitan{is it essential to define $C_{h,\epsilon}, K_{h,\epsilon}$ here instead of later?}
% Here $\nu(\pi) = \tr(\pi N)$ where $N$ is the number operator. 
Moreover, the output has mean photon number $\nu(\hat\pi)$ at most $17(h+1)^3$.
\end{theorem}

\subsubsection{Setup}

Write $\pi_{b,A}=\ketbra{\psi(b,A)}$ for the pure Gaussian state in \eqref{eq:barg-state}. We optimize over the following convex set of Bargmann parameters:
\begin{equation}\label{eq:lf-parameter-domain}
    r_h=\sqrt{\frac{h}{h+1}},\qquad
    \cD_h=\left\{(b,A)\in\c P_n:
        \|A\|_{\rm op}\leq r_h,\quad
        \|A\|_F\leq\sqrt h,\quad
        \|b\|_2\leq2\sqrt h\right\}.
\end{equation}
Here $h\geq1$ and $A=A^\top$ by the definition of $\c P_n$. Since $r_h<1$, every point of $\cD_h$ defines a normalizable pure Gaussian state. We will show that this domain contains all pure Gaussian states with total photon number at most $h$, while its other points still have mean photon number bounded independently of $n$. 
% Throughout this argument the bounded quantity is $\nu(\pi)=\tr(\pi N)=E(\pi)-n/2$, with $N=\sum_{j=1}^n a_j^\dagger a_j$,

\begin{lemma}[A convex domain for best Gaussian approximators with bounded photon number]\label{lemma:lf-energy-domain}
For $h\geq1$, the set $\cD_h$ in \eqref{eq:lf-parameter-domain} is convex and consists of valid pure Gaussian parameters. Every pure Gaussian state $\pi_{b,A}$ with $\nu(\pi_{b,A})\leq h$ has $(b,A)\in\cD_h$. Conversely, every $(b,A)\in\cD_h$ satisfies
\begin{equation}\label{eq:lf-domain-photon-bound}
    \nu(\pi_{b,A})\leq h(h+1)+16h(h+1)^2\leq17(h+1)^3.
\end{equation}
\end{lemma}
\noindent We defer the proof to Appendix~\ref{app:proof-lf-energy-domain}.

We next study the normalization $Z(b,A)=\|\Phi(b,A)\|^2$ from \eqref{eq:barg-state}. Lemma~\ref{lemma:lf-normalization-formula} lets us compute $Z$ and its derivatives using real matrices; we now use it to bound $\log Z$ and its gradient on $\cD_h$.

\begin{lemma}[Convexity and bounds for the log normalization]\label{lemma:lf-log-normalization}
The function $L(b,A)=\log Z(b,A)$ is smooth and convex on the real convex domain $\{(b,A)\in\c P_n:\|A\|_{\rm op}<1\}$. For every $(b,A)\in\cD_h$,
\begin{equation}\label{eq:lf-normalization-bounds}
\begin{split}
    0\leq L(b,A)&\leq B_h\equiv17(h+1)^2,\\
    \|\nabla L(b,A)\|_{\c P}&\leq G_h\equiv
       \frac{6\sqrt h}{1-r_h}+\frac{8h}{(1-r_h)^2}.
\end{split}
\end{equation}
Here $B_h\asymp(h+1)^2$ and $G_h\asymp(h+1)^3$ for $h\geq1$.
The gradient uses the real inner product on $\c P_n$, whose norm is $\|(v_b,v_A)\|_{\c P}^2=\|v_b\|_2^2+\tfrac12\|v_A\|_F^2$.
% In particular, neither bound depends on the number of modes.
\end{lemma}
\begin{proof}
\emph{Convexity.} Think of \eqref{eq:lf-normalization-integral} as integration against the measure $d\mu(z)=\pi^{-n}e^{-\|z\|_2^2}d^{2n}z$ and write the exponent in the integral as as $\ell_\theta(z)=2\Re(b^\top z)+\Re(z^\top Az)$, where $\theta=(b,A)$. This is real-linear in $\theta$. For $0<t<1$, H\"older's inequality with exponents $1/(1-t)$ and $1/t$ gives
\[
    Z((1-t)\theta+t\theta')
       =\int (e^{\ell_\theta})^{1-t}(e^{\ell_{\theta'}})^t\,d\mu
       \leq Z(\theta)^{1-t}Z(\theta')^t.
\]
Convexity follows immediately. Smoothness follows from \eqref{eq:lf-normalization-formula}, since $M(A)$ is positive definite when $\|A\|_{\rm op}<1$.

\emph{Value bound.} Let $r=\|A\|_{\rm op}<1$. Lemma~\ref{lemma:lf-normalization-formula} implies that
\[
    L(b,A)=-\frac12\sum_j\log(1-s_j^2)+u^\top M^{-1}u.
\]
Both terms are nonnegative. Using $-\log(1-x)\leq x/(1-x)$ for $0\leq x<1$, $\|u\|_2=\|b\|_2$, and $\|M^{-1}\|_{\rm op}=(1-r)^{-1}$, we obtain
\begin{equation}\label{eq:lf-normalization-upper}
    0\leq L(b,A)\leq
       \frac{\|A\|_F^2}{2(1-r^2)}+\frac{\|b\|_2^2}{1-r}.
\end{equation}
On $\cD_h$, recall that $1/(1-r_h^2)=h+1$ and $1/(1-r_h)=(1+r_h)(h+1)\leq2(h+1)$ to bound the RHS by
\[
    \frac12h(h+1)+8h(h+1)\leq17(h+1)^2.
\]

\emph{Gradient bound.} We bound the directional derivative in an arbitrary real direction $v=(v_b,v_A)$ with $\|v\|_{\c P}=1$. Set
\[
    H=DM(A)[v_A]=\begin{pmatrix}-\Re v_A&\Im v_A\\ \Im v_A&\Re v_A\end{pmatrix},
    \qquad q=Du(b)[v_b]=\begin{pmatrix}\Re v_b\\-\Im v_b\end{pmatrix}.
\]
The block entries and the definition of the parameter norm imply
\[
    \tr H=0,\qquad
    \|H\|_F=\sqrt2\|v_A\|_F\leq2,\qquad
    \|q\|_2=\|v_b\|_2\leq1.
\]
Differentiating \eqref{eq:lf-normalization-formula}, using $D(M^{-1})[H]=-M^{-1}HM^{-1}$, gives
\begin{equation}\label{eq:lf-normalization-derivative}
    DL(b,A)[v]=
       -\tfrac12\tr(M^{-1}H)+2q^\top M^{-1}u
       -u^\top M^{-1}HM^{-1}u.
\end{equation}
Write $\delta=1-r_h>0$. On $\cD_h$, $\|M^{-1}\|_{\rm op}\leq\delta^{-1}$ and $\|u\|_2\leq2\sqrt h$. To bound the trace term without introducing a factor depending on $n$, use $\tr H=0$ to replace $M^{-1}$ by $M^{-1}-I$. Since the block formula for $M$ gives $\|I-M\|_F=\sqrt2\|A\|_F$,
\[
    \|M^{-1}-I\|_F=\|M^{-1}(I-M)\|_F
       \leq\frac{\sqrt{2h}}{\delta}.
\]
The Frobenius Cauchy--Schwarz inequality now yields
\[
    \tfrac12|\tr(M^{-1}H)|
       =\tfrac12|\tr((M^{-1}-I)H)|
       \leq\tfrac12\|M^{-1}-I\|_F\|H\|_F
       \leq\frac{\sqrt{2h}}{\delta}.
\]
The other two terms in \eqref{eq:lf-normalization-derivative} satisfy
\begin{align*}
    |2q^\top M^{-1}u|&\leq2\|q\|_2\|M^{-1}\|_{\rm op}\|u\|_2
       \leq\frac{4\sqrt h}{\delta},\\
    |u^\top M^{-1}HM^{-1}u|&\leq\|H\|_{\rm op}\|M^{-1}u\|_2^2
       \leq\frac{8h}{\delta^2},
\end{align*}
where $\|H\|_{\rm op}\leq\|H\|_F\leq2$. Adding these bounds and using $\sqrt2+4\leq6$ gives $|DL(b,A)[v]|\leq G_h$. Finally, $\|\nabla L(b,A)\|_{\c P}=\sup_{\|v\|_{\c P}=1}|DL(b,A)[v]|$, proving the gradient bound.
\end{proof}

\subsubsection{Learning the low degree components}

Let $P_K=\mathbf1_{\{N\leq K\}}=\Pi_{\leq K}$ project onto the Fock basis vectors with total photon number at most $K$. In Bargmann space, this keeps precisely the monomials of total degree at most $K$. We first choose a cutoff that approximates every best Gaussian approximator in $\cD_h$, then learn the finite block $\sigma_K=P_K\sigma P_K$ of the input state. The cutoff applies to the best Gaussian approximators; the input $\sigma$ can have arbitrary photon number.

\begin{lemma}[Uniform photon cutoff]\label{lemma:lf-uniform-cutoff}
Let $h\geq1$ and $0<\epsilon\leq1/2$. With $B_h=17(h+1)^2$ as in Lemma~\ref{lemma:lf-log-normalization}, set
\begin{equation}\label{eq:lf-cutoff-choice}
    s_h=\frac14\log(1+1/h),\qquad
    K=K_{h,\epsilon}=\left\lceil\frac{B_h+2\log(64/\epsilon)}{s_h}\right\rceil.
\end{equation}
Thus $s_h\asymp(h+1)^{-1}$ and $K=K_{h,\epsilon}\asymp(h+1)^3+(h+1)\log(1/\epsilon)$. For every $\theta=(b,A)\in\cD_h$, writing $\ket{\psi_\theta}=\ket{\psi(b,A)}$ and $\pi_\theta=\ketbra{\psi_\theta}$, we have
\begin{equation}\label{eq:lf-cutoff-tail}
    \|(I-P_K)\ket{\psi_\theta}\|^2\leq(\epsilon/64)^2.
\end{equation}
Consequently, for every density operator $\sigma$,
\begin{equation}\label{eq:lf-cutoff-error}
    \left|F(\sigma,\pi_\theta)
       -\bra{\psi_\theta}P_K\sigma P_K\ket{\psi_\theta}\right|
       \leq\epsilon/16.
\end{equation}
\end{lemma}
\begin{proof}
Our goal will be to prove
\begin{equation}\label{eq:lf-photon-exponential-moment}
    \bra{\psi_\theta}e^{s_hN}\ket{\psi_\theta}\leq e^{B_h}
    \qquad(\theta\in\cD_h).
\end{equation}
We first deduce the lemma from this inequality, then prove the inequality.

\emph{Proof of lemma using \eqref{eq:lf-photon-exponential-moment}}
Since $N$ has nonnegative integer eigenvalues, $I-P_K\preceq e^{-s_h(K+1)}e^{s_hN}$. Thus
\[
    t_\theta\equiv\|(I-P_K)\ket{\psi_\theta}\|^2
       \leq e^{-s_h(K+1)}\bra{\psi_\theta}e^{s_hN}\ket{\psi_\theta}
       \leq e^{B_h-s_h(K+1)}\leq(\epsilon/64)^2.
\]
To prove \eqref{eq:lf-cutoff-error}, let $\ket\phi=P_K\ket{\psi_\theta}$ and $\ket\chi=(I-P_K)\ket{\psi_\theta}$. Then
\[
    \bra{\psi_\theta}\sigma\ket{\psi_\theta}-\bra\phi\sigma\ket\phi
       =\bra\chi\sigma\ket{\psi_\theta}+\bra\phi\sigma\ket\chi.
\]
The magnitude of this is at most $\epsilon/16$.

\emph{Proof of \eqref{eq:lf-photon-exponential-moment}}
Multiplying each $m$-photon coefficient by $e^{sm/2}$ corresponds in Bargmann space to replacing $z$ by $e^{s/2}z$. Applied to $e^{b^\top z+z^\top Az/2}$, this replaces $(b,A)$ by $(e^{s/2}b,e^sA)$. Taking squared norms therefore gives
\begin{equation}\label{eq:lf-photon-mgf}
    \bra{\psi(b,A)}e^{sN}\ket{\psi(b,A)}
       =\frac{Z(e^{s/2}b,e^sA)}{Z(b,A)},
       \qquad e^s\|A\|_{\rm op}<1.
\end{equation}
The condition ensures that the scaled Bargmann function is square integrable, so this also justifies finiteness of the expectation.

Write $r=r_h=\sqrt{h/(h+1)}$ and take $s=s_h=-\tfrac12\log r$. For $\tilde b=e^{s/2}b$ and $\tilde A=e^sA$, the domain constraints give
\[
    \|\tilde A\|_{\rm op}\leq\sqrt r<1,\qquad
    \|\tilde A\|_F^2\leq h/r\leq h+1,\qquad
    \|\tilde b\|_2^2\leq4h/\sqrt r\leq4(h+1).
\]
Also $(1-r)^{-1}=(1+r)(h+1)\leq2(h+1)$ and $(1-\sqrt r)^{-1}=(1+\sqrt r)/(1-r)\leq4(h+1)$. The general normalization bound \eqref{eq:lf-normalization-upper}, which does not require the scaled parameters to lie in $\cD_h$, now yields
\[
    \log Z(\tilde b,\tilde A)
       \leq\frac{h+1}{2(1-r)}+\frac{4(h+1)}{1-\sqrt r}
       \leq17(h+1)^2=B_h.
\]
Since $Z(b,A)\geq1$, we are done.
\end{proof}

\noindent The cutoff space has dimension
\[
    D_K=\#\{\mathbf k\in\mathbb N^n:\textstyle\sum_j k_j\leq K\}
       =\binom{n+K}{K}\leq(n+1)^K.
\]

We learn $\sigma_K$ as an unnormalized positive matrix of trace at most one.

\begin{lemma}[Tomography of the cutoff block]\label{lemma:lf-cutoff-tomography}
Let $K\geq0$ be an integer, $0<\epsilon\leq1/2$, and $0<\eta\leq1/2$. Given copies of an arbitrary state $\sigma$, measurements on individual copies suffice to construct a matrix $\hat\sigma_K$ supported on the cutoff space such that, with probability at least $1-\eta$,
\[
    \hat\sigma_K\succeq0,\qquad \tr\hat\sigma_K\leq1,\qquad
    \|\hat\sigma_K-P_K\sigma P_K\|_{\rm op}\leq\epsilon/16.
\]
The number of copies is $O(D_K^4\epsilon^{-2}\log(2D_K/\eta))$.
\end{lemma}
\noindent We defer the proof to Appendix~\ref{app:proof-lf-cutoff-tomography}.

Together these lemmas give a uniformly accurate objective on $\cD_h$. Extend $\hat\sigma_K$ by zero outside the cutoff space and define
\begin{equation}\label{eq:lf-surrogate}
    p(\theta)=\bra{\Phi(\theta)}\hat\sigma_K\ket{\Phi(\theta)},
    \qquad f(\theta)=\frac{p(\theta)}{Z(\theta)}
       =\bra{\psi_\theta}\hat\sigma_K\ket{\psi_\theta}.
\end{equation}
On the success event, the triangle inequality and $\|P_K\ket{\psi_\theta}\|\leq1$ imply
\begin{equation}\label{eq:lf-surrogate-error}
    |f(\theta)-F(\sigma,\pi_\theta)|\leq\epsilon/8
    \qquad\text{for every }\theta\in\cD_h.
\end{equation}
Also $0\leq f\leq1$ because $\hat\sigma_K\succeq0$ and $\tr\hat\sigma_K\leq1$. Each Fock coefficient of $P_K\ket{\Phi(\theta)}$ is a Laurent polynomial in $(b,A)$ of degree at most $K$, by expanding the exponential in \eqref{eq:barg-state}. Thus $p$ is a real polynomial of degree at most $2K$ in the real and imaginary parts of the parameters. This is the finite polynomial numerator we will optimize, with $\log Z$ controlling the denominator.

% global optimization stuff goes here

\subsubsection{Global optimization of the normalized overlap}
\label{ssec:lf-global-optimization}

It remains to maximize $f=p/Z$ from \eqref{eq:lf-surrogate} over $\cD_h$. Write $\ell=\log Z$. By Lemma~\ref{lemma:lf-log-normalization}, $\ell$ is convex, $0\leq\ell\leq B_h$, and $\|\nabla\ell\|_{\c P}\leq G_h$ on this domain. Unfortunately, $p$ does not satisfy any nice concavity properties, and we are not aware of any local optimization methods that are useful for global optimization. Instead, we will find a subspace of bounded dimension that determines $p$ up to small error. We will then brute force over a fine-enough net of this subspace and use standard tools from convex optimization to minimize the $Z$ in the orthogonal subspace.

\begin{theorem}[Optimization of the normalized overlap]\label{thm:lf-normalized-optimizer}
Fix $h\geq1$, an integer $K\geq1$, and $0<\delta\leq1/2$. Given a matrix $\hat\sigma_K\succeq0$ supported on $P_K\c F_n$ with $\tr\hat\sigma_K\leq1$, there exists a deterministic algorithm that returns $\hat\theta\in\cD_h$ such that
\begin{equation}\label{eq:lf-normalized-optimizer}
    f(\hat\theta)\geq\max_{\theta\in\cD_h}f(\theta)-\delta.
\end{equation}
It uses at most $C_{h,K,\delta}(n+1)^{cK}$ real-arithmetic operations, for a universal constant $c$. One sufficient choice of the prefactor is
\begin{equation}\label{eq:lf-normalized-cost-constant}
\begin{split}
    \Gamma&=(1+5h)^K(K+1)^4K!,\\
    C_{h,K,\delta}
       &=\exp\!\left[
          C(K+1)^6\frac{\Gamma^2}{\delta^2}
          \log\!\left(\frac{C(h+1)(K+1)\Gamma}{\delta}\right)
          \right],
\end{split}
\end{equation}
for a sufficiently large universal $C$. These are sufficient,
nonoptimized choices for the constants in the stated algorithm.
\end{theorem}
% We count spectral factorizations and elementary functions at polynomial cost in the real-arithmetic model. The dependence of $C_{h,K,\delta}$ on its parameters need not be polynomial.

\paragraph{Reducing the numerator to a tensor polynomial.}
The next theorem gives the precise reduction we need. Its proof, including the coefficient-norm bound, is deferred to Appendix~\ref{app:lf-tensor-reduction}.

\begin{theorem}[Tensor representation and effective coordinates]\label{thm:lf-tensor-reduction}
Let $H=\C^n$ and let $\operatorname{Sym}^2H$ be the symmetric subspace of $H^{\otimes2}$ with its inherited norm. Identify a symmetric matrix with its coefficient tensor, so its tensor norm is its Frobenius norm. Put
\begin{equation}\label{eq:lf-tensor-scaling}
\begin{split}
    \mathcal X&=\C\oplus H\oplus\operatorname{Sym}^2H,
    \qquad R_h=\sqrt{1+5h},\\
    y(b,A)&=\frac{(1,b,A/\sqrt2)}{R_h},
    \qquad \Gamma=R_h^{2K}(K+1)^4K!.
\end{split}
\end{equation}
From $\hat\sigma_K$ one can construct a positive semidefinite operator $B$ on $\mathcal X^{\otimes K}$ with $\|B\|_F\leq\Gamma$ such that
\begin{equation}\label{eq:lf-tensor-numerator}
    p(\theta)=g(y(\theta)),\qquad
    g(v)=\langle v^{\otimes K},Bv^{\otimes K}\rangle.
\end{equation}
For $\theta,\theta'\in\cD_h$,
\[
    \|y(\theta)\|_2\leq1,\qquad
    \|y(\theta)-y(\theta')\|_2=\frac{\|\theta-\theta'\|_{\c P}}{R_h}.
\]
The polynomial $g$ is $2K\Gamma$-Lipschitz on the unit ball of $\mathcal X$. For every $0<\xi<\Gamma$, one can also construct a complex subspace $W\subseteq\mathcal X$ satisfying
\begin{equation}\label{eq:lf-effective-coordinates}
    \dim_{\C}W\leq8(K+1)^6(\Gamma/\xi)^2,\qquad
    \sup_{\|v\|_2\leq1}|g(v)-g(P_Wv)|\leq\xi.
\end{equation}
The constructions cost at most $C(K+1)^C K!(n+1)^{cK}$ arithmetic
operations for universal $C,c$, counting spectral factorizations in the
real-arithmetic model. The accuracy $\xi$ sets the singular-value
threshold and the dimension bound in \eqref{eq:lf-effective-coordinates};
it introduces no further factor in this construction cost.
\end{theorem}
\noindent Thus $p$ is a tensor quadratic form evaluated on the image $y(\cD_h)$. The approximation in \eqref{eq:lf-effective-coordinates} gives a subspace with dimension independent of $n$ such that the action of $B$ is essentially determined by its action on that subspace.

\paragraph{Choosing the denominator by convex optimization.}
Let $W$ and $\xi$ be the parameters from Theorem~\ref{thm:lf-tensor-reduction}, and let $\Lambda=\max\{1,2K\Gamma\}$. The following lemma formalizes the idea that to approximately maximize $f=p/Z$, it suffices to brute force over a net of the subspace $W$, while doing convex optimization to minimize the denominator at each point of the net. Let $[x]_+=\max\{0,x\}$.

\begin{lemma}[Search over numerator coordinates]\label{lemma:lf-ratio-search}
Set
\begin{equation}\label{eq:lf-penalty-parameters}
    t=\frac{\xi}{3\Lambda},\qquad \lambda=\frac{B_h+1}{t},\qquad
    J_z(\theta)=\ell(\theta)+\lambda[\|P_Wy(\theta)-z\|_2-t]_+.
\end{equation}
Let $\mathcal N$ be a finite $t$-net of the unit ball of $W$. For each $z\in\mathcal N$, let $\theta_z\in\cD_h$ be such that
\[
    J_z(\theta_z)\leq\min_{\theta\in\cD_h}J_z(\theta)+\zeta,
    \qquad 0<\zeta\leq1.
\]
Let
\[
    \hat z\in\operatorname*{argmax}_{z\in\mathcal N}f(\theta_z),
    \qquad \hat\theta=\theta_{\hat z}.
\]
Then
\begin{equation}\label{eq:lf-ratio-search-error}
    f(\hat\theta)\geq\max_{\theta\in\cD_h}f(\theta)
       -(3\xi+\zeta).
\end{equation}
Moreover, each $J_z$ is convex on $\cD_h$.
\end{lemma}
\begin{proof}
The domain $\cD_h$ is compact and $f$ is continuous, so it has a maximizer $\theta^*$. Choose $z\in\mathcal N$ with $\|P_Wy(\theta^*)-z\|_2\leq t$. This makes the penalty at $\theta^*$ zero, and hence
\[
    J_z(\theta_z)\leq J_z(\theta^*)+\zeta
       =\ell(\theta^*)+\zeta\leq B_h+1.
\]
Both terms of $J_z$ are nonnegative. Bounding each of them inidividually by $B_h+1$ yields
\begin{equation}\label{eq:lf-cell-controls}
    \ell(\theta_z)\leq\ell(\theta^*)+\zeta,\qquad
    \|P_Wy(\theta_z)-z\|_2\leq t+\frac{B_h+1}{\lambda}=2t.
\end{equation}
The projected coordinates of $\theta_z$ and $\theta^*$ are therefore at distance at most $3t$. Applying \eqref{eq:lf-effective-coordinates} at both points, followed by the Lipschitz bound for $g$, yields
\[
    |p(\theta_z)-p(\theta^*)|\leq2\xi+3\Lambda t=3\xi.
\]
Since $p\geq0$, the denominator bound in \eqref{eq:lf-cell-controls} implies
\begin{align*}
    f(\theta_z)&\geq e^{-\ell(\theta^*)-\zeta}p(\theta_z)\\
       &\geq e^{-\zeta}f(\theta^*)-3\xi e^{-\ell(\theta^*)-\zeta}
        \geq f(\theta^*)-\zeta-3\xi.
\end{align*}
Here we used $0\leq f\leq1$, $\ell\geq0$, and $1-e^{-\zeta}\leq\zeta$. By the choice of $\hat z$, $f(\hat\theta)\geq f(\theta_z)$, proving \eqref{eq:lf-ratio-search-error}.

$J_z$ is a sum of convex functions, so it is convex.
\end{proof}

\paragraph{Convex minimization.}
We use Theorem~3.2 of Bubeck~\citep{bubeck2015convex} as a black box in the following form. 

\begin{fact}[Projected subgradient descent]\label{fact:lf-convex-minimization}
Let $\mathcal H$ be a finite-dimensional real Hilbert space with inner product $\langle\cdot,\cdot\rangle_{\mathcal H}$ and induced norm $\|x\|_{\mathcal H}=\sqrt{\langle x,x\rangle_{\mathcal H}}$. Let $D\subseteq\mathcal H$ be compact and convex, with $0\in D$ and $\|\theta\|_{\mathcal H}\leq D_0$ on $D$, where $D_0>0$. Let $J:D\to\R$ be continuous and convex. Suppose an oracle returns, at every $\theta\in D$, a subgradient $v\in\partial J(\theta)$ with $\|v\|_{\mathcal H}\leq G$, where $G>0$; that is,
\[
    J(\theta')\geq J(\theta)+\langle v,\theta'-\theta\rangle_{\mathcal H}
    \qquad\text{for every }\theta'\in D.
\]
Suppose we also have an oracle for orthogonal projection onto $D$ in this Hilbert-space norm:
\[
    \operatorname{proj}_D(x)
       =\operatorname*{argmin}_{y\in D}\|x-y\|_{\mathcal H}^2.
\]
Given an accuracy $\zeta>0$, let $T=\max\{1,\lceil(D_0G/\zeta)^2\rceil\}, \alpha=\frac{D_0}{G\sqrt T}$.
Now consider the following algorithm. Initialize $\theta_0=0$, and consider the iterates $\theta_{j+1}=\operatorname{proj}_D(\theta_j-\alpha v_j)$, where $v_j$ is returned by the subgradient oracle. Then $\bar\theta=T^{-1}\sum_{j=0}^{T-1}\theta_j\in D$ satisfies
\[
    J(\bar\theta)\leq\min_DJ+\frac{D_0G}{\sqrt T}\leq\min_DJ+\zeta.
\]
Clearly, this algorithm requires $T$ subgradient and $T$ projection oracle queries.
\end{fact}

\noindent We apply this convex optimization algorithm to the functions $J_z(\theta)$ using the efficient subgradient and projection oracles in the following proposition.

\begin{proposition}[Oracles for the Gaussian optimization problem]\label{prop:lf-convex-oracles}
For every net point $z$, the function $J_z$ in \eqref{eq:lf-penalty-parameters} has a subgradient oracle on $\cD_h$ with norm bound $G_h+\lambda/R_h$. Projection onto $\cD_h$ in the parameter norm is also computable. Both operations have polynomial arithmetic cost in $n$ and $\dim_{\C}W$.
\end{proposition}
\noindent The proof and efficient oracle implementations are given in Appendix~\ref{app:lf-convex-oracles}.

We now have all the pieces in place to describe the full algorithm; see the description of Algorithm~\ref{alg:lf-normalized-overlap}. It takes in the estimate of $P_K \sigma P_K$, and outputs the parameters of a pure Gaussian state that approximately maximizes fidelity with $\sigma$.

\begin{algorithm}[H]
\caption{Global optimization of the normalized Gaussian overlap}
\label{alg:lf-normalized-overlap}
\KwIn{$h\geq1$, $K\geq1$, $0<\delta\leq1/2$, and $\hat\sigma_K\succeq0$ with $\tr\hat\sigma_K\leq1$ supported on $P_K\c F_n$.}
\KwOut{$\hat\theta\in\cD_h$ with $f(\hat\theta)\geq\max_{\cD_h}f-\delta$.}
Set $\xi=\delta/6$ and $\zeta=\delta/2$\;
Construct $B$, $y$, $\Gamma$, and the subspace $W$ at accuracy $\xi$ using Theorem~\ref{thm:lf-tensor-reduction}; set $g(v)=\langle v^{\otimes K},Bv^{\otimes K}\rangle$\;
Set $\Lambda=\max\{1,2K\Gamma\}$, $t=\xi/(3\Lambda)$, and $\lambda=(B_h+1)/t$\;
Construct a $t$-net $\mathcal N$ of the unit ball of $W$\;
Set $D_0=\sqrt{5h}$, $G=G_h+\lambda/R_h$, $T=\max\{1,\lceil(D_0G/\zeta)^2\rceil\}$, and $\alpha=D_0/(G\sqrt T)$\;
\ForEach{$z\in\mathcal N$}{
    Define $J_z(\theta)=\ell(\theta)+\lambda[\|P_Wy(\theta)-z\|_2-t]_+$ and initialize $\theta_0=0$\;
    \For{$j=0,\ldots,T-1$}{
        Compute $v_j\in\partial J_z(\theta_j)$ with $\|v_j\|_{\c P}\leq G$ using Proposition~\ref{prop:lf-convex-oracles}\;
        Set $\theta_{j+1}=\operatorname{proj}_{\cD_h}(\theta_j-\alpha v_j)$ using \eqref{eq:lf-domain-projection} and radial projection of $b$\;
    }
    Set $\theta_z=T^{-1}\sum_{j=0}^{T-1}\theta_j$\;
    Compute $f(\theta_z)=g(y(\theta_z))e^{-\ell(\theta_z)}$\;
}
Choose $\hat z\in\operatorname*{argmax}_{z\in\mathcal N}f(\theta_z)$ and return $\hat\theta=\theta_{\hat z}$\;
\end{algorithm}

It remains to analyze the algorithm.
\begin{proof}[Proof of Theorem~\ref{thm:lf-normalized-optimizer}]
We analyze Algorithm~\ref{alg:lf-normalized-overlap}. With its choices $\xi=\delta/6$ and $\zeta=\delta/2$, Lemma~\ref{lemma:lf-ratio-search} gives the desired error $3\xi+\zeta=\delta$, provided each convex minimization attains error $\zeta$. We verify this guarantee and the runtime below.

\emph{Constructing the net.}
Identify $W$ with $\R^q$, where $q=2\dim_{\C}W$. For $q>0$, set $m=\lceil2\sqrt q/t\rceil$, take the grid $\{-1+2j/m:0\leq j\leq m\}^q$, and project its points radially onto the unit ball. Every point in the ball is within $t/2$ of a grid point; radial projection moves that grid point by at most another $t/2$. This gives a $t$-net with
\[
    |\mathcal N|\leq(m+1)^q\leq(2+2\sqrt q/t)^q.
\]
If $W=\{0\}$, take $\mathcal N=\{0\}$. Explicitly, the choices in the
algorithm give
\begin{equation}\label{eq:lf-net-cost-parameters}
\begin{split}
    q&\leq Q_{h,K,\delta}
       :=\left\lceil576(K+1)^6\frac{\Gamma^2}{\delta^2}\right\rceil,
       \qquad t=\frac{\delta}{36K\Gamma},\\
    |\mathcal N|&\leq
       \left(2+\frac{72K\Gamma\sqrt{Q_{h,K,\delta}}}{\delta}\right)^{Q_{h,K,\delta}}.
\end{split}
\end{equation}

\emph{Applying the convex optimization black box.}
The domain $\cD_h$ is compact and convex, contains $0$, and has $\|\theta\|_{\c P}\leq\sqrt{5h}$. Lemma~\ref{lemma:lf-ratio-search} establishes convexity of $J_z$, which is also continuous. Proposition~\ref{prop:lf-convex-oracles} supplies the required subgradients and projections with $G=G_h+\lambda/R_h$. Thus Fact~\ref{fact:lf-convex-minimization}, applied with $D_0=\sqrt{5h}$ and the $T,\alpha$ in Algorithm~\ref{alg:lf-normalized-overlap}, gives
\[
    J_z(\theta_z)\leq\min_{\cD_h}J_z+\zeta.
\]
Every candidate therefore meets the premise of Lemma~\ref{lemma:lf-ratio-search}, using $O(T)$ oracle calls each. Here the iteration count is explicitly
\[
    T=\left\lceil\frac{20h}{\delta^2}
       \left(G_h+\frac{36K\Gamma(B_h+1)}{R_h\delta}\right)^2\right\rceil
       \asymp (h+1)^7\delta^{-2}+(h+1)^4K^2\Gamma^2\delta^{-4}.
\]

\emph{Objective evaluation and total cost.}
Evaluate $f(\theta_z)=p(\theta_z)e^{-\ell(\theta_z)}$ directly using the finite quadratic form for $p$ and \eqref{eq:lf-normalization-formula} for $\ell$. 
Constructing the tensor and $W$, and evaluating each numerator, costs
$C(K+1)^C K!(n+1)^{cK}$ operations. Each convex iteration has polynomial
cost in $n,q$. Multiplying by the net-size and iteration bounds above,
and absorbing their polynomial factors, gives the explicit prefactor
\eqref{eq:lf-normalized-cost-constant} and hence
$C_{h,K,\delta}(n+1)^{cK}$ operations for universal $c$, as claimed.
\end{proof}

\begin{proof}[Proof of Theorem~\ref{thm:lf-bounded-energy}]
Choose a large enough $K$ by \eqref{eq:lf-cutoff-choice} and perform Lemma~\ref{lemma:lf-cutoff-tomography} with accuracy parameter $\epsilon$ and failure probability $\eta$. On its success event, \eqref{eq:lf-surrogate-error} holds simultaneously for all $\theta\in\cD_h$, so optimizing on the estimated density matrix in the subspace $P_K$ suffices to approximately optimize the true fidelity with $\sigma$. Apply Theorem~\ref{thm:lf-normalized-optimizer} with $\delta=\epsilon/2$ to get $\hat\pi=\pi_{\hat\theta}$, which we claim is the approximately maximizing pure Gaussian state. Every $\pi$ with $\nu(\pi)\leq h$ has parameters in $\cD_h$ by Lemma~\ref{lemma:lf-energy-domain}, so
\begin{align*}
    F(\sigma,\hat\pi)
       &\geq f(\hat\theta)-\epsilon/8\\
       &\geq\max_{\cD_h}f-\epsilon/2-\epsilon/8
        \geq F(\sigma,\pi)-3\epsilon/4.
\end{align*}
Taking the supremum over $\pi$ with bounded photon number proves the fidelity guarantee. Moreover, $\nu(\hat\pi)\leq17(h+1)^3$.

Tomography uses $O(D_K^4\epsilon^{-2}\log(2D_K/\eta))$ copies and polynomial classical processing. Since $D_K\leq(n+1)^K$, its copy count is at most
$C(K+1)\epsilon^{-2}(n+1)^{4K+1}\log(2/\eta)$.
An explicit choice of the prefactor is
\begin{equation}\label{eq:lf-bounded-cost-parameters}
    C_{h,\epsilon}
       =C\bigl((K_{h,\epsilon}+1)\epsilon^{-2}
           +C_{h,K_{h,\epsilon},\epsilon/2}\bigr),
\end{equation}
where $C$ is universal and $C_{h,K,\delta}$ is specified in
\eqref{eq:lf-normalized-cost-constant}. The exponent of $n+1$ can be taken
to be $c_*K_{h,\epsilon}$ for a sufficiently large universal $c_*$.
Taking $K=K_{h,\epsilon}$ and adding the optimization cost
gives \eqref{eq:lf-optimizer-cost}.

\end{proof}

\subsection{Main result}
\label{ssec:lf-main-proof}

Now we state our main result for the low-fidelity regime. The proof is by combining Theorem~\ref{thm:lfpr-unsqueezing} and Therorem~\ref{thm:lf-bounded-energy}.

\begin{theorem}[Low-fidelity agnostic learning]\label{thm:lf-agnostic}\label{thm:lf-outline-target}
Let $\rho$ be an arbitrary $n$-mode bosonic state. Let $\opt=\inf_{\pi\in\cG}\infd(\rho,\pi)$, let $\pi_*\in\cG$ attain this optimum value, and let $E_0\geq3$ be an upper bound on $\nu(\pi_*)$.
Given $\epsilon\in(0,1]$ and $\delta\in(0,1/2]$, there is an algorithm which outputs $\hat\pi\in\cG$ such that, with probability at least $1-\delta$,
\[
    \infd(\rho,\hat\pi)\leq\opt+\epsilon,
\]
using at most
\[
    n^{\poly(1/\epsilon)}\cdot\poly\bigl(\log(1/\delta),\,\log\log(E_0)\bigr)
\]
copies of $\rho$. The runtime of our algorithm is
\[
    \exp\bigl(\exp(\poly(1/\epsilon))\bigr)\cdot n^{\poly(1/\epsilon)}\cdot\poly\bigl(\log(1/\delta),\,\log\log E_0\bigr).
\]
\end{theorem}

% We finish the agnostic guarantee by using the bounded-photon-number optimizer proved in Section~\ref{ssec:lf-bounded-energy}. This last step competes with the original comparator, and it does not require identifying which retained basis is useful.

\begin{proof}
Set $\alpha=\epsilon/4$, $A_0=4E_0+2$, and $W=\lceil8/(3\alpha)\rceil$. Apply Theorem~\ref{thm:lfpr-unsqueezing} with failure probability $\delta/3$. If it reports failure, return vacuum. Otherwise, for every returned $G$, apply Theorem~\ref{thm:lf-bounded-energy} to $G\rho G^\dagger$ with $h=H_\alpha$, accuracy $\epsilon/4$, and failure probability $\delta/(3W)$. Rotate the resulting state back by applying $G^\dagger$ to obtain $\tilde\pi_G$.

Estimate each $F(\rho,\tilde\pi_G)$ to additive error $\epsilon/4$ and failure probability $\delta/(3W)$ on fresh copies, and return the candidate with the largest estimate. For any pure Gaussian candidate $\tilde\pi_G=U\ketbra0U^\dagger$, applying $U^\dagger$ and testing vacuum has success probability $F(\rho,\tilde\pi_G)$. Hoeffding's inequality therefore gives the stated accuracy with $O(\epsilon^{-2}\log(6W/\delta))$ copies, without any proximity-to-vacuum assumption on $\rho$.

If $F(\rho,\pi_*)<\alpha$, every output satisfies $F(\rho,\hat\pi)\geq0>F(\rho,\pi_*)-\epsilon$. Otherwise, \eqref{eq:lfpr-energy-condition} gives $\kappa(V_*)\leq A_0$. With probability at least $1-\delta$, preprocessing, every optimizer call, and every validation succeed. Conditional failure bounds are valid because these routines use fresh copies. On this event a useful basis $G_*$ satisfies $\nu(G_*\pi_*G_*^\dagger)\leq H_\alpha$, and hence
\[
    F(\rho,\tilde\pi_{G_*})\geq F(\rho,\pi_*)-\epsilon/4.
\]
Selecting the largest estimate loses at most two validation errors, so
\[
    F(\rho,\hat\pi)\geq F(\rho,\tilde\pi_{G_*})-2(\epsilon/4)
      \geq1-\opt-3\epsilon/4\geq1-\opt-\epsilon.
\]
This proves the desired infidelity guarantee.

The preprocessing cost is \eqref{eq:lfpr-preprocessing-cost}, with $c_\alpha=C_{\rm dec}(16/\epsilon)^{r_{\rm dec}}$ and $H_\alpha$ given by \eqref{eq:lfpr-terminal-constant}. The $W=\lceil32/(3\epsilon)\rceil$ optimizer calls use in total at most
\[
    W C_{H_\alpha,\epsilon/4}
       (n+1)^{c_*K_{H_\alpha,\epsilon/4}}
       \log(6W/\delta)
\]
copies and arithmetic operations, using the constants from the proof of Theorem~\ref{thm:lf-bounded-energy}. Validation costs $O(W\epsilon^{-2}\log(6W/\delta))$ copies. To record the parameter dependence, set
\begin{equation}\label{eq:lf-cost-parameters}
\begin{split}
    K_\epsilon&=K_{H_\alpha,\epsilon/4},\qquad
       k_\epsilon=\max\{c_\alpha,c_*K_\epsilon\},\\
    C_\epsilon&=C\Bigl[C_\alpha
       +W\bigl(C_{H_\alpha,\epsilon/4}+\epsilon^{-2}\bigr)
          (1+\log(6W))\Bigr].
\end{split}
\end{equation}
Here $C_\alpha,c_\alpha$ are from \eqref{eq:lfpr-preprocessing-cost}, and
$K_{h,a},C_{h,a},c_*$ from \eqref{eq:lf-cutoff-choice} and \eqref{eq:lf-bounded-cost-parameters}.
The total arithmetic cost is bounded by
$C_\epsilon(n+1)^{k_\epsilon}$ times a universal polynomial in
$1+\log(1/\delta)$ and $1+\log\log(e+E_0)$. Since
$H_{\epsilon/4}\asymp\epsilon^{-s}$ with $s\asymp r_{\rm dec}$, we have
\[
    K_\epsilon\asymp k_\epsilon\asymp\epsilon^{-3s},\qquad
    C_\epsilon\leq
       \exp\!\left(\exp\!\left[C\epsilon^{-3s}\log(2/\epsilon)\right]\right).
\]
For the copy count, the $W$ optimizer calls only require the tomography
cost from Lemma~\ref{lemma:lf-cutoff-tomography}, totaling
\[
    C W(K_\epsilon+1)\epsilon^{-2}
       (n+1)^{4K_\epsilon+1}\log(6W/\delta).
\]
Adding preprocessing and validation, and absorbing their
$\exp(\poly(1/\epsilon))$ factors into $(n+1)^{\poly(1/\epsilon)}$,
gives the stated copy and decoder-call bounds. The net search contributes
only arithmetic operations. These choices have no dependence on $\opt$;
it enters only the target fidelity $1-\opt-\epsilon$.
Substituting $A_0=4E_0+2$ proves the complexity bounds in
Theorem~\ref{thm:lf-outline-target}.
\end{proof}

% \newpage

% \subsection{List-decodable heterodyne tomography}
% \alvan{list-decodable learning of a Gaussian, then turn this into a low-fidelity robust heterodyne tomography primitive that spits out a list, and u can estimate fidelity over the list with non-gaussian measurements to get a $n\Otilde(\opt)$ guarantee or something.}

% \subsection{Algorithm}
% \alvan{but, to get $\opt+\epsilon$ we need to port the iterative unsqueezing and maybe pure state projection? to use lists, build a tree, etc. unsure if we should just introduce the main algorithm with pruning or the non-pruning one first then talk about pruning. i'm preferencing the former.}

\section{Agnostic learning with Gaussian measurements}
\label{sec:gaussian-measurements}
In Section~\ref{sec:hf}, we showed that learning in the high-fidelity regime is naively possible to $\Otilde(\opt)+\epsilon$ error by embedding robust estimators into the optimal tomography protocol of \citep{chen2026towards}; we then gave an algorithm to iteratively refine this warm start estimate to $\opt+\epsilon$ error using on-off photon detection or alternatively use more complicated non-Gaussian measurements which are harder to experimentally implement but enable us to achieve the same sample complexity as tomography itself, up to log factors. This leaves us, however, with two questions:

\begin{center}
    \textit{Are these non-Gaussian measurements necessary? \\Can we do better than $\Otilde(\opt)+\epsilon$ using only Gaussian measurements? }
\end{center}

In this section, we answer both in the affirmative, finding that Gaussian measurements are capable of refining the warm start estimate\alvantodo{kinda a lie cuz it doesn't properly use the warm start estimate} to $O(\opt)+\epsilon $ error but that $\opt+\epsilon$ is only efficiently achievable using non-Gaussian measurements.

\subsection{Refining to $O(\opt)+\epsilon $ with Gaussian measurements}
We first state our formal guarantee for the task of agnostically learning pure bosonic Gaussian states in infidelity, using only Gaussian measurements.

\begin{theorem}\label{thm:hf-g}
   Let $\rho$ be an $n$-mode state where $\infd(\rho,\ketbra{\psi})\leq \opt\leq \overline{\opt}$ for pure Gaussian state $\ket{\psi}$ with covariance $V$. Then, given $\epsilon > 0$ and $\delta \in (0,1/2)$, there exists an algorithm takes $\Otilde (n^2\epsilon^{-1}\log\log \|V^{-1}\|_{\rm op}\log^5 (1/\delta))$ copies of $\rho$ and uses only Gaussian measurements to output a pure Gaussian state $\ket*{\hat \psi}$ satisfying
   \begin{equation}\label{eq:hf-g-guarantee}
       \infd(\rho,\ketbra*{\hat \psi})\leq O(\opt)+\epsilon
   \end{equation}
   with probability at least $1-\delta$. 
\end{theorem}
\begin{proof}
We achieve this by making a petite modification to the final estimation step of the naive warm start algorithm which proceeds by iteratively unsqueezing the state, estimating it via robust heterodyne tomography, and then projecting the estimate to a pure Gaussian state. As before, WLOG, we ignore the unsqueezing steps and assume $\|V^{-1}\|_{\rm op}\leq 2$ so that we may use the notation of the original state for simplicity along with the fact that $\frac 1 2 I\preceq V\preceq 2I$. \alvantodo{cite}

Now, recall that in the final estimation step of the naive warm start algorithm, we use $\trd(\rho,\ketbra{\psi})\leq \sqrt{\opt}$ from Fuchs-van de Graaf to view the outcome distribution, $P$, of heterodyning $\rho$ as being $\sqrt{\opt}$-close in total variation distance to $Q\equiv \cN(\mu_h, \Sigma_h)$ such that we can apply the robust Gaussian learner in Proposition~\ref{prop:robust-gaussian-learner} to learn $(\mu_h,\Sigma_h)$ to absolute Frobenius error $\Otilde(\sqrt{\opt})$. By Lemma~\ref{lemma:frob-h-to-trd}, this means that the pure state projection is $\Otilde(\sqrt{\opt})$-close in trace distance to $\ketbra{\psi}$, and by the manipulation in \eqref{eq:tr-to-F-transfer}, this gives an estimate which is $\Otilde(\opt)$-close in infidelity to $\rho$. Thus, by the same machinery, proving Theorem~\ref{thm:hf-g} can be reduced to learning $(\mu_h,\Sigma_h)$ to absolute Frobenius error $O(\sqrt{\opt})$ from the heterodyne outcomes. 

We begin by introducing a different classical robust estimation primitive. 
\begin{proposition}[Proposition 1.5 of \citep{diakonikolas2020outlier}] \label{prop:robust-bounded-cov-learner}\alvantodo{im intentionally ignoring the corruption budget stuff in DKP20}
    Let $\c D $ be a distribution over $\R^d$ with mean $\mu$ and bounded covariance $\Sigma\preceq BI$ for some known constant $B>0$. Then, there exists a polytime algorithm which given $N=\Otilde (\epsilon^{-2}d \log(1/\delta))$ $\eta$-corrupted samples from $\c D$ outputs $\hat \mu$ such that 
    \begin{equation}\label{eq:robust-bounded-cov-learner-guarantee}
        \|\hat \mu -\mu\|_2\leq O(\sqrt{\eta})+\epsilon
    \end{equation}
    with probability at least $1-\delta$. 
\end{proposition}
\alvantodo{talk about strong contamination}

\noindent At first glance, Proposition~\ref{prop:robust-bounded-cov-learner} seems worse than Proposition~\ref{prop:robust-gaussian-learner}: $\tvd(P,Q)\lesssim \sqrt{\opt}$ means that taking $\c D=Q$ and $\eta=\sqrt{\opt}$ gives only $O(\opt^{1/4})$ absolute error in estimating the mean of the heterodyne distribution with this learner. However, unlike Proposition~\ref{prop:robust-gaussian-learner}, Proposition~\ref{prop:robust-bounded-cov-learner} does not require $\c D $ to be Gaussian. The key insight is then to write $P$ as some $O(\opt)$-corruption of a distribution $\c D$ where instead of taking $\c D=Q $, we allow $\c D$ to be a non-Gaussian distribution but require the mean of $\c D$ to be $O(\sqrt \opt) $-close to the mean of $ Q $. Then, by triangle inequality and the guarantee in \eqref{eq:robust-bounded-cov-learner-guarantee}, the estimated mean is $O(\sqrt \opt)$-close to the mean of $\c D$ and also the mean of $Q $. 

We first recall that $\infd(\rho,\ketbra{\psi})\leq \opt$ implies not just $\sqrt{\opt}$-closeness in total variation distance but also $\opt$-closeness in squared Hellinger distance:\alvantodo{maybe prelim?}
\begin{align}\label{eq:F-to-H2}
    d_H^2(P,Q)=1-\int \sqrt{p(x)q(x)}\; dx \leq 1-\sqrt{F(\rho,\ketbra{\psi})}\leq 1-\sqrt{1-\opt}\leq  \opt
\end{align}
where the first inequality uses the monotonicity of fidelity under measurement and where the last inequality uses the elementary fact $1-\sqrt{1-x}\leq x$ on $[0,1]$. The main technical lemma is then the following conversion from Hellinger closeness to the target Gaussian and Huber closeness to a perhaps non-Gaussian distribution that is close in mean to the original target distribution.

\begin{lemma}\label{lemma:hel-to-huber}
Let $P,Q$ be distributions over $\R^d$ such that $d_H^2(P,Q)<h$ for universal constant $h$, with $Q$ having mean $\mu$ and bounded covariance $\Sigma\preceq BI$ for known constant $B>0$. Then, there exists $\eta \lesssim d_H^2(P,Q)$ and distributions $\c D,\c R$ such that
\begin{equation}\label{eq:hel-to-huber-guarantee}
    P=(1-\eta)\c D+\eta R,  \quad\|\mu_{\c D}-\mu\|_2\lesssim d_H(P,Q),\quad  \|\Sigma_{\c D}\|_{\rm op}\leq B'
\end{equation}
where $\mu_{\c D},\Sigma_{\c D}$ are the mean and covariance of $\c D$ and $B'$ is a known function of $B$. 
\end{lemma}
\begin{proof}
We form $\c D$ by taking $P$ and removing portions where the density significantly exceeds the density of $Q$. Specifically, fixing an $\alpha \geq 4$, let $S=\{x:p(x)>\alpha q(x)\}$ be the subset of the support we will remove, and decompose $P=(1-\eta)\c D+\eta R$ with $\c D =P(\cdot |S^c)$, $\c R=P(\cdot | S)$, and $\eta=\Pr_P[x\in S]$. Then, 
\begin{align}\label{eq:hel-to-huber-eta-dh}
    \eta=\int_S p(x)dx \leq \alpha \int_S (\sqrt{p(x)}-\sqrt{q(x)})^2dx\leq 2\alpha d_H^2(P,Q)
\end{align}
where the first inequality uses the fact that $(\sqrt{p(x)}-\sqrt{q(x)})^2\geq (1-\alpha^{-1/2})^2p(x)\geq p(x)/\alpha$ on $S$ for $\alpha \geq 4$ and where the second inequality uses the definition of Hellinger distance. Now, set $h=1/(4\alpha)$ such that $\eta \leq 1/2$.

We now bound the difference in means. Let $u\in \R^d$ be a unit vector. Then, 
\begin{align}
    |\langle u,\mu_{\c D}-\mu\rangle |&=\frac{1}{1-\eta}\left|\int_{S^c}\langle u,x-\mu\rangle p(x)dx \right|=\frac{1}{1-\eta}\left|\int_{S^c}\langle u,x-\mu\rangle (p(x)-q(x))dx-\int_S \langle u,x-\mu\rangle q(x)dx\right|\\
    &\leq \frac{1}{1-\eta } \left(\int_{S^c}\langle u,x-\mu\rangle^2 (\sqrt{p(x)}+\sqrt{q(x)})^2dx\right)^{1/2}\left(\int_{S^c}(\sqrt{p(x)}-\sqrt{q(x)})^2dx\right)^{1/2}\\
    &\quad +\frac{1}{1-\eta}\bigg(\int_S\langle u,x-\mu\rangle^2 q(x)dx\bigg)^{1/2}\left(\int_S q(x)dx\right)^{1/2}\\
    &\leq \frac{1}{1-\eta}\cdot (\sqrt{\alpha}+1)\left(\int_{S^c}\langle u,x-\mu\rangle^2q(x)dx\right)^{1/2}\cdot \sqrt{2}d_H(P,Q)+\frac{1}{1-\eta}\sqrt{B\Pr_Q[x\in S]}\\
    &\leq \frac{\sqrt{B}}{1-\eta}\left((\sqrt{\alpha}+1)\sqrt{2} d_H(P,Q)+\sqrt{\eta/\alpha}\right)\\
    &\lesssim d_H(P,Q)
\end{align}
where the first line uses the fact that $\int_{\R^d}\langle u,x-\mu\rangle q(x)dx=0$, the second line uses triangle inequality and Cauchy-Schwarz, and the final line uses $\eta\lesssim d_H^2(P,Q)$ from \eqref{eq:hel-to-huber-eta-dh}. Since this holds for any unit vector $u\in \R^d $, by the variational characterization of the $\ell_2$ norm, $\|\mu_{\c D}-\mu\|_2 \lesssim d_H(P,Q)$ as desired.

We finally show that that $\c D$ has bounded covariance as well. For unit vector $u$, 
\begin{align}
    u^\top \Sigma_{\c D}u=\E_{\c D}[\langle u,x-\mu_{\c D}\rangle^2 ]\leq \E_{\c D}[\langle u,x-\mu\rangle ^2]\leq \frac{\alpha}{1-\eta}\E_{Q}[\langle u,x-\mu\rangle^2]\leq 2\alpha B\equiv B'(B)
\end{align}
where the first inequality uses the fact that the mean minimizes the MSE and where the second inequality uses the fact that $p_{\c D}(x)=\frac{p(x)\mathbbm{1}[x\in S^c]}{1-\eta}\leq\frac{ \alpha }{1-\eta}q(x) $.
Then, taking the supremum over all such $u$, $\|\Sigma_{\c D}\|_{\rm op}\leq B'(B)$ as well. 
\end{proof}

\noindent Collating Lemma~\ref{lemma:hel-to-huber} and Proposition~\ref{prop:robust-bounded-cov-learner} together gives the following mean estimator for bounded covariance distributions that is robust to Hellinger misspecification. 
\begin{corollary}\label{cor:robust-mean-H2-bounded-cov}
   Let $P,Q$ be distributions over $\R^d $ such that $d_H^2(P,Q)\leq \eta\leq \overline{\eta}$ for some universal constant $\overline{\eta}>0$, with $Q$ having mean $\mu$ and bounded covariance $\Sigma \preceq BI$ for known constant $B>0$. Then, there exists a polytime algorithm which given $N=\Otilde(\epsilon^{-2}d\log(1/\delta))$ samples from $P$ outputs $\hat \mu$ such that $$\|\hat \mu-\mu\|_2\leq O(\sqrt{\eta})+\epsilon$$ with probability at least $1-\delta$. 
\end{corollary}
\begin{proof}
Applying Lemma~\ref{lemma:hel-to-huber}, Proposition~\ref{prop:robust-bounded-cov-learner} and identifying $\cD$, we have a $\hat \mu$ such that:
\begin{align}
    \|\hat \mu-\mu\|_2\leq \|\hat \mu-\mu_{\c D}\|_2+\|\mu_{\c D}-\mu\|_2 \leq (O(\sqrt{\eta })+\epsilon) +\sqrt{\eta}=O(\sqrt{\eta})+\epsilon
\end{align}
where the first term is from \eqref{eq:robust-bounded-cov-learner-guarantee} and the second term is from \eqref{eq:hel-to-huber-guarantee}. 
\end{proof}
\noindent When $Q$ is Gaussian, we can also learn the covariance with a touch more work. 
\begin{corollary}\label{cor:robust-H2-bounded-cov-gaussian-learner}
    Let $P,Q$ be distributions over $\R^d $ such that $d_H^2(P,Q)\leq \eta\leq \overline{\eta}$ for some universal constant $\overline{\eta}>0$, with $Q=\c N(\mu,\Sigma)$ having bounded covariance $\Sigma \preceq BI$ for known constant $B>0$. Then, there exists a polytime algorithm which given $N=\Otilde(\epsilon^{-2}d^2\log(1/\delta))$ samples from $P$ outputs $\hat \mu,\hat \Sigma$ such that $$\|\hat \mu-\mu\|_2\leq O(\sqrt{\eta})+\epsilon,\quad \|\hat \Sigma-\Sigma\|_F\leq O(\sqrt \eta)+\epsilon$$ with probability at least $1-\delta$.   
\end{corollary}
\begin{proof}
    Corollary~\ref{cor:robust-mean-H2-bounded-cov} uses $\Otilde(\epsilon_1^{-2}d\log(1/\delta))$ samples to give a $\hat \mu$ such that $\|\Delta\|_2\leq O(\sqrt \eta)+\epsilon_1$, where $\Delta\equiv \hat \mu-\mu$. Now, taking fresh samples $X_i$, we form $Y_i=(X_i-\hat \mu)(X_i-\hat \mu)^\top$, which has the distribution $Q_Y$ or $P_Y$ if $X_i\iid Q$ or $X_i\iid P$ respectively. Viewing $Y_i$ as a vector such that the $\ell_2$ guarantee corresponds with the Frobenius norm, we claim that Corollary~\ref{cor:robust-mean-H2-bounded-cov} gives an estimate $\hat \Sigma$ with $\|\hat \Sigma-\Sigma\|_F\leq O(\sqrt \eta)+\epsilon$. 

     First, by data processing, we have that 
    \begin{align}
        d_H^2(P_Y,Q_Y)\leq d_H^2(P,Q)\leq \eta
    \end{align}
    as desired. Then, we can view $X_i=\mu+Z$ for $Z_i\iid \cN(0,\Sigma)$, meaning $\E_{Q_Y}[Y]=\E[(Z-\Delta)(Z-\Delta)^\top]=\Sigma +\Delta\Delta^\top$. Assuming for now that the conditions of Corollary~\ref{cor:robust-mean-H2-bounded-cov} are valid,
    \begin{align}
        \|\hat \Sigma-\Sigma\|_F\leq \|\hat \Sigma-\E_{Q_Y}[Y]\|_F+\|\Delta \Delta^\top \|_F\leq (O(\sqrt{\eta})+\epsilon_1) +\|\Delta \|_2^2\leq O(\sqrt{\eta})+\epsilon_1+(O(\sqrt{\eta})+\epsilon_1)^2\leq O(\sqrt{\eta})+\epsilon
    \end{align}
    where both calls of Corollary~\ref{cor:robust-mean-H2-bounded-cov} use slack $\epsilon_1$, which we define such that the final slack is $\epsilon$.  Thus, it suffices to show that $Q_Y$ has bounded covariance. We do this by showing that ${\rm Var}_{Q_Y}(\langle M,Y\rangle_F)\leq B$ for all $M\in \R^{d\times d}$ with $\|M\|_F=1$, for some $B>0$. Since $Y$ is symmetric, it suffices to consider symmetric $M$. Then, by standard moment identities regarding multivariate Gaussians, specifically (379) of \citep{petersen2008matrix}, 
    \begin{align}
        {\rm Var}_{Q_Y}(\langle M,Y\rangle_F)={\rm Var}_{Q_Y}((Z-\Delta)^\top M(Z-\Delta))= 2\tr(M\Sigma M\Sigma )+4\Delta^\top M \Sigma M\Delta &\leq (2B^2+4B\|\Delta\|_2^2)\cancel{\|M\|_F^2}\\&\leq 2B^2+4B
    \end{align}
    where we use the bounds $\Sigma \preceq BI$ and $\|\Delta\|_2^2\leq 1$ in the final two inequalities. Thus, $Q_Y$ has bounded covariance, meaning we are done. 
\end{proof}

\noindent Thus, \eqref{eq:F-to-H2} and Corollary~\ref{cor:robust-H2-bounded-cov-gaussian-learner} give estimators $(\hat \mu_h,\hat \Sigma_h)$ that resolve $(\mu_h,\Sigma_h)$ to additive absolute error $O(\sqrt{\opt})+\epsilon$, meaning we are done.
\end{proof}

\subsection{Hardness for Gaussian measurements}\label{ssec:gaussian-measurement-hardness}
We now show that non-Gaussian measurements are info-theoretically necessary for achieving $\opt+\epsilon$ error. 

\begin{theorem}\label{thm:gaussian-measurement-hardness}
For universal constant $c>0$, sufficiently large $n$, and $\eta\in (0,1/4]$, there exist two $(n+1)$-mode states $\rho_{+},\rho_-$ satisfying 
$$\opt_{+}=\opt_-=O (\eta),\quad \trd(\rho_{\pm},\ket{0}^{\otimes (n+1)})\leq \eta$$
where $\opt_{\pm}\equiv \inf_{\ket \psi\in \c G}\infd(\rho_{\pm},\ket{\psi})$, such that any algorithm using only Gaussian measurements that outputs $\ket*{\hat \psi_{\pm }}$ with the guarantee $\infd(\rho_{\pm},\ket*{\hat \psi_{\pm}})\leq \opt_{\pm }+c\eta^2 $ with probability at least $2/3$ must have copy complexity $N\geq \eta^{-1}\exp(\Omega(n^{1/3}))$.
\end{theorem}

\noindent Since Gaussian measurements cannot distinguish states better than their Wigner functions allow, we construct $\rho_+,\rho_-$ with extremely close Wigner functions for which the learner is forced to output different pure Gaussian states. We do this by first constructing two Wigner-close separable bosonic states $\sigma_+,\sigma_-$ whose vacuum probabilities differ significantly and then coupling this difference to a coherent direction in a new mode. 

We start by importing the following moment matching construction.
\begin{lemma}\label{eq:lb-match-wu}
   For every integer $k\geq 1$, there exist random variables $U_+,U_-$ supported on $\{0\}\cup [1,Ck^2]$ with the following properties:
   \begin{align}
       \E[U_\pm]=1,\quad \E[U_+^j]=\E[U_-^j] \quad (j\in [k]), \quad \Pr[U_+=0]-\Pr[U_-=0]\geq c
   \end{align}
   where $C,c>0$ are universal constants.
\end{lemma}
\noindent We defer the proof to Appendix~\ref{app:proof-lb-match-wu}.

We now use $U_+,U_-$ to construct $\sigma_+,\sigma_-$ such that the moment matching gives Wigner-closeness while the difference in probability mass on zero gives a significant difference in vacuum probability. 

\begin{lemma}\label{lemma:gm-hidden-vacuum}
There exists $c>0$ such that for every sufficiently large $n$, there exist $n$-mode states $\sigma_+,\sigma_-$ that are separable, diagonal in the Fock basis and have nonnegative Wigner functions such that:
\begin{align}\label{eq:gm-hidden-vacuum}
    \bra{0}^{\otimes n}\sigma_+\ket{0}^{\otimes n}-\bra{0}^{\otimes n}\sigma_-\ket{0}^{\otimes n}\geq c,\qquad
    \tvd(W_{\sigma_+},W_{\sigma_-})\leq\exp(-\Omega(n^{1/3}))
\end{align}
\end{lemma}
\begin{proof}
    Choosing $k=c_k n^{1/3}$ for suitably small constant $c_k>0$, define the following mixed state:
    \begin{align}\label{eq:gm-signal-states}
    \sigma_\pm=\E_{u\sim U_{\pm}}\left[
       \left(\left(1-\frac{su}{2n}\right)\ketbra0
                    +\frac{su}{2n}\ketbra1\right)^{\otimes n}
    \right]
\end{align}
for suitably large scaling constant $s$, where we take $n$ large enough such that $su/(2n)<1/2$. Considering the vacuum probabilities, 
\begin{align}
     \bra{0}^{\otimes n}\sigma_+\ket{0}^{\otimes n}-\bra{0}^{\otimes n}\sigma_-\ket{0}^{\otimes n}&=\E_{u\sim U_+}\left[\left(1-\frac{su}{2n}\right)^n \right]-\E_{u\sim U_-}\left[\left(1-\frac{su}{2n}\right)^n \right]\\
     &\geq \Pr[U_+=0]-(\Pr[U_-=0]+e^{-s/2}\Pr[U_->0])\\
     &\geq c-e^{-s/2}
\end{align}
where we use the fact that $(1-su/(2n))^n \leq e^{-su/2}\leq e^{-s/2}$ when $u\geq 1$. Taking suitable $c,s$, the vacuum probability is bounded away from $0$. 

We now show that the Wigner functions are close in total variation distance. Defining $g(q,p)=\pi^{-1}e^{-q^2-p^2}$, it is a known fact that the vacuum and one-photon Wigner functions are given by $W_{\ketbra{0}}=g$ and $W_{\ketbra{1}}=(2q^2+2p^2-1)g$ respectively, where the distribution of $q^2+p^2$ under joint density $g(q,p)$ can be viewed as an exponential random variable of unit mean.\alvantodo{prelim} Using these facts and linearity, the $i$-th mode of $\sigma_\pm$ has Wigner function:
\begin{align}
   W_{\left(1-\frac{su}{2n}\right)\ketbra0
                    +\frac{su}{2n}\ketbra1} =g(q_i,p_i)\left(1-\frac{su}{2n}+\frac{su}{2n}(2q_i^2+2p_i^2-1)\right)=g(q_i,p_i)\left(1+\frac{su}{n}Z_i\right)
\end{align}
where $Z_i=q_i^2+p_i^2-1$ gives $\E_{g_i}[Z_i]=0$ and $\E_{g_i}[Z_i^2]=1$. Then, by linearity and the tensor-product rule for the Wigner function,
\begin{align}\label{eq:Wsigmapm}
   W_{\sigma_\pm} &=\E_{u\sim U_\pm}\left[\prod_{i=1}^n g(q_i,p_i)\left(1+\frac{su}{n}Z_i\right)\right]=g^{\otimes n}\E_{u\sim U_\pm }\left[\prod_{i=1}^n \left(1+\frac{su}{n}Z_i\right)\right]
\end{align}
where $g^{\otimes n}(q,p)=\prod_{i=1}^n g(q_i,p_i)$ is the natural product density. By moment matching,
\begin{align}
    \frac{W_{\sigma_+}-W_{\sigma_-}}{g^{\otimes n}}=\sum_{S\subseteq[n]:|S|>k}\left(\frac s n\right)^{|S|}\left(\E[U_+^{|S|}]-\E[U_-^{|S|}]\right)\prod_{i\in S}Z_i 
\end{align}
Then, by Cauchy-Schwarz,
\begin{align}
    &\|W_{\sigma_+}-W_{\sigma_-}\|_1^2\\
    &=\E_{g^{\otimes n}}\left[\left| \frac{W_{\sigma_+}-W_{\sigma_-}}{g^{\otimes n}}\right|\right]^2\leq \E_{g^{\otimes n}}\left[\left( \frac{W_{\sigma_+}-W_{\sigma_-}}{g^{\otimes n}}\right)^2\right]\\
    &=\sum_{S_1,S_2\subseteq[n]:|S_1|>k,|S_2|>k}\left(\frac s n\right)^{|S_1|+|S_2|}\left(\E[U_+^{|S_1|}]-\E[U_-^{|S_1|}]\right)\left(\E[U_+^{|S_2|}]-\E[U_-^{|S_2|}]\right)\underbrace{\E_{g^{\otimes n}}\left[\prod_{i\in S_1}Z_i \prod_{i\in S_2}Z_i\right]}_{\mathbbm{1}[S_1=S_2]}\\
    &=\sum_{j=k+1}^{n}\binom{n}{j}\left(\frac s n\right)^{2j}(\E[U_+^j]-\E[U_-^j])^2\\
    &\leq 4\sum_{j=k+1}^n (eC^2 s^2 k^3/n)^j
\end{align}
where we use the independence, zero mean, and unit variance properties of $Z_i$ as well as boundedness $U_\pm \in [0,Ck^2]$ and the standard fact $\binom{n}{j}\leq n^j/j!\leq (en/j)^j$. Then, if $k=c_k n^{1/3}$ for suitably small $c_k$, we can bound $e C^2s^2 k^3/n\leq 1/4$, giving $\|W_{\sigma_+}-W_{\sigma_-}\|_1^2\leq O(4^{-k})$ by bounding the tail sum of a geometric series. Since $W_{\sigma_\pm}\geq 0$ from \eqref{eq:Wsigmapm}, $\tvd(W_{\sigma_+},W_{\sigma_-})\leq \exp(-\Omega(k))$ which proves \eqref{eq:gm-hidden-vacuum}.
\end{proof}

\noindent We now force our agnostic learner to distinguish this difference in vacuum probability by placing the difference along a coherent direction. Specifically, we define:
\begin{align}\label{eq:gm-extra-mode}
    \rho_\pm =(1-\eta)\ketbra{0}^{\otimes (n+1)}+\frac{\eta}{2}(\sigma_+\otimes \omega_\pm +\sigma_-\otimes \omega_\mp)
\end{align}
where $\omega_\pm =\frac 3 4 \ketbra{0}+1/4\ketbra{1}\pm \frac 1 4 (\ketbra{0}{1}+\ketbra{1}{0})$ which we can verify has $W_{\omega_\pm}\geq 0$, meaning $W_{\rho_\pm}\geq 0$ by linearity such that $\|W_{\rho_\pm}\|_1=1$ as well.\alvantodo{motivate this better} By construction, these states are $\eta$-close to vacuum and differ by a Gaussian phase rotation on the last mode that maps $\omega_+\mapsto \omega_-$ and vice versa, meaning $\opt_+=\opt_-=O(\eta)$. We observe that the difference factors as $\rho_+-\rho_-=\frac{\eta}{2} (\sigma_+-\sigma_-)\otimes (\omega_+-\omega_-)$. Then, by the tensor-product rule for Wigner functions, 
\begin{align}
    \|W_{\rho_+}-W_{\rho_-}\|_1=\frac \eta 2 \|W_{\sigma_+}-W_{\sigma_-}\|_1\|W_{\omega_+}-W_{\omega_-}\|_1\lesssim \eta \exp(-\Omega(n^{1/3}))
\end{align}
If $P_+,P_-$ are the distributions of the learner's output on $\rho_+$ and $\rho_-$, by Lemma 3 of \citep{sabapathy2021bosonic},
\begin{align}\label{eq:gm-output-distance}
    \tvd(P_+,P_-)\leq \frac 1 2 \|W_{\rho_+}^{\otimes N}-W_{\rho_-}^{\otimes N}\|_1\leq \frac N 2 \|W_{\rho_+}-W_{\rho_-}\|_1 \leq N\eta \exp(-\Omega(n^{1/3}))
\end{align}
which establishes that no protocol using only Gaussian measurements, including those that use adaptivity or arbitrary classical processing, can distinguish $\rho_+$ from $\rho_-$. Thus, Theorem~\ref{thm:gaussian-measurement-hardness} holds unless $\ket*{\hat \psi_+}=\ket*{\hat\psi_-}$. 

We next show that the guarantees for $\rho_+$ and $\rho_-$ can not simultaneously be satisfied. Suppose there exists a $\ket{\psi}$ such that $\infd(\rho_\pm ,\ket{\psi})\leq \opt_++c\eta^2$. Then, if $\bar \rho\equiv \frac 1 2 (\rho_++\rho_-)$, 
\begin{align}\label{eq:same-psi-start}
    \bra{\psi}\bar \rho\ket{\psi}\geq 1-\opt_+-c\eta^2 
\end{align}
We now show the following lemma to derive a contradiction.
\begin{lemma}\label{lemma:barrho-bound}
$1-\opt_\pm > \|\bar \rho\|_\op +\Omega(\eta^2)$.
\end{lemma}
\begin{proof}
From \eqref{eq:gm-extra-mode}, 
\begin{align}
    \bar \rho=(1-\eta)\ketbra{0}^{\otimes (n+1)}+\frac \eta 2 (\sigma_++\sigma_-)\otimes \left(\frac 3 4 \ketbra{0}+\frac 1 4 \ketbra{1}\right)
\end{align}
which is diagonal in the Fock basis. For $\eta \leq 1/4$, since $\|\sigma_\pm \|_\op \leq 1$ and $\bra{0}\omega_+\ket{0}=\bra{0}\omega_-\ket{0}$, it follows that 
\begin{align}\label{eq:rhobarop}
\|\bar \rho\|_{\op}=\bra{0}^{\otimes (n+1)}\bar \rho\ket{0}^{\otimes (n+1)}=\bra{0}^{\otimes (n+1)}\rho_\pm \ket{0}^{\otimes (n+1)}
\end{align}
We now exhibit a pure Gaussian state that achieves a fidelity that surpasses this vacuum fidelity by $\Omega(\eta^2)$. Using the fact that $\bra{0}\omega_\pm \ket{1}=\pm 1/4$ and $\bra{1}\omega_\pm \ket{1}=1/4$, consider the following elements of $\rho_\pm$: 
\begin{align}\label{eq:last-mode-t}
    (\bra{0}^{\otimes n}\otimes \bra{0})\rho_\pm (\ket{0}^{\otimes n}\otimes \ket 1 )&=\pm \frac \eta 8 \bra{0}^{\otimes n}(\sigma_+-\sigma_-)\ket{0}^{\otimes n}\equiv \pm t
    % \\
    % (\bra{0}^{\otimes n}\otimes \bra{1})\rho_\pm (\ket{0}^{\otimes n}\otimes \ket 1 )&=\frac \eta 8 \bra{0}^{\otimes n}(\sigma_++\sigma_-)\ket{0}^{\otimes n}
\end{align}
We then consider the state $\ket{\phi_\pm}=\ket{0}^{\otimes n}\otimes \ket{\psi(\pm t,0)}$ where $\ket{\psi(\pm t,0)}=e^{-t^2/2}\sum_{j\geq 0}\frac{(\pm t)^j}{\sqrt{j!}}\ket{j}$ from \eqref{eq:barg-state} is a coherent state with this displacement on the last mode. Then,  
\begin{align}
    1-\opt_\pm \geq \bra{\phi_\pm }\rho_\pm \ket{\phi_\pm }&=e^{-t^2}\bigg(\bra{0}^{\otimes (n+1)}\rho_\pm \ket{0}^{\otimes (n+1)}+2(\pm t) (\bra{0}^{\otimes n}\otimes \bra{0})\rho_\pm (\ket{0}^{\otimes n}\otimes \ket 1 )
    \\&\quad \quad \quad \quad + t^2(\bra{0}^{\otimes n}\otimes \bra{1})\rho_\pm (\ket{0}^{\otimes n}\otimes \ket 1 )\bigg)\\
    &\geq e^{-t^2}(\|\bar \rho\|_{\rm op}+2t^2)\\
    &\geq \|\bar \rho\|_{\op}+t^2 -2t^4
\end{align}
where we drop the last term, recall \eqref{eq:rhobarop}, and use the fact that $e^{-t^2}\geq 1-t^2$ and $\|\bar\rho\|_\op\leq 1$. Since $t=\Theta(\eta)$ by Lemma~\ref{eq:gm-hidden-vacuum} and \eqref{eq:last-mode-t}, we can take $t\leq 1/2$ such that $t^2-2t^4\geq t^2/2>\Omega(\eta^2)$, proving the claim. 
\end{proof}
\noindent Then, \eqref{eq:same-psi-start} and Lemma~\ref{lemma:barrho-bound} imply $\bra{\psi}\bar \rho\ket{\psi}\geq 1-\opt_+-c\eta^2 > \|\bar\rho\|_\op$ for suitable $c>0$, which is a contradiction. This proves that no $\ket{\psi}$ can satisfy both the guarantee for $\rho_+$ and that for $\rho_-$. It follows that $\tvd(P_+,P_-)\geq 1/3$, which together with \eqref{eq:gm-output-distance} proves Theorem~\ref{thm:gaussian-measurement-hardness}.

\section{Hardness in the low-fidelity regime}\label{sec:clique-lb}
For pure product states in the qubit setting, \citep{bakshi2025learning} show that a polytime agnostic learner in the low-fidelity regime can be used to efficiently approximate the spectral norm of a four-tensor $T\in \C^{m\times m\times m\times m}$ to additive error. For $\|T\|_F=1$, they do this by first encoding $T\mapsto \ket{\psi_T}\equiv \sum_{ijk\ell}T_{ijk\ell}\ket{u_iu_ju_ku_\ell}$ where $\{\ket{u_i}\}_{i=1}^m$ are a Haar random choice of orthogonal unit vectors in ${\rm span}\{\ket{100\dots00},\ket{010\dots 00},\dots\ket{000\dots 01}\}\subset \C^{n}$, i.e. the subspace of single excitations. Then, taking $n=\poly(m/\epsilon)$, they argue that the fidelity achieved by the best product state approximates some known constant multiple of the spectral norm of $T$ to additive error $\epsilon$. 

The hardness of this tensor problem comes from a neat fact in graph theory due to \citep{motzkin1965maxima} which relates the following maximization problem over a graph, $G$, to the size of the maximum clique $\omega(G)$.
\begin{fact}[Theorem 1 of \citep{motzkin1965maxima}]
   Let $G=([n],E)$ be a graph with vertices labeled by the integers. Then,  
   \begin{equation}\label{eq:motzkin}
      \max_{\|x\|_2=1}\sum_{\{i,j\}\in E}|x_i|^2|x_j|^2=\max_{\stackrel{p_i\geq 0}{\sum_i p_i=1}}\sum_{\{i,j\}\in E} p_i p_j=\frac 1 2 \left(1-\frac{1}{\omega(G)}\right)
   \end{equation}
   where the maximizer assigns weight $1/\omega(G)$ to each vertex of a maximum clique.
\end{fact}
\noindent In particular, Theorem 8.4 of \citep{friedland2018nuclear} massages the left hand side of \eqref{eq:motzkin} into the spectral norm of a cleverly constructed four-tensor such that efficiently approximating this norm to inverse-polynomial error resolves a unique integer for $\omega(G)$, an NP-hard task. 

Now, if we employ a similar strategy for pure Gaussian states with the $\{\ket{u_i}\}_{i=1}^m$ being single-photon excitations, squeezing, i.e. the quadratic terms in \eqref{eq:barg-state}, cause $\lvert\braket{\psi(b,A)}{\psi_T}\rvert^2$ to contain significant terms with couplings that do not appear in the spectral norm of $T$. Nevertheless, these can be controlled by instead encoding each index as coherent superpositions of pairs of photons $\ket{u_i}=\frac{1}{\sqrt r}\sum_{j=1}^r \hat{a}_{L(i,j)}^\dagger \hat{a}^\dagger_{R(i,j)}\ket{a}$ with a large $r=\poly(1/\epsilon)$ such that the optimum fidelity between $\ket{\psi_T}$ and a pure Gaussian state is exactly the spectral norm of $T$ up to additive error $\epsilon$. While one can show that this approach gives the desired lower bound, expanding and bounding the fidelity in terms of the Bargmann parameters is a bit arduous\alvantodo{add this in v2 maybe?}. Instead, this section gives a more straightforward proof that uses \eqref{eq:motzkin} to directly solve maximum clique with the claimed polytime agnostic learner. 

\begin{theorem}\label{thm:lf-hardness}
    Unless $\mathrm{NP}\subseteq \mathrm{BQP}$, there does not exist a quantum algorithm which takes any $\epsilon >0$, copies of any $n$-mode state $\rho$, and  $\poly(n,1/\epsilon)$ time to output a pure Gaussian state $\ket*{\hat \psi}$ such that $\infd(\rho,\ketbra*{\hat \psi})\leq \opt +\epsilon $ with probability at least $0.99$.
\end{theorem}

\noindent The key trick is to define the projector $P_G=\sum_{\{i,j\}\in E} \hat{a}_i^\dagger \hat{a}_j^\dagger \ketbra{0} \hat{a}_j \hat{a}_i$ which projects onto the two-photon excited states corresponding with the edges of $G$. Considering its expectation value with a pure Gaussian state via \eqref{eq:barg-state-expand}, 
\begin{align}
    \bra{\psi(b,A)}P_G\ket{\psi(b,A)}=\sum_{\{i,j\}\in E}\lvert \bra{\psi(b,A)}\hat{a}_i^\dagger \hat{a}_j^\dagger \ket{0}\rvert^2=\frac{1}{Z(b,A)}\sum_{\{i,j\}\in E}|A_{ij}+b_ib_j|^2
\end{align}
We now observe that restricting to the family $b(x;\beta)=\beta x$ and $A(x;\alpha)=\alpha xx^\top $ for $\beta >0$ and $\alpha \in \C$ collapses the summation to exactly the maximand of \eqref{eq:motzkin}:
\begin{align}\label{eq:motzkin-psiPg-1}
    \bra{\psi(x;\beta,\alpha)}P_G\ket{\psi(x;\beta,\alpha)}=\frac{|\alpha+\beta^2|^2}{Z(\beta x,\alpha xx^\top)}\sum_{\{i,j\}\in E}|x_i|^2|x_j|^2
\end{align}
We now claim that $Z(\beta x,\alpha xx^\top)$ is independent of $x$. Applying the passive unitary which carries $(\beta x,\alpha xx^\top)\to (\beta e_1,\alpha e_1e_1^\top)\equiv (b',A')$,
\begin{align}
    \log Z(\beta x,\alpha xx^\top)=\log Z(b',A')&=-\frac 1 2 \log \det M(A')+u(b')^\top M(A')^{-1}u(b')\\
    &=-\frac 1 2 \log (1-|\alpha|^2) + \beta^2\frac{1+\Re\alpha}{1-|\alpha|^2}\\
    \label{eq:Z(beta,alpha)}
    Z(\beta,\alpha)\equiv Z(\beta x,\alpha xx^\top ) &=\frac{\exp(\frac{\beta^2(1+\Re \alpha)}{1-|\alpha|^2})}{\sqrt{1-|\alpha|^2}}
\end{align}
where we use Lemma~\ref{lemma:lf-normalization-formula} with $M(A')=\begin{bmatrix}
        I-\Re\{\alpha\} e_1e_1^\top & \Im \{\alpha\} e_1e_1^\top \\
        \Im \{\alpha\} e_1e_1^\top & I+\Re\{\alpha\} e_1e_1^\top 
    \end{bmatrix}$ and $u(b')=\begin{bmatrix}
        \beta e_1 \\ 0 
    \end{bmatrix}$ to simplify. Thus, the maximal value of $\bra{\psi(x;\beta,\alpha)}P_G\ket{\psi(x;\beta,\alpha)}$ for any fixed $(\alpha,\beta)$ gives a known expression involving $\omega(G)$. 

    Unfortunately, our agnostic learner optimizes over all pure Gaussian states as opposed to just this family, so even if we let our arbitrary state be $\rho=P_G/\tr(P_G)$, it is possible that a pure Gaussian state outside of this family may exploit squeezing to achieve a significantly larger expectation against $P_G$ than the best possible in this family. Thus, we would like to perturb $\rho$ with a term that discourages pure Gaussian states outside of this family without biasing any particular member of the family. Since all $\ket{\psi(x;\beta,\alpha)}$ for a fixed $(\beta,\alpha)$ can reach each other via a passive unitary, we expect this term to have an expectation value which is invariant under passive changes of modes. 
    
    A natural choice is the one-photon projector $P_1=\sum_{i=1}^n \hat{a}_i^\dagger \ketbra{0}\hat{a}_i$. In particular, we prove that there exist $(\alpha^*,\beta^*)$ such that any $\ket{\psi}\in \c G$ that nearly maximizes $p_1(\ket{\psi})=\bra{\psi}P_1\ket{\psi}$ must be close to a state of the form $\ket{\psi(x;\beta^*,\alpha^*)}$, with all maximizers taking exactly this form. 

    \begin{lemma}\label{lemma:P1-stability}
    Let $(\beta^*,\alpha^*)=(\sqrt{3/2},-1/2)$ and abbreviate $\ket{\psi^*(x)}\equiv \ket{\psi(x;\beta^*,\alpha^*)}$ for $\|x\|_2=1$. Then, for any $\ket{\psi(b,A)}\in \c G$ with $b\neq 0$,
    \begin{align}\label{eq:p1-stability-guarantee}
    p_1(\ket{\psi})\leq p_1^*,\quad \inf_{\|x\|_2=1}\infd(\ket\psi,\ket{\psi^*(x)})\lesssim p_1^*-p_1(\ket{\psi})
    \end{align}
    where $p_1^*=3\sqrt{3}/(4e)$ is the maximal value achieved by states exactly of the form $\ket{\psi(x)}$.
    \end{lemma}
    \noindent We defer the proof to Appendix~\ref{app:proof-P1-stability}.

    Taking stock of our progress, the one-photon projector $P_1$ heavily biases pure Gaussian states of the form $\ket{\psi^*(x)}$ but does not discriminate a particular $x\in \C^n $. On the other hand, $P_G$, when restricted to the family $\ket{\psi^*(x)}$, selects out the $x$ that maximizes the left hand side of \eqref{eq:motzkin}, yielding an expression for $\omega(G)$. Thus, we take our arbitrary state $\rho$ to be $P_1$ with a small perturbation in the direction $P_G$:
    \begin{align}\label{eq:lf-lb-rhoG}
        \rho_G=\frac{P_1+\lambda P_G}{n+\lambda |E|}
    \end{align}
    which clearly has $\tr \rho_G=1$. We now use the stability of optimizing $p_1(\cdot)$ and \eqref{eq:motzkin-psiPg-1} to relate the optimum of our pure Gaussian state optimization problem to $\omega(G)$. 
    \begin{lemma}
        Let $\OPT_G=\max_{\|x\|_2=1}\sum_{\{i,j\}\in E}|x_i|^2|x_j|^2$ and $\OPT_{\rho_G}\equiv   \max_{\ket{\psi}\in \c G }F(\rho_G,\ketbra{\psi})$. Then, 
        \begin{align}\label{eq:optG-to-optRhoG}
        p_1^*+\frac{2}{3}p_1^*\lambda \OPT_G\leq (n+\lambda |E|)\; \OPT_{\rho_G}\leq  p_1^*+\frac{2}{3}p_1^*\lambda \OPT_G+C\lambda^2
        \end{align}
        for universal constant $C>0$. 
    \end{lemma}
    \begin{proof}
        Define $p_G(\ket{\psi})\equiv \bra{\psi}P_G\ket{\psi}$ such that:
        \begin{align}
            (n+\lambda |E|)\OPT_{\rho_G}=\max_{\ket{\psi}\in \c G } (p_1(\ket{\psi})+\lambda p_G(\ket{\psi}))
        \end{align}
        Then, the lower bound immediately follows from restricting to the family $\ket{\psi^*(x)}$:
        \begin{align}
            \max_{\ket{\psi^*(x)}} (p_1(\ket{\psi^*(x)})+\lambda p_G(\ket{\psi^*(x)}))&=p_1^*+\lambda\sqrt{1-|\alpha^*|^2}\exp(-\frac{(\beta^*)^2(1+\Re \alpha^*)}{1-|\alpha^*|^2})|\alpha^*+(\beta^*)^2|^2 \OPT_G\\
            &=p_1^*+\lambda\frac{\sqrt 3}{2e}\OPT_G\\
            &=p_1^*+\frac 2 3 \lambda p_1^*\OPT_G \label{eq:pg-opt-in-family}
        \end{align}
        where the first term is given by Lemma~\ref{lemma:P1-stability} and the second follows from \eqref{eq:motzkin-psiPg-1} and \eqref{eq:Z(beta,alpha)}. 
        
        For the upper bound, we leverage Lemma~\ref{lemma:P1-stability} to control the performance of $\ket{\psi}\in \c G$. Specifically, suppose $\ket{\psi}$ fails to meet the optimum value of $p_1(\cdot)$ by $\Delta_1 =p_1^*-p_1(\ket{\psi})$. Then, by Lemma~\ref{lemma:P1-stability}, there exists an $x\in \C^n$ with $\|x\|_2=1$ such that $\infd(\ket{\psi},\ket{\psi^*(x)})\lesssim \Delta_1$. Using this to control the deviation of $p_G(\cdot)$ from the optimum:
        \begin{align}
            |p_G(\ket{\psi^*(x)})-p_G(\ket{\psi})|\leq \trd(\ketbra{\psi},\ketbra{\psi^*(x)})=\sqrt{\infd(\ket{\psi},\ket{\psi^*(x)})}\leq C_1 \sqrt{\Delta_1}
        \end{align}
        for some constant $C_1>0$. Applying this, 
        \begin{align}
           p_1(\ket{\psi})+\lambda p_G(\ket{\psi})\leq p_1^*-\Delta +\lambda p_G(\ket{\psi^*(x)})+\lambda C_1\sqrt{\Delta_1}&\leq p_1^*+\frac{2}{3}\lambda p_1^*\OPT_G -(\sqrt{\Delta_1}-\lambda C_1/2)^2+\frac{C_1^2}{4}\lambda^2\\
           &\leq p_1^*+\frac 2 3 \lambda p_1^*\OPT_G +C_1 \lambda^2
        \end{align}
        where we bound $p_G(\ket{\psi^*(x)})\leq \max_{\ket{\psi^*(x)}}p_G(\ket{\psi^*(x)})=\frac 2 3 \lambda p_1^*\OPT_G $ from the computation in \eqref{eq:pg-opt-in-family}, complete the square on the error terms, and define $C$ accordingly. This proves the claim. 
    \end{proof}

\noindent We are now ready to reduce maximum clique. 
\begin{proof}[Proof of Theorem~\ref{thm:lf-hardness}]
   Suppose we have a polytime agnostic learner for pure Gaussian states to error $\opt+\epsilon$. From \eqref{eq:lf-lb-rhoG}, given $G$, we can prepare a copy of $\rho$ by preparing $\hat{a}_i^\dagger \hat{a}_j^\dagger \ket{0}$ with probability $\frac{\lambda}{n+\lambda |E|}$ for $\{i,j\}\in E$ and $\hat{a}_i^\dagger \ket{0}$ with probability $\frac{n}{n+\lambda |E|}$ for $i\in [n]$. Feeding ${\rm poly}(n,1/\epsilon_1)$ copies of $\rho_G$ to our learner gives a $\ket{\hat \psi}$ such that $F(\rho_G,\ketbra*{\hat \psi})\geq \OPT_{\rho_G}-\epsilon $. Since $\ket*{\hat \psi}$ is a pure Gaussian state, we can then use Lemma~\ref{lemma:hf-value-oracle} with $\Otilde (\epsilon^{-2})$ copies of $\rho_G$ to estimate $F(\rho_G,\ketbra*{\hat \psi})$ to additive error $\epsilon$, using Gaussian unitaries and on-off photon detection. It follows that we can estimate $\OPT_{\rho_G} $ to additive error $\epsilon$ in $\poly(n,1/\epsilon)$ time.

   Calling this estimate $\widehat{\OPT}_{\rho_G}$, form the following estimates:
   \begin{align}
   \widehat{\OPT}_G&=\frac{3}{2p_1^*\lambda}((n+\lambda |E|)\widehat{\OPT}_{\rho _G}-p_1^*)\\
   \widehat{\omega(G)}&={\rm round}\left(\frac{1}{1-2\widehat{\OPT}_G}\right)
   \end{align}
   Viewing $\OPT_G$ as a function of $\omega(G)$,we have that $\OPT_G(\omega+1)-\OPT_G(\omega)=\frac 1 {2\omega(\omega+1)}\gtrsim n^{-2}$. 
   Since $\omega(G)$ is an integer, $\widehat{\omega(G)}=\omega(G)$ if $|\widehat{\OPT}_G-\OPT_G|\leq c n^{-2}$ for small constant $c>0$. 
   Re-arranging \eqref{eq:optG-to-optRhoG}, 
   \begin{align}
       |\widehat{\OPT}_G-\OPT_G|\leq \frac{3}{2p_1^*\lambda }(C\lambda^2 +(n+\lambda |E|)\epsilon )\lesssim \lambda + (n/\lambda+n^2)\epsilon 
   \end{align}
   Thus, taking $\lambda=O(n^{-2})$ and $\epsilon=O(n^{-5})$ suffices. This gives an efficient algorithm for computing $\omega(G)$, which is NP-hard. \alvantodo{should probably handle energy dependence to be precise, its not that meaningful}
\end{proof}

% \subsubsection{Lower bound}
% \alvan{see \texttt{high\_fidelity\_gaussian\_hardness.tex} for a lower bound against poly-time $\opt+\epsilon$ algorithms using only Gaussian measurements, in the high-fidelity regime. The strategy is to construct a family of pairs of near-vacuum states with arbitrarily small $\opt$ s.t. the Wigner distributions are exponentially close but they are each near different coherent states. }

% \subsection{What about mixed Gaussian states?}
% \alvan{I have yet to run and try to prove lower or upper bounds for this, but upper is hard of course because of the $\sqrt n$ loss. Maybe you can get lower bound by exploiting this? }

% \section{Discussion}

% \section*{Acknowledgements}
% ?
% 
\subsection*{AI Usage Disclosure}
Multiple technical ideas were derived in discussions with GPT-5.6 and GPT-6.
In multiple conversations with GPT, the authors proposed the overall framework, the connection with robust statistics, pruning, and polynomial optimization, and the details of these ideas were then worked out in collaboration with GPT.
The preparation of the manuscript was also assisted by Claude Opus 5.5 and Fable 5.1.
The authors claim responsibility for any errors and omissions in the manuscript.

\newpage

\bibliographystyle{alpha}
\bibliography{refs,refs-claude}

\appendix
\clearpage

\section{Deferred proofs}
\label{app:deferred-proofs}

\subsection{Deferred proofs from Section~\ref{sec:prelims}}
\label{app:deferred-prelims}

\subsubsection{Proof of Lemma~\ref{lemma:lf-bargmann-moments}}
\label{app:proof-lf-bargmann-moments}
\begin{proof}
In the Bargmann representation, $\hat{a}_j$ acts as $\partial_{z_j}$ and $\hat{a}_k^\dagger$ acts as multiplication by $z_k$. Differentiating the exponential in \eqref{eq:barg-state}, using the symmetry of $A$, gives
\[
    \partial_{z_j}e^{b^\top z+z^\top Az/2}
       =\left(b_j+\sum_k A_{jk}z_k\right)e^{b^\top z+z^\top Az/2}.
\]
Consequently, with $\hat{a}=(\hat{a}_1,\ldots,\hat{a}_n)^\top$ and $\hat{a}^\dagger$ denoting the column of creation operators,
\begin{equation}\label{eq:lf-bargmann-annihilation}
    (a-Aa^\dagger)\ket\psi=b\ket\psi.
\end{equation}
Hitting this whole expression by $\bra{\psi}$ yields $b=\beta-A\bar\beta$. Its conjugate gives $\bar\beta=\bar b+\bar A\beta$, so substitution yields $(I-A\bar A)\beta=b+A\bar b$. Since $A$ is symmetric, $\bar A=A^\dagger$ and $I-A\bar A$ is positive definite, proving the inverse formula in \eqref{eq:lf-displacement-relation}. The triangle and reverse triangle inequalities applied to $\beta-A\bar\beta$ prove \eqref{eq:lf-displacement-bounds}.

We next prove the photon-number formula \eqref{eq:lf-bargmann-energy}. By definition, $\nu(\pi_{b,A})$ is the expectation of the total number operator in the normalized state $\ket\psi$:
\begin{equation}\label{eq:lf-photon-expectation}
    \nu(\pi_{b,A})=\tr\left(\pi_{b,A}\sum_{j=1}^n \hat{a}_j^\dagger \hat{a}_j\right)
       =\sum_{j=1}^n\bra\psi \hat{a}_j^\dagger \hat{a}_j\ket\psi
       =\sum_{j=1}^n\|\hat{a}_j\ket\psi\|^2.
\end{equation}
To evaluate this sum, center the annihilation operators by setting $d_j=\hat{a}_j-\beta_j I$. They satisfy $\langle d_j\rangle_\psi=0$ and $[d_j,d_k^\dagger]=\delta_{jk}I$. Now expand $\hat{a}_j=d_j+\beta_j I$:
\begin{align*}
    \bra\psi \hat{a}_j^\dagger \hat{a}_j\ket\psi
       &=\bra\psi d_j^\dagger d_j\ket\psi
         +\beta_j\bra\psi d_j^\dagger\ket\psi
         +\bar\beta_j\bra\psi d_j\ket\psi+|\beta_j|^2\\
       &=\|d_j\ket\psi\|^2+|\beta_j|^2.
\end{align*}
Thus \eqref{eq:lf-photon-expectation} becomes
\begin{equation}\label{eq:lf-centered-photon-number}
    \nu(\pi_{b,A})=\|\beta\|_2^2+\sum_{j=1}^n\|d_j\ket\psi\|^2.
\end{equation}
It remains to show that $\sum_{j=1}^n \|d_j \ket{\psi} \|^2 = \sum_{j=1}^n s_j^2/(1-s_j^2)$.
Substituting $\hat{a}_j = d_j + \beta_j I$ into \eqref{eq:lf-bargmann-annihilation} yields
\[
    (d-Ad^\dagger)\ket\psi=0.
\]
Now, choose a Takagi factorization $A=USU^\top$, where $U$ is unitary and $S=\operatorname{diag}(s_1,\ldots,s_n)$. Define $c=U^\dagger d$, so the column of adjoints is $c^\dagger=U^\top d^\dagger$. Multiplication by $U^\dagger$ therefore gives, for each $j$,
\[
    c_j\ket\psi=s_jc_j^\dagger\ket\psi,
    \qquad [c_j,c_k^\dagger]=\delta_{jk}I.
\]
% All these second moments are finite: in Bargmann space, the squared wavefunction times the Gaussian weight is bounded by a constant times
% $\exp(-(1-\|A\|_{\rm op})\|z\|_2^2+2\|b\|_2\|z\|_2)$, which remains integrable after multiplication by any fixed polynomial. 
Letting $u_j=\|c_j\ket\psi\|^2$, we find that
\[
    u_j=s_j^2\|c_j^\dagger\ket\psi\|^2
       =s_j^2\bra\psi c_jc_j^\dagger\ket\psi
       =s_j^2(u_j+1),
    \qquad u_j=\frac{s_j^2}{1-s_j^2},
\]
where the third equality follows from the fact that $[c_j,c_j^\dagger] = I$. Since $c=U^\dagger d$ and $U$ is unitary, it preserves the norm of the vector $(d_1\ket\psi,\ldots,d_n\ket\psi)$ in $\c F_n^{\oplus n}$. Thus
\[
    \sum_j\|d_j\ket\psi\|^2
       =\sum_j\|c_j\ket\psi\|^2
       =\sum_j u_j
       =\sum_j\frac{s_j^2}{1-s_j^2}.
\]
Substituting this into \eqref{eq:lf-centered-photon-number} proves \eqref{eq:lf-bargmann-energy}. 
% In particular, substituting also the displacement formula \eqref{eq:lf-displacement-relation} gives an expression entirely in terms of the original parameters:
% \begin{equation}\label{eq:lf-photon-number-parameters}
%     \nu(\pi_{b,A})
%        =\left\|(I-A\bar A)^{-1}(b+A\bar b)\right\|_2^2
%          +\tr\left[A\bar A(I-A\bar A)^{-1}\right].
% \end{equation}
% Here the trace equals $\sum_j s_j^2/(1-s_j^2)$ because $A\bar A=AA^\dagger$ has eigenvalues $s_j^2$.
\end{proof}

\subsubsection{Proof of Lemma~\ref{lemma:lf-normalization-formula}}
\label{app:proof-lf-normalization-formula}
\begin{proof}
\emph{Spectrum of $M(A)$.} Write $A=P+iQ$ and take Takagi vectors $u_j=p_j+iq_j$ as in \eqref{eq:takagi-factorization}. Taking real and imaginary parts of $A\bar u_j=s_j u_j$ gives
\[
    Pp_j+Qq_j=s_jp_j,\qquad Qp_j-Pq_j=s_jq_j.
\]
Direct multiplication by $M(A)=\left(\begin{smallmatrix}I-P&Q\\Q&I+P\end{smallmatrix}\right)$ now shows
\[
    M(A)\begin{pmatrix}p_j\\-q_j\end{pmatrix}
       =(1-s_j)\begin{pmatrix}p_j\\-q_j\end{pmatrix},\qquad
    M(A)\begin{pmatrix}q_j\\p_j\end{pmatrix}
       =(1+s_j)\begin{pmatrix}q_j\\p_j\end{pmatrix}.
\]
These $2n$ real vectors form an orthonormal basis: the real and imaginary parts of $u_j^\dagger u_k=\delta_{jk}$ give
\[
    p_j^\top p_k+q_j^\top q_k=\delta_{jk},\qquad
    p_j^\top q_k-q_j^\top p_k=0.
\]
Thus we have exhibited all the eigenvalues of $M(A)$, namely $1-s_j$ and $1+s_j$. They are positive because $s_j\leq\|A\|_{\rm op}<1$.

\emph{Normalization formula.} To express \eqref{eq:lf-normalization-integral} as a real integral, write $z=x+iy$ and $b=b_R+ib_I$, keeping $A=P+iQ$ as above, with all components real. Since $A$ is symmetric, so are $P$ and $Q$, and expansion gives
\[
    \Re(z^\top Az)=x^\top Px-y^\top Py-2x^\top Qy,
    \qquad \Re(b^\top z)=b_R^\top x-b_I^\top y.
\]
Consequently, for $w=(x^\top,y^\top)^\top$, the exponent in \eqref{eq:lf-normalization-integral} is
\begin{align*}
    &-x^\top(I-P)x-y^\top(I+P)y-2x^\top Qy
       +2b_R^\top x-2b_I^\top y\\
    &\hspace{2em}=-w^\top M(A)w+2u(b)^\top w.
\end{align*}
Here $d^{2n}z=dx\,dy=dw$: we are simply identifying $\C^n$ with $\R^{2n}$. Hence
\begin{equation}\label{eq:lf-real-normalization-integral}
    Z(b,A)=\pi^{-n}\int_{\R^{2n}}
        e^{-w^\top M(A)w+2u(b)^\top w}\,dw.
\end{equation}

Abbreviating $M=M(A)$ and $u=u(b)$, complete the square:
\[
    -w^\top Mw+2u^\top w
       =-(w-M^{-1}u)^\top M(w-M^{-1}u)+u^\top M^{-1}u.
\]
Translating by $v=w-M^{-1}u$ and using the standard Gaussian integral for $M\succ0$ gives
\[
    Z(b,A)=\pi^{-n}e^{u^\top M^{-1}u}
       \int_{\R^{2n}}e^{-v^\top Mv}\,dv
       =\frac{e^{u^\top M^{-1}u}}{\sqrt{\det M}}.
\]
Taking logarithms proves \eqref{eq:lf-normalization-formula}.
\end{proof}

\subsection{Deferred proofs from Section~\ref{sec:hf}}
\label{app:deferred-hf}

\subsubsection{Proof of Theorem~\ref{thm:robust-heterodyne}}
\label{app:proof-robust-heterodyne}
\begin{proof}
    We perform a heterodyne measurement on each copy to obtain samples from a distribution $P$ that is ${\opt_{\rm tr}}$-close in total variation distance to the classical Gaussian $\c N(\mu_h,\Sigma_h)$ where $\mu_h={\bf m}$ and $\Sigma_h=\frac{V+I}{2}$. Taking $\overline{\opt}_{\rm tr}=\overline{\opt}_{\rm TV}$, we apply Proposition~\ref{prop:robust-gaussian-learner} to obtain estimates $\hat \mu_h$ and $\hat \Sigma_h$ which satisfy the relative $\ell_2$ and Frobenius guarantees in \eqref{eq:robust-gaussian-learner-relative-guarantee-F}.
    Then, we construct the estimator $\hat V\equiv \frac{2\hat \Sigma_h}{1-\zeta}-I$. 
    
    Now, from the relative Frobenius guarantee in \eqref{eq:robust-gaussian-learner-relative-guarantee-F} for $\hat \mu_h$ and $\hat \Sigma_h$, we have that $(1-\zeta)\Sigma_h\preceq \hat \Sigma_h\preceq (1+\zeta)\Sigma_h$. Applying this, 
    \begin{equation}\label{eq:rht-V-err-psd}
        0\preceq \hat V -V \preceq \frac{4\zeta}{1-\zeta}\Sigma_h\preceq \frac{16}{3}\zeta \Sigma_h 
    \end{equation}
    where the last inequality follows from $\zeta \leq 1/4$. This means that $\hat V+i\Omega \succeq V+i\Omega \succeq 0$, so $\pi(\hat \mu_h, \hat V)$ is a valid Gaussian state. Then, by Fact~\ref{fact:trd-to-param}, 
    \begin{align}
        \trd(\pi,\hat \pi)&\lesssim \|V^{-1/2}(\hat \mu_h -\mu_h)\|_2+\tr ((V^{-1}+\hat V^{-1})|\hat V-V|)\\
        &\leq \|V^{-1/2}\Sigma_h^{1/2}\|_{\rm op}\|\Sigma_h^{-1/2}(\hat \mu_h-\mu_h)\|_2+2\tr(V^{-1}(\hat V-V))\\
        &\lesssim \zeta \|V^{-1/2}\Sigma_h^{1/2}\|_{\rm op}+\zeta \tr (V^{-1}\Sigma_h)\\
        &\lesssim \zeta(n+\tr V^{-1})
    \end{align}
    where the second term in the second inequality comes from $\hat V^{-1}\preceq V^{-1}$ and where the fourth inequality uses the bound $\|V^{-1/2}\Sigma_h^{1/2}\|_{\rm op}\leq \|V^{-1/2}\Sigma_h^{1/2}\|_F=\sqrt{\tr(V^{-1}\Sigma_h)}$ and the fact that $\tr(V^{-1}\Sigma_h)=n+\frac 1 2 \tr V^{-1}$. Then, \eqref{eq:rht-bound} follows by triangle inequality:
    \begin{equation}
        \trd(\rho,\hat \pi)\leq \trd(\rho,\pi)+\trd(\pi,\hat \pi)\mathrel{\lesssim}{\opt_{\rm tr}} +\zeta (n+\tr V^{-1})\lesssim (n+\tr V^{-1})\cdot \Otilde({\opt_{\rm tr}})\lesssim E(\pi)\cdot \Otilde({\opt_{\rm tr}})
    \end{equation}
    where the final inequality comes from $\tr V^{-1}=\tr(\Omega V^{-1}\Omega^\top)\leq \tr V \leq 4 E(\pi)$, with the second inequality coming from taking the Schur complement of the top right block of $V+i\Omega \succeq 0$. 
\end{proof}

\subsubsection{Proof of Theorem~\ref{thm:reit}}
\label{app:proof-reit}
\begin{proof}
    Lemmas 12 and 13 in \citep{bittel2025energy} prove the following unsqueezing property which substantiates the whole algorithm in Theorem~\ref{thm:eit}:
    \begin{equation}
    \|V_{k+1}^{-1}\|_{\rm op}\leq \sqrt{1+\frac{2\zeta}{1-\zeta}(\|V_k^{-1}\|_\infty+1)}\leq \sqrt{\frac 5 3 +\frac 2 3 \|V^{-1}\|_\infty}
    \Rightarrow 
    \|V_K^{-1}\|_{\rm op}\leq \sqrt{7/3}\cdot \|V^{-1}\|_{\rm op}^{1/2^K}
    \end{equation}
    Observing the proofs, we notice that the two sufficient claims for our purposes are showing that both $0\preceq \hat V_k-V_k\preceq \frac{2\zeta}{1-\zeta}(V_k+I)$ and $\zeta \leq 1/4$ hold for our robust estimators. Luckily, this is immediate. First, unitary invariance preserves the trace-distance bound $\trd(\rho_k,\pi_k)\leq\opt_{\rm tr}\leq \overline{\opt}_{\rm tr}=\overline{\opt}_{\rm TV}$ across iterations, meaning these two conditions directly follow from \eqref{eq:rht-V-err-psd} and Proposition~\ref{prop:robust-gaussian-learner}. The claimed copy complexity then comes from the $O(n^2 {\opt_{\rm tr}}^{-2}\log^51/\delta)$ copies needed by each usage of robust heterodyne tomography.
\end{proof}

\subsubsection{Proof of Lemma~\ref{lemma:frob-h-to-trd}}
\label{app:proof-frob-h-to-trd}
\begin{proof}\alvantodo{check for simplification}
    We begin by converting the absolute error guarantees of \eqref{eq:frob-h-to-trd-hypoth} for the heterodyne outcome distribution into the following relative guarantees on the raw, uninflated parameter estimates for the state:
\begin{equation}\label{eq:riuproj-V-relf}
     \|V^{-1/2}(\hat {\bf m}-{\bf m})\|_2\leq \sqrt 2 \zeta \quad {\rm and }\quad  \|V^{-1/2}(\hat V_{\rm raw}-V)V^{-1/2}\|_{F}\leq 4\zeta 
\end{equation}
To show this, we first recall that the eigenvalues of a pure Gaussian state's covariance $V=SS^\top $ come in pairs of the form $(\lambda,1/\lambda)$, meaning $\|V^{-1}\|_{\rm op}\leq 2$ implies $I/2\preceq V \preceq 2I $ and $0\preceq \Sigma_h\preceq 3V/2$. Then, by standard properties of $\|\cdot\|_{\rm op}$, 
\begin{align}
    \|V^{-1/2}(\hat {\bf m}-{\bf m})\|_2&\leq \sqrt 2 \zeta \\
    \|V^{-1/2}(\hat V_{\rm raw}-V)V^{-1/2}\|_{F}&\leq 2\|\hat V_{\rm raw}-V\|_F=4\|\hat \Sigma_h-\Sigma_h\|_F\leq 4\zeta 
\end{align}
proving \eqref{eq:riuproj-V-relf}. This gives $\hat V_{\rm raw}\succeq(1/2-2\zeta)I\succ0$ for $\zeta<1/4$ and $\|\hat V_{\rm raw}-V\|_F \leq 2\zeta $.

We now prove the key inequality, $\|\hat V-\hat V_{\rm raw}\|_F\lesssim \zeta $, which argues that removing the diagonal matrix in $\hat V_{\rm raw}=\hat S\hat D\hat S^\top $ perturbs the estimate by $O(\zeta)$ in Frobenius norm. Manipulating,
\begin{align}\label{eq:riuproj-proj-frob-err-1}
    \|\hat V_{\rm raw}-\hat V\|_F=\|\hat S(\hat D-I)\hat S^\top \|_F \leq \|\hat S\|_{\rm op}^2 \|\hat D-I\|_F=\|\hat S\hat S^\top \|_{\rm op}\|\hat D-I\|_F 
\end{align}

We do this by taking inspiration from \citep{chen2026towards} and introducing the auxiliary matrix $A(W)\equiv -\Omega W\Omega W$, where $W$ is any real positive-definite matrix. Since $\hat S$ and $S$ are symplectic, $\Omega \hat S =\hat S^{-\top}\Omega$ and $\Omega S=S^{-\top }\Omega$. Then, 
\begin{align}
    A(\hat V_{\rm raw})&=-\Omega \hat S \hat D \hat S^\top \Omega \hat S\hat D \hat S^\top =-\hat S^{-\top }\Omega \hat D \Omega \hat D \hat S^{\top }=\hat S^{-\top }\hat D^2 \hat S^\top 
\end{align}
where we use the fact that $\Omega^2=-I$ and $\hat D, \Omega $ share the same block structure and commute. A similar calculation gives that $A(V)=S^{-\top } I^2S^\top=I $, meaning $A(\hat V_{\rm raw})-A(V)=\hat S^{-\top }(\hat D^2-I)\hat S^{\top }$. Then, writing the symplectic eigenvalues of $\hat V_{\rm raw}$, $\hat D={\rm diag}(\hat \lambda_1,\dots,\hat \lambda_n,\hat \lambda_1,\dots,\hat \lambda_n)$ and both $A(\hat V_{\rm raw })-A(V)$ and $\hat D^2-I$ have the eigenvalues $\{\hat\lambda_1^2-1,\dots,\hat \lambda_n^2-1,\hat \lambda_1^2-1,\dots,\hat \lambda_n^2-1\}$. Then, since $\hat V_{\rm raw}\succ 0$ implies $\hat \lambda_i>0$, 
\begin{align}
\|\hat D-I\|_F^2&=2\sum_{i=1}^{n}(\hat \lambda_i-1)^2\leq 2\sum_{i=1}^n(\hat \lambda_i^2-1)^2=\sum_{j=1}^{2n}|\lambda_j(A(\hat V_0)-A(V))|^2\\
&\leq \|A(\hat V_{\rm raw})-A(V)\|_F^2\\
&\leq \|-\Omega (\hat V_{\rm raw}-V) \Omega V-\Omega V \Omega (\hat V_{\rm raw}-V) -\Omega (\hat V_{\rm raw}-V) \Omega (\hat V_{\rm raw}-V)\|_{F}^2\\
    &\leq (2\|V\|_{\rm op}\|\hat V_{\rm raw}-V\|_{F}+\|\hat V_{\rm raw}-V\|_{F}^2)^2\\
    &\leq 9^2 \zeta^2 \label{eq:riuproj-D-I-frob}
\end{align}
Here, the first line uses the elementary facts that $|x-1|\leq |x^2-1|$ for any $x>0$ and $\sum_j |\lambda_j(M)|^2\leq \|M\|_F^2$ for any square matrix $M$. The remainder of the manipulation proceeds from the prior absolute bound $\|\hat V_{\rm raw}-V\|_F\leq 2\zeta $ and $\zeta < 1/4$. 

We now switch our focus to the first term of \eqref{eq:riuproj-proj-frob-err-1}. Taking $\zeta < 1/18$, from \eqref{eq:riuproj-D-I-frob}, $\hat D \succeq (1-9\zeta)I\succeq I/2$ and $\hat V_{\rm raw}\preceq V+2\zeta I\leq 3I$, from $\|\hat V_{\rm raw}-V\|_F\leq 2\zeta$ and $V\preceq 2I$. Then, 
\begin{align}\label{eq:riuproj-sst-psd-ub}
    \hat S \hat S^\top \leq 2 \hat S \hat D \hat S^\top =2\hat V_{\rm raw}\preceq 6I
\end{align}
Substituting \eqref{eq:riuproj-sst-psd-ub} and \eqref{eq:riuproj-D-I-frob} into \eqref{eq:riuproj-proj-frob-err-1}, we have that $\|\hat V_{\rm raw}-\hat V\|_F\lesssim \zeta$. Then, by triangle inequality, 
\begin{align}
    \|\hat V-V\|_F\leq \|\hat V-\hat V_{\rm raw}\|_F+\|\hat V_{\rm raw}-V\|_F \lesssim \zeta+2\zeta\lesssim \zeta 
\end{align}
meaning $\|V^{-1/2}(\hat V-V)V^{-1/2}\|_F\lesssim \zeta $ since $I/2\preceq V$. Then, this relative bound on the covariance combined with the relative bound on the mean in \eqref{eq:riuproj-V-relf} gives that $\tvd(\c N ({\bf m},V/2),\c N (\hat {\bf m},\hat V/2))\lesssim \zeta$ as well. Applying Lemma B.3 of \citep{chen2026towards} to translate this closeness in phase space to trace distance closeness between the two states,
\begin{equation}
    \trd(\ketbra{\psi},\ketbra*{\hat\psi})\leq 150\cdot\tvd(\cN({\bf m},V/2),\cN(\hat {\bf m},\hat V/2))\lesssim \zeta
\end{equation}
provided $\zeta$ is sufficiently small such that this total variation distance is less than $1/600$.
\end{proof}

\subsubsection{Proof of Lemma~\ref{lemma:barg-derivatives}}
\label{app:proof-barg-derivatives}
\begin{proof}
\sitan{I think a lot of this can be simplified by using Cauchy's estimate?}
WLOG, by multilinearity, $h_1,\dots,h_k$ each have unit norm $\|\cdot \|_{ \c P}=1$. We work with the unnormalized vector $\ket{\Phi(\theta)}=e^{K_\theta}\ket{0}$ for now. Since $K_\theta$ is linear in $\theta$ and all creation operators commute,
\begin{align}
{e^{K_{\theta+th}}}&=e^{K_\theta+tK_h}=e^{K_\theta}e^{tK_h}\\
   \ket{ D\Phi(\theta)[h_1]}&=\frac{d}{dt}\bigg|_{t=0}e^{K_{\theta +th_1}}\ket 0={K_{h_1}}e^{K_\theta }\ket 0
   \label{eq:D-Phi-single}
   \\
    \ket{D^k \Phi(\theta)[h_1,\dots,h_k]}&=K_{h_1}\cdots K_{h_k}e^{K_\theta}{\ket0}=\sum_{j=0}^\infty\frac{K_{h_1}\cdots K_{h_k}K_\theta^j\ket0}{j!} \label{eq:Dk-Phi-series}
\end{align}
Focusing on bounding this numerator, we sequentially apply Lemma~\ref{eq:Ktheta-growth} to the $j+k$ operators in the product, recalling that each term increases the total photon count by at most two:
\begin{align}
    \|K_{h_1}\cdots K_{h_k}K_\theta^j\ket0\|_2
    &{\leq}\left(\prod_{s=0}^{j+k-1}2(2s+2)\right)\|\theta\|_{\mathcal P}^j\cancel{\prod_{i=1}^k\|h_i\|_{\mathcal P}}\\
    &=4^{j+k}(j+k)!\|\theta\|_{\mathcal P}^j\label{eq:creator-product}
\end{align}
Applying triangle inequality to \eqref{eq:Dk-Phi-series} and then using \eqref{eq:creator-product} to bound each term, since $\|\theta\|_{\c P}\leq 1/8$, 
\begin{align}\label{eq:Dk-l2}
    \| \ket{D^k \Phi(\theta)[h_1,\dots,h_k]}\|_2 \leq 4^k k! \sum_{j=0}^\infty \binom{j+k}{k}(4\|\theta\|_{\c P})^j = \frac {4^k k!} {(1-4r)^{k+1}}\leq 2\cdot 8^k k
\end{align}
where the equality uses the generalized binomial theorem $\sum_{j= 0}^\infty\binom{j+k}{k}x^j=(1-x)^{-k-1}$ for $|x|<1$. 

We now revive the normalization $Z(\theta)=\braket{\Phi(\theta)}{\Phi(\theta)}$ and use the multivariable product rule for the Fr\`echet derivative:
\begin{align}
   |D^k Z(\theta)[h_1,\dots,h_k]|&=\left|\sum_{S\subseteq [k]} \braket{D^{|S|}\Phi(\theta)[h_i:i\in S] }{D^{k-|S|}\Phi(\theta)[h_i:i\not\in S]}\right|\\
   &\leq \sum_{j=0}^k \binom{k}{j}\left\|\bra{D^j\Phi(\theta)}\right\|_{\rm op}\left\|\ket{D^{k-j}\Phi(\theta)}\right\|_{\rm op}\\
   &{\leq}\sum_{j=0}^k \binom{k}{j}(2\cdot 8^j j! ) (2\cdot 8^{k-j}(k-j)!)\\
   &=4\cdot 8^k  (k+1)! \label{eq:Dk-Z-bound}
\end{align}
where we group by $|S|=j$ in the second line after triangle inequality and also apply \eqref{eq:Dk-l2} in the third line, with $\|\cdot \|_{\rm op}$ denoting the operator norm with respect to $\|\cdot \|_{\c P}$. 

We now analyze $N(\theta)=g\circ Z$ where $g(x)=x^{-1/2}$. Since $Z\geq 1$ from the vacuum term, we restrict to $x\geq 1$ and can bound $|g^{(m)}(x)|=x^{-m-1/2}\prod_{j=0}^{m-1}(j+1/2)\leq m!$. 
Then, applying the higher-order chain rule for the Fr\`echet derivative \citep{constantine1996multivariate} along with \eqref{eq:Dk-Z-bound} and $|g^{(m)}(x)|\leq m!$, 
\begin{align}
    |D^k N(\theta)[h_1,\dots,h_k]|
    &=\left |\sum_{\mathcal S\in\mathfrak P_k}g^{(|\mathcal S|)}(Z(\theta))
      \prod_{S\in\mathcal S}D^{|S|}Z(\theta)[h_i:i\in S]\right|\\
      &\leq \sum_{\mathcal S\in\mathfrak P_k}\underbrace{|g^{(|\mathcal S|)}(Z(\theta))|}_{\leq |\c S|!}
      \prod_{S\in\mathcal S}\underbrace{\|D^{|S|}Z(\theta)\|_{\rm op}}_{4\cdot 8^{|S|} (|S|+1)!}\\
      &\leq \sum_{j=1}^k 4^j\sum_{\mathcal S\in\mathfrak P_k: |\c S|=j} j! \prod_{S\in \c S } 16^{|S|}|S|!\\
      &\leq 16^k\sum_{j=1}^k 4^j\underbrace{\sum_{\c  S\in \mathfrak P_k:|\c S|=j}j!\prod_{S\in \c S} |S|!}_{k!\binom{k-1}{j-1}}\\
      &\leq 16^k k! \cdot 4\cdot 5^{k-1}\\
     \|D^k N(\theta)\|_{\rm op} &\leq 80^k k! \label{eq:Dk-N-op}
\end{align}
where ${\mathfrak P}_k$ is the set of partitions of $[k]$. The bracketed term in the fourth line simplifies by noticing that both quantities count the number of ways to group $k$ distinguishable elements into $j$ groups, either by enumerating over all partitions of size $j$ and permuting within each group and then permuting the groups themselves, or alternatively permuting all the elements and then choosing the $j-1$ cuts that define the groups. The fifth line follows from the binomial theorem ${\sum_{r=1}^{k}\binom{k-1}{r-1}x^r=x(1+x)^{k-1}}$ with $x=4$.

Then, since $\ket{\psi(\theta)}=N\ket{\Phi(\theta)}$, we apply the multivariable product rule one last time:
\begin{align}
    \|\ket{D^k \psi(\theta)[h_1,\dots,h_k]}\|_{ 2}&=\left\|\sum_{S\subseteq [k]}D^{|S|}N(\theta)[h_i:i\in S] \cdot \ket{D^{k-|S|}\Phi(\theta)[h_i:i\not\in S]}\right\|_{2}\\
    &\leq {\sum_{j=0}^k}\binom{k}{j}\underbrace{\|D^{j}N(\theta)\|_{\rm op}}_{{\leq}80^j j!}\underbrace{{\|D^{k-j}\Phi(\theta)\|_{\rm op}}}_{{\leq}2\cdot 8^{k-j}(k-j)!}\\
    &\leq 2\cdot 88^k k! 
\end{align}
where we use \eqref{eq:Dk-l2}, \eqref{eq:Dk-N-op}, and binomial theorem for $(80+8)^k $. Thus, the lemma holds for $C_D= 176$ or better. 
\end{proof}

\subsubsection{Proof of Lemma~\ref{lemma:hf-gradient-error}}
\label{app:proof-hf-gradient-error}
\begin{proof}
    Let $e_1,\dots,e_{d_{\c P}}$ be a real orthonormal basis for $\c P $, where $d_{\c P}=n^2+3n$. Taking $h=\kappa/(3\sqrt d) $, estimate $f_\rho(\cdot)$ at each of $\theta,\theta+he_1,\dots,\theta+h e_d$ in $B_R$ to accuracy $\kappa_0\equiv \kappa^2/(12d)$ and failure probability $\delta/(d+1)$. Then, construct the zeroth-order estimator:
    \begin{align}\label{eq:hf-finite-difference}
        \widehat{\nabla f_\rho}(\theta)
      =\sum_{i=1}^d\frac{\hat f_\rho(\theta+he_i)-\hat f_\rho(\theta)}{h}\,e_i.
    \end{align}
    Then, by the fundamental theorem of calculus and \eqref{eq:convex-smooth}, 
\begin{align}
    f_\rho(\theta+he_i)&=f_\rho(\theta)+hD f_\rho(\theta)[e_i]+\int_0^h (h-t)g''(t) dt \\
    \left|\frac{f_\rho (\theta+he_i)-f_\rho(\theta)}{h}-D f_\rho(\theta)[e_i]\right|&=\frac 1 h \int_0^h (h-t)D^2 f_\rho(\theta+te_i)[e_i,e_i]dt\leq \frac 3 h \int_0^h (h-t)dt =\frac {3h}{2}
\end{align}
Since estimating $f_\rho(\cdot)$ contributes at most $2\kappa_0/h$ error to the zeroth-order estimate in each coordinate, by triangle inequality, 
\begin{align}
    \|\widehat{\nabla f_\rho}(\theta)-\nabla f_\rho(\theta)\|_{\c P}
    &\leq\sqrt d\left(\frac{3h}{2}+\frac{2\kappa_0}{h}\right)
      =\kappa.
\end{align}
The copy complexity and failure probability are immediate. 
\end{proof}

\subsection{Deferred proofs from Section~\ref{sec:lf}}
\label{app:deferred-lf}

\subsubsection{Proof of Lemma~\ref{lemma:lfpr-warm-frame}}
\label{app:proof-lfpr-warm-frame}
\begin{proof}
The true precision $K=\Sigma^{-1}$ belongs to $\mathcal K_A$.
Moreover, $\Sigma\succeq I/2$, $\hat\Sigma\succeq\Sigma/t_\alpha$,
and $\|\hat\Sigma-\Sigma\|_F\leq At_\alpha^2$. The resolvent identity
therefore gives
\[
 \|\hat K-K\|_F\leq
 \|\hat\Sigma^{-1}\|_{\rm op}
 \|\hat\Sigma-\Sigma\|_F\|\Sigma^{-1}\|_{\rm op}
 \lesssim At_\alpha^3.
\]
Projection optimality and the triangle inequality imply
$\|K_0-K\|_F\lesssim At_\alpha^3$. The spectral constraints give
$\|K_0^{-1}\|_{\rm op}\lesssim A$ and
$\|V_0^{\pm1}\|_{\rm op}\lesssim A$. Since
$\|K^{-1}\|_{\rm op}=\|\Sigma\|_{\rm op}\leq A$, we obtain
\begin{align*}
 \|V_0-V\|_F
 &=2\|K_0^{-1}(K-K_0)K^{-1}\|_F
       \lesssim A^3t_\alpha^3,\\
 \|Z-I\|_F
 &=\|S_0^{-1}(V-V_0)S_0^{-\top}\|_F
       \lesssim A^4t_\alpha^3,
 \qquad Z=S_0^{-1}VS_0^{-\top}.
\end{align*}
For the transformed mean $e=S_0^{-1}({\bf m}-\hat\mu)$, the mean
guarantee similarly gives
\[
 \|e\|_2^2\leq\|V_0^{-1}\|_{\rm op}\|{\bf m}-\hat\mu\|_2^2
 \lesssim A^2t_\alpha(1+t_\alpha).
\]
The eigenvalues of the pure covariance $Z$ occur in reciprocal pairs.
For each pair, writing $x=\lambda+\lambda^{-1}-2\geq0$ gives
\[
 (\lambda-1)^2+(\lambda^{-1}-1)^2=x^2+2x\geq2x.
\]
Thus $\tr Z-2n\leq\|Z-I\|_F^2/2$. The transformed photon number is
therefore bounded by
\[
 \nu=\frac{\tr Z-2n}{4}+\frac{\|e\|_2^2}{2}
 \lesssim t_\alpha^6A^8+t_\alpha(1+t_\alpha)A^2
 \leq u_\alpha A^8.
\]
For example, the preceding estimates give
$\|Z-I\|_F\leq96A^4t_\alpha^3$ and
$\|e\|_2^2\leq3A^2t_\alpha(1+t_\alpha)$, so the stated value of $u_\alpha$ suffices.
% This proves \eqref{eq:lfpr-warm-energy} with $p=8$.
\end{proof}

\subsubsection{Proof of Lemma~\ref{lemma:lfpr-compatibility}}
\label{app:proof-lfpr-compatibility}
\begin{proof}
Conjugating $N_G+N_H$ by $G$ gives $N+UNU^\dagger$, where
\[
    UNU^\dagger
       =\frac12(\hat{\mathbf R}-\mathbf d)^\top C^{-1}
                    (\hat{\mathbf R}-\mathbf d)-\frac n2 I.
\]
Set
\[
    M=I+C^{-1},\qquad
    u=M^{-1}C^{-1}\mathbf d=(I+C)^{-1}\mathbf d.
\]
Expanding the quadratic terms and using $Mu=C^{-1}\mathbf d$ gives
\begin{align*}
    N+UNU^\dagger
       &=\frac12\hat{\mathbf R}^\top M\hat{\mathbf R}
          -\mathbf d^\top C^{-1}\hat{\mathbf R}
          +\frac12\mathbf d^\top C^{-1}\mathbf d-nI\\
       &=\frac12(\hat{\mathbf R}-u)^\top M(\hat{\mathbf R}-u)
          +\frac12\mathbf d^\top
             [C^{-1}-C^{-1}M^{-1}C^{-1}]\mathbf d-nI\\
       &=\frac12(\hat{\mathbf R}-u)^\top M(\hat{\mathbf R}-u)
          +\frac12\mathbf d^\top(I+C)^{-1}\mathbf d-nI.
\end{align*}
Here scalar terms multiply the identity. For the last equality, use
$M^{-1}=(I+C)^{-1}C$ and
$I-(I+C)^{-1}=C(I+C)^{-1}$ to obtain
\[
    C^{-1}-C^{-1}M^{-1}C^{-1}
       =C^{-1}[I-(I+C)^{-1}]=(I+C)^{-1}.
\]
A displacement removes $u$ without changing the spectrum.
A passive change of modes diagonalizes the pure covariance $C$ into
reciprocal pairs $(e^{2r_j},e^{-2r_j})$. The corresponding symplectic eigenvalues of $M=I+C^{-1}$ are
\[
    \omega_j=\sqrt{(1+e^{2r_j})(1+e^{-2r_j})}=2\cosh r_j.
\]
The Gaussian transformation obtained from the Williamson decomposition of $M$ turns the centered quadratic Hamiltonian into
$\sum_j\omega_j(\hat{a}_j^\dagger \hat{a}_j+1/2)$, whose ground energy is
$\tfrac12\sum_j\omega_j=\sum_j\cosh r_j$.
Adding the scalar terms above proves \eqref{eq:lfpr-compatibility-formula}.
\end{proof}

\subsubsection{Proof of Lemma~\ref{lemma:lf-energy-domain}}
\label{app:proof-lf-energy-domain}
\begin{proof}
The domain is the intersection of three norm balls in the real vector space $\c P_n$, so it is convex. Its operator-norm constraint ensures $\|A\|_{\rm op}<1$, as required for the Bargmann state.

Suppose first that $\nu(\pi_{b,A})\leq h$. Each term in \eqref{eq:lf-bargmann-energy} is nonnegative, so $\|\beta\|_2^2\leq h$ and $s_j^2/(1-s_j^2)\leq h$ for every $j$. Rearranging the latter inequality gives $s_j^2\leq h/(h+1)$, and therefore $\|A\|_{\rm op}\leq r_h$. Moreover,
\[
    \|A\|_F^2=\sum_j s_j^2
       \leq\sum_j\frac{s_j^2}{1-s_j^2}\leq h,
    \qquad
    \|b\|_2\leq(1+\|A\|_{\rm op})\|\beta\|_2
       \leq2\sqrt h.
\]
These are precisely the three constraints defining $\cD_h$.

Conversely, suppose $(b,A)\in\cD_h$. Since $s_j\leq r_h$ and $1-r_h^2=1/(h+1)$,
\[
    \sum_j\frac{s_j^2}{1-s_j^2}
       \leq\frac{\sum_j s_j^2}{1-r_h^2}
       =\frac{\|A\|_F^2}{1-r_h^2}
       \leq h(h+1).
\]
Moreover, \eqref{eq:lf-displacement-bounds} and the identity
\[
    \frac1{1-r_h}=\frac{1+r_h}{1-r_h^2}
       =(1+r_h)(h+1)\leq2(h+1)
\]
give
\[
    \|\beta\|_2^2
       \leq\frac{\|b\|_2^2}{(1-r_h)^2}
       \leq16h(h+1)^2.
\]
The conclusion immediately follows.
\end{proof}

\subsubsection{Proof of Lemma~\ref{lemma:lf-cutoff-tomography}}
\label{app:proof-lf-cutoff-tomography}
\begin{proof}
\emph{Estimate matrix entries}
Let $D=D_K$ and let $\ket{e_1},\ldots,\ket{e_D}$ be the canonical orthonormal basis for the cutoff space. For each $i$ measure $E_i=\ket{e_i}\bra{e_i}$, and for each $i<j$ measure
\[
    R_{ij}=\frac{\ket{e_i}\bra{e_j}+\ket{e_j}\bra{e_i}}2,
    \qquad
    J_{ij}=\frac{\mathrm i\ket{e_i}\bra{e_j}-\mathrm i\ket{e_j}\bra{e_i}}2.
\]
The expectations of these observables on $\sigma$ are, respectively, $(\sigma_K)_{ii}$, $\Re(\sigma_K)_{ij}$, and $\Im(\sigma_K)_{ij}$. Their measurement outcomes lie in $[-1,1]$.

There are $D+2\binom D2=D^2$ quantities to estimate. Let $\delta=\epsilon/(64D)$ and measure each observable on
\[
    m=\left\lceil2\delta^{-2}\log(2D^2/\eta)\right\rceil
\]
fresh copies. Hoeffding's inequality bounds the error probability for each empirical mean by $2e^{-m\delta^2/2}\leq\eta/D^2$, so by the union bound all errors are at most $\delta$ with probability at least $1-\eta$.

Let $\tilde\sigma_K$ be defined using these empirical estimates, and make it Hermitian by taking conjugates. On this event the diagonal errors have magnitude at most $\delta$, and the squared magnitude of each off-diagonal error is at most $2\delta^2$. Hence
\[
    \|\tilde\sigma_K-\sigma_K\|_F^2
       \leq D\delta^2+2D(D-1)\delta^2
       \leq2D^2\delta^2,
    \qquad
    \|\tilde\sigma_K-\sigma_K\|_F\leq\epsilon/32.
\]
The copy count is $D^2m=O(D^4\epsilon^{-2}\log(2D/\eta))$.

\emph{Project onto set of subnormalized states}
Let $\hat\sigma_K$ be a nearest matrix to $\tilde\sigma_K$ in Frobenius norm in the closed convex set
\[
    \mathcal S_D=\{X=X^\dagger:X\succeq0,\ \tr X\leq1\}.
\]
It is not hard to see that this projection can be computed by diagonalizing $\tilde\sigma_K$ and projecting its eigenvalues onto $\{x\in\R^D:x_i\geq0,\ \sum_i x_i\leq1\}$. Since $\sigma_K\in\mathcal S_D$, $\tilde\sigma_K$ is at most $\epsilon/32$ from $\mathcal S_D$, so $\hat \sigma_K$ is at most $\epsilon/16$ from $\sigma_K$.
\end{proof}

\subsection{Deferred proofs from Section~\ref{sec:gaussian-measurements}}
\label{app:deferred-gaussian-measurements}

\subsubsection{Proof of Lemma~\ref{eq:lb-match-wu}}
\label{app:proof-lb-match-wu}
\begin{proof}
From Lemma 7 of \citep{wu2019chebyshev} and 2.11.1 of \citep{timan2014theory}, there exist two unit mean random variables on $\{0\}\cup [1+\nu,\lambda]$who match the first $k$ moments and have:
\begin{align}
    \Pr[U_+>0]-\Pr[U_->0]=\inf_{p\in \c P_L}\sup_{x\in [1,Ck^2]}|1/x-p(x)|&=\frac{\left(1+\sqrt{\frac {1+\nu}{\lambda}}\right)^2}{1+\nu}\left(1-\frac{2\sqrt{\frac {1+\nu}{\lambda}}}{1+\sqrt{\frac {1+\nu}{\lambda}}}\right)^k\\
    &=(1+\lambda^{-1/2})^2\left(1-\frac{2}{1+\sqrt \lambda}\right)^k
\end{align}
where $\c P_L$ is the set of real polynomials of degree at most $L$ and where we take $\nu=0$ and $\lambda=Ck^2$.\alvantodo{check} Bounding, 
\begin{align}
    \Pr[U_+>0]-\Pr[U_->0]=(1+\lambda^{-1/2})^2\left(1-\frac{2}{1+\sqrt \lambda}\right)^k\geq 1-\frac{2k}{1+\sqrt \lambda }\geq 1-\frac{2}{\sqrt C}\equiv c
\end{align}
for sufficiently large $C$.
\end{proof}

\subsection{Deferred proofs from Section~\ref{sec:clique-lb}}
\label{app:deferred-clique-lb}

\subsubsection{Proof of Lemma~\ref{lemma:P1-stability}}
\label{app:proof-P1-stability}
    \begin{proof}
        We begin with the one-mode setting where $b,A\in \C$. Since $p_1(\ket{\psi})$ does not change if we rotate $\ket{\psi}\mapsto e^{-i\theta \hat{a}^\dagger \hat{a}}\ket{\psi}=\ket*{\psi(e^{-i\theta}b,e^{-2i\theta}A)}$, WLOG, we can take $b>0$ and $A\in \C$. Identifying $b=\beta$ and $A=\alpha$, we see then that all pure Gaussian states are in the family, and using \eqref{eq:barg-state-expand} and \eqref{eq:Z(beta,alpha)},
        \begin{align}\label{eq:p1-bounding}
        p_1(\ket{\psi})=\frac{\beta^2}{Z(\beta,\alpha)}=\beta^2\sqrt{1-|\alpha|^2}\exp(-\frac{\beta^2(1+\Re \alpha)}{1-|\alpha|^2})&\leq \beta^2\sqrt{1-|\alpha|^2}\exp(-\frac{\beta^2\cancel{(1-|\alpha|)}}{\cancel{(1-|\alpha|)}(1+|\alpha|)})\\
        &\leq e^{-1}(1+|\alpha|)\sqrt{1-|\alpha|^2}\\
        &\leq p_1^*
        \end{align}
        where the first inequality uses $\Re \alpha \geq -|\alpha|$, the second inequality uses the fact that $te^{-t/a}$ is maximized at $t=a$, and third inequality uses the fact that $(1+t)\sqrt{1-t^2}$ is maximized at $t=1/2$. Then, equality throughout gives that $\Re \alpha=-|\alpha|=-1/2$ and $\beta^2=1+|\alpha|=3/2$, meaning $(\beta,\alpha)=(\beta^*,\alpha^*)$. 
        
        We now show the desired stability bound near the maximizer, beginning with stability in parameter distance. Packaging $\theta=(\beta,\Re \alpha,\Im \alpha)\in \R^3$ and $\theta^*=(\beta^*,\alpha^*,0)$, define the objective:
        \begin{align}
        -\log p_1(\ket\psi)\equiv f(\theta)&=-2\log \theta_1 -\frac 1 2 \log (1-\theta_2^2-\theta_3^2)+\frac{\beta^2(1+\theta_2)}{1-\theta_2^2-\theta_3^2}
        \end{align}
        Differentiating, one can verify that $\nabla f(\theta^*)=0$, $\nabla^2 f(\theta^*)=\begin{bmatrix}
        8/3&4\sqrt6/9&0\\
        4\sqrt6/9&28/9&0\\
        0&0&4
    \end{bmatrix}\succeq \frac{16}{9}I$, and every entry of $\nabla^2 f(\theta)$ is continuous when $\theta_1>0$ and $\theta_2^2+\theta_3^2<1$. Then, by continuity of the Hessian and $p(\theta)$, there exists a universal constant $\xi$ for which any $\theta$ with $\|\theta-\theta^*\|_2< \xi$ has $\nabla^2 f(\theta)\succeq I$ and $p_1(\theta)\geq p_1^*/2$. For now, restrict $\theta$ to this ball. Using this and the fundamental theorem of calculus applied twice,
    \begin{align}
       f(\theta)-f(\theta^*)&=\int_0^1(1-t)(\theta-\theta^*)^\top \nabla^2 f(\theta^*+t(\theta-\theta^*))(\theta-\theta^*)dt \geq \int_0^1 (1-t)\|\theta-\theta^*\|_2^2dt\\
       &\geq \frac 1 2 \|\theta-\theta^*\|_2^2
    \end{align}
    This then gives stability in parameter distance:
    \begin{align}\label{eq:p1-param-stable}
        |\beta-\beta^*|^2+|\alpha-\alpha^*|^2\leq f(\theta)-f(\theta^*)=\log \frac{p_1^*}{p_1(\ket{\psi})}\leq \frac{p_1^*-p_1(\ket{\psi})}{p_1(\ket{\psi})}\lesssim p_1^*-p_1(\ket{\psi})
    \end{align}
    where use the fact that $\log t \leq t-1$ and $p_1(\theta)\geq p_1^*/2$. 
    
    We now convert this parameter distance control into stability in infidelity. From \eqref{eq:lf-bargmann-annihilation}, we have that $(\hat{a}-\alpha^*\hat{a}^\dagger-\beta^*)\ket{\psi^*}=0$ where $[\hat{a}-\alpha^*\hat{a}^\dagger-\beta^*,\hat{a}^\dagger-\alpha^*\hat{a}-\beta^*]=(1-(\alpha^*)^2)I$. Then, define the operator:
    \begin{align}
        \hat{a}_{\alpha^*,\beta^*}\equiv \frac{\hat{a}-\alpha^*\hat{a}^\dagger-\beta^*}{\sqrt{1-(\alpha^*)^2}}
    \end{align}
    such that $\hat{a}_{\alpha^*,\beta^*}\ket{\psi^*}=0$ and $[\hat{a}_{\alpha^*,\beta^*},\hat{a}_{\alpha^*,\beta^*}^\dagger]=1$, meaning. Since this operator is just formed by squeezing and displacement, $\hat{a}_{\alpha^*,\beta^*}=U\hat{a}U^\dagger$ for some Gaussian unitary $U$ where $U\ket{0}=\ket{\psi^*}$. Then, since $I-\ketbra{0}\preceq \hat{a}^\dagger \hat{a}$, conjugating by $U$ gives $I-\ketbra{\psi^*}\preceq \hat{a}_{\alpha^*,\beta^*}^\dagger \hat{a}_{\alpha^*,\beta^*} $. Considering the infidelity, 
    \begin{align}
        \infd(\ket{\psi},\ket{\psi^*})&=\bra{\psi}(I-\ketbra{\psi^*})\ket{\psi}\leq \|\hat{a}_{\alpha^*,\beta^*}\ket{\psi}\|_2^2=\frac 4 3 \|(\beta-\beta^*)\ket{\psi}+(\alpha-\alpha^*)\hat{a}^\dagger \ket{\psi}\|_2^2\\
        &\leq \frac 83 |\beta-\beta^*|^2+|\alpha-\alpha^*|^2\|\hat{a}^\dagger \ket{\psi}\|_2^2\\
        &\lesssim (p_1^*-p_1(\ket{\psi}))(1+\|\hat{a}^\dagger \ket{\psi}\|_2^2)
    \end{align}
    where the second inequality uses the fact $\|u+v\|_2^2\leq 2 \|u\|_2^2+2\|v\|_2^2$ and the final bound applies \eqref{eq:p1-param-stable}. Then, since $\hat{a}\hat{a}^\dagger=\hat{a}^\dagger \hat{a}+I$, 
    $$\|\hat{a}^\dagger \ket{\psi}\|_2=\sqrt{\|\hat{a}\ket{\psi}\|_2^2+1}\leq \|\hat{a}\ket{\psi}\|_2+1=\|(\alpha \hat{a}^\dagger +\beta)\ket{\psi}\|_2+1\leq \beta+|\alpha|\|\hat{a}^\dagger \ket{\psi}\|_2+1$$
    meaning $\|\hat{a}^\dagger \ket{\psi}\|_2\leq \frac{1+\beta}{1-|\alpha|}\lesssim 1$ since $\alpha$ is bounded away from $1$ by $\|\theta-\theta^*\|< \xi $. Thus, there exists a constant $C$ for which $\infd(\ket{\psi},\ket{\psi^*})\leq C(p_1^*-p_1(\ket{\psi}))$ when $\|\theta-\theta^*\|< \xi$. 

    Then, to finish the claim, it suffices to show that for $\|\theta-\theta^*\|_2\geq  \xi$, $p_1^*-p_1(\ket{\psi})\geq \gamma$ for some $\gamma > 0$ so that $\infd(\ket{\psi},\ket{\psi^*})\leq C'(p_1^*-p_1(\ket{\psi}))$ for $C'=\max \{C,1/\gamma\}$. Outside the ball, if $p_1<p_1^*/2$, we are already bounded away. Further, from the bounding \eqref{eq:p1-bounding}, it is clear that $\{\theta:p_1(\ket{\psi(\theta)})\geq p_1^*/2\}$ is compact. Then, the intersection of this set with the complement of the ball is also compact, making $\theta^*$ the unique minimizer on this region, which finishes the one-mode case. \alvantodo{make more rigorous?}

    We now generalize to arbitrary modes. Consider $\infd(\ket{\psi},\ket{\psi^*(x)})$ for $x=b/\|b\|$. Then, applying a passive unitary, by the unitary invariance of fidelity, 
    \begin{align}\label{eq:p1-multi-passive-infd}
        \infd(\ket{\psi(b,A)},\ket{\psi^*(x)})=\infd(\ket{\psi'(b',A')},\ket{\psi^*(e_1)})
    \end{align}
    where $b'=\|b\|_2 e_1=Wb$ for unitary matrix $W$ and $A'=WAW^\top$. Looking at \eqref{eq:barg-state}, 
    $$\ket{\psi^*(e_1)}=\exp(\beta^* \hat{a}_1^\dagger+ \frac{\alpha^*}{2}(\hat{a}_1^\dagger)^2)\ket{0}^{\otimes n}= \ket{\psi^*}\otimes \ket{0}^{\otimes (n-1)}$$
    where $\ket{\psi^*}$ is the single-mode state in the family that maximizes $p_1(\cdot)$. Then, considering the projection of $\ket{\psi'}$ onto the first mode given by projector $P_{e_1}=I\otimes \ketbra{0}^{\otimes (n-1)}$ which projects onto first-mode states,
    \begin{align}
        P_{e_1}\ket{\psi'}&=\frac{1}{\sqrt{Z(b',A')}}\exp(\|b\|_2\hat{a}_1^\dagger + \frac{A_{11}'}{2}(\hat{a}_1^\dagger)^2)\ket{0}^{\otimes n}=\kappa \ket{\psi(\beta,\alpha)}\otimes \ket{0}^{\otimes  (n-1)}
    \end{align}
    for $\kappa=\frac{Z(\beta,\alpha)}{Z(b',A')}=\|P_{e_1}\ket{\psi'}\|_2^2\leq 1$. Then, 
    \begin{align}
        \infd(\ket{\psi'},\ket{\psi^*(e_1)})&=1-\lvert \bra{\psi^*(e_1)}P_{e_1}\ket{\psi'}\rvert^2=1-\kappa\lvert \braket{\psi^*}{\psi(\beta,\alpha)}\rvert^2=(1-\kappa)+\kappa \infd(\ket{\psi(\beta,\alpha)}, \ket{\psi^*})\\
        &\leq (1-\kappa)+\kappa C'(p_1^*-p_1(\ket{\psi(\beta,\alpha)}))\\
        &\leq \max\{1/p_1^*,C'\}(p_1^*(1-\kappa)+\kappa(p_1^*-p_1(\ket{\psi(\beta,\alpha)})))\\
        &\lesssim p_1^*-\kappa p_1(\ket{\psi(\beta,\alpha)}) \label{eq:p1-multi-final-1}
    \end{align}
    where we use the single-mode stability bound in the second line. Manipulating via \eqref{eq:barg-state}, 
    \begin{align}\label{eq:p1-multi-to-single}
        \kappa p_1(\ket{\psi(\beta,\alpha)})=\kappa \frac{\beta^2}{Z(\beta,\alpha)}=\frac{\beta^2}{Z(b',A')}=\frac{\|b'\|_2^2}{Z(b',A')}=p_1(\ket{\psi'})=p_1(\ket{\psi})
    \end{align}
    Putting \eqref{eq:p1-multi-passive-infd}, \eqref{eq:p1-multi-to-single}, and \eqref{eq:p1-multi-final-1} together, we have that $\infd(\ket{\psi(b,A)},\ket{\psi^*(x)})\lesssim p_1^*-p_1(\ket{\psi(b,A)})$ for the parameter $x=b/\|b\|_2$, from which the claim trivially follows. 
    \end{proof}

\clearpage

\section{Tensor representation of the truncated numerator}
\label{app:lf-tensor-reduction}

We prove Theorem~\ref{thm:lf-tensor-reduction}. The main point is to represent the numerator by a tensor whose Frobenius norm is bounded independently of $n$. A bound on the values of $p$ alone would not give this coefficient bound.

\paragraph{A linear map producing the truncated Gaussian vector.}
Set $w(\theta)=(1,b,A/\sqrt2)\in\mathcal X$. We construct a linear map
\begin{equation}\label{eq:lf-gaussian-lift}
    \mathcal L_K:\mathcal X^{\otimes K}\longrightarrow P_K\c F_n,
    \qquad \mathcal L_Kw(\theta)^{\otimes K}=P_K\ket{\Phi(\theta)},
\end{equation}
with $\|\mathcal L_K\|_{\rm op}\leq c_K=(K+1)^2\sqrt{K!}$.

Since the creation operators commute,
\[
    P_K e^{L_b+Q_A}\ket0
       =\sum_{2r+s\leq K}\frac{Q_A^rL_b^s}{r!s!}\ket0.
\]
We construct one linear map for each summand. The term indexed by $(r,s)$ uses $r$ copies of $A$ and $s$ copies of $b$, and lies in the $(2r+s)$-photon sector. The constant coordinate of $w(\theta)$ supplies the remaining $K-r-s$ input factors.

We first record the normalization for passing from tensors to photon states. Let $\iota_m$ be the isometry from the $m$-photon sector onto $\operatorname{Sym}^mH$ given by
\begin{equation}\label{eq:lf-fock-tensor-isometry}
    \iota_m\bigl(\hat{a}^\dagger(v_1)\cdots \hat{a}^\dagger(v_m)\ket0\bigr)
       =\sqrt{m!}\,\operatorname{Sym}_m(v_1\otimes\cdots\otimes v_m),
\end{equation}
where $\hat{a}^\dagger(v)=\sum_i v_i \hat{a}_i^\dagger$ and $\operatorname{Sym}_m$ is the orthogonal symmetrization projector. Indeed, for standard basis vectors $v_j$ with occupation counts $k_i$, the creation vector has squared norm $\prod_i k_i!$, while the symmetrized tensor has squared norm $\prod_i k_i!/m!$. The factor $\sqrt{m!}$ is thus the correct normalization.

Fix $r,s\geq0$ with $m=2r+s\leq K$. Let $E_{r,s}:\mathcal X^{\otimes K}\to H^{\otimes m}$ be the coordinate projection onto 
\[
    (\operatorname{Sym}^2H)^{\otimes r}
       \otimes H^{\otimes s}\otimes\C^{\otimes(K-r-s)},
\]
followed by the canonical inclusion into $H^{\otimes m}$. Thus $\|E_{r,s}\|_{\rm op}\leq1$, and for $w=(t,b,A/\sqrt2)$,
\[
    E_{r,s}w^{\otimes K}
       =t^{K-r-s}2^{-r/2}A^{\otimes r}\otimes b^{\otimes s}.
\]
Now, we are ready to define
\[
    \mathcal L_{r,s}
       =\frac{\sqrt{m!}}{2^{r/2}r!s!}\,
          \iota_m^{-1}\operatorname{Sym}_m E_{r,s},
\]
which corresponds to the $(r,s)$ term in the decomposition of $\Phi(\theta)$.
Using \eqref{eq:lf-fock-tensor-isometry} and $Q_A=\frac12\sum_{i,j}A_{ij}\hat{a}_i^\dagger \hat{a}_j^\dagger$, we obtain
\begin{equation}\label{eq:lf-lift-monomial}
    \begin{aligned}
    \mathcal L_{r,s}w^{\otimes K}
       &=\frac{t^{K-r-s}\sqrt{m!}}{2^r r!s!}\,
          \iota_m^{-1}\operatorname{Sym}_m
          (A^{\otimes r}\otimes b^{\otimes s})\\
       &=\frac{t^{K-r-s}}{r!s!}\,Q_A^rL_b^s\ket0.
    \end{aligned}
\end{equation}
To conclude the construction, let $\mathcal L_K=\sum_{2r+s\leq K}\mathcal L_{r,s}$, viewing each output as an element of $P_K\c F_n$. Taking $t=1$ in \eqref{eq:lf-lift-monomial} proves \eqref{eq:lf-gaussian-lift}.

Finally, $E_{r,s}$ and $\operatorname{Sym}_m$ are contractions and $\iota_m^{-1}$ is an isometry. Hence, on the full input space $\mathcal X^{\otimes K}$,
\[
    \|\mathcal L_{r,s}\|_{\rm op}
       \leq\frac{\sqrt{(2r+s)!}}{2^{r/2}r!s!}
       \leq\sqrt{K!}.
\]
There are at most $(K+1)^2$ admissible pairs $(r,s)$, so
\[
    \|\mathcal L_K\|_{\rm op}
       \leq\sum_{2r+s\leq K}\|\mathcal L_{r,s}\|_{\rm op}
       \leq(K+1)^2\sqrt{K!}.
\]

\paragraph{The tensor coefficient bound.}
Define $\mathcal B_K=\mathcal L_K^*\hat\sigma_K\mathcal L_K$. The support of $\hat\sigma_K$ and \eqref{eq:lf-gaussian-lift} imply
\[
    p(\theta)=\langle w(\theta)^{\otimes K},
                         \mathcal B_Kw(\theta)^{\otimes K}\rangle.
\]
Positivity gives $\mathcal B_K\succeq0$, while $\|\hat\sigma_K\|_F\leq\tr\hat\sigma_K\leq1$ gives
\[
    \|\mathcal B_K\|_F
       \leq\|\mathcal L_K\|_{\rm op}^2\|\hat\sigma_K\|_F
       \leq(K+1)^4K!.
\]

Set $B=R_h^{2K}\mathcal B_K$ and $y=w/R_h$. These give \eqref{eq:lf-tensor-numerator} and $\|B\|_F\leq\Gamma$. Also
\[
    \|w(\theta)\|_2^2=1+\|\theta\|_{\c P}^2
       \leq1+\tfrac92h\leq R_h^2,
    \qquad
    \|w(\theta)-w(\theta')\|_2=\|\theta-\theta'\|_{\c P},
\]
which proves the stated properties of $y$.

\paragraph{Lipschitz continuity and effective dimension.}
For $\|u\|_2,\|v\|_2\leq1$, a telescoping argument shows that $\|u^{\otimes K}-v^{\otimes K}\|_2\leq K\|u-v\|_2$. Hence
\begin{align*}
    |g(u)-g(v)|
       &\leq\|B\|_{\rm op}
          (\|u^{\otimes K}\|_2+\|v^{\otimes K}\|_2)
          \|u^{\otimes K}-v^{\otimes K}\|_2\\
       &\leq2K\Gamma\|u-v\|_2.
\end{align*}
For the claim about ``dimensionality reduction'', apply Lemma~6.5 of the \href{https://arxiv.org/abs/2411.04283v3}{extended version} of~\cite{bakshi2025learning}. Their lemma is quite general and will directly apply to our situation. The intuition is that a tensor $T$ of bounded Frobenius norm only has so many ``strong'' directions that can contribute to the injective norm, and they can be found by just flattening the tensor on each coordinate and taking the singular vectors from each SVD. 

We now apply it to our setting. Take the order-$2K$ tensor with entries $T_{\alpha,\beta}=\overline{B_{\alpha,\beta}}/\Gamma$, where $\alpha,\beta$ are length-$K$ multi-indices. To plug into their lemma, $\langle T,\bar v^{\otimes K}\otimes v^{\otimes K}\rangle=g(v)/\Gamma$, and $\|T\|_F=\|B\|_F/\Gamma\leq1$. Set all lower-degree tensors to zero and use accuracy $\xi/\Gamma$. The cited lemma gives an efficiently computable complex subspace $W$ with
\[
    \dim_{\C}W\leq8(K+1)^6(\Gamma/\xi)^2,\qquad
    |g(v)-g(P_Wv)|\leq\xi\quad(\|v\|_2\leq1).
\]

\paragraph{Construction cost.}
We have $\dim\mathcal X=1+n+n(n+1)/2\leq(n+1)^2$ and $\dim P_K\c F_n=\binom{n+K}{K}\leq(n+1)^K$. Thus all matrices and tensors above have $(n+1)^{O(K)}$ entries. The entries of $\mathcal L_K$ are computed using at most $K!$ permutations for each of at most $(K+1)^2$ terms. We use matrix multiplication to construct $B$, and taking SVDs (which are used in \cite{bakshi2025learning}) has the same $(n+1)^{O(K)}$ cost.

\section{Subgradient and projection oracles}
\label{app:lf-convex-oracles}

\begin{proof}[Proof of Proposition~\ref{prop:lf-convex-oracles}]
\emph{Subgradients.}
Here we describe how to compute the subgradient of $\lambda[\|P_Wy(\theta)-z\|_2-t]_+$. We choose subgradient $0$ when $\|P_Wy(\theta)-z\|_2\leq t$, and otherwise take $\lambda$ times the gradient of $\|P_W y(\theta)-z \|_2$. The function $\theta\mapsto\|P_Wy(\theta)-z\|_2$ is $1/R_h$-Lipschitz with respect to $\|\cdot\|_{\c P}$, by Theorem~\ref{thm:lf-tensor-reduction}. Hence our choice of penalty subgradient has norm at most $\lambda/R_h$.

Adding $\nabla_{\c P}\ell(\theta)$, computed from \eqref{eq:lf-normalization-derivative}, gives a subgradient $g_z(\theta)$ of $J_z$ satisfying
\[
    \|g_z(\theta)\|_{\c P}\leq G_h+\lambda/R_h.
\]
Clearly, all operations here are polynomial in $n$ and $\dim_\C W$.

\emph{Projection.}
The squared parameter distance is $\|b-b_0\|_2^2+\tfrac12\|A-A_0\|_F^2$, and the domain constraints separate between $b$ and $A$. Thus we project $b_0$ onto its Euclidean ball of radius $2\sqrt h$ and minimize $\|A-A_0\|_F^2$ over symmetric matrices satisfying $\|A\|_F\leq\sqrt h$ and $\|A\|_{\mathrm{op}}\leq r_h$.

\emph{First, the matrix projection preserves the Takagi factor.}
Write $A_0=U\operatorname{diag}(s)U^\top$, where $s_i\geq0$. For any feasible matrix $C$, put $D=U^\dagger C\overline U$ and $x_i=\max\{\operatorname{Re}D_{ii},0\}$. The aligned matrix $\widetilde C=U\operatorname{diag}(x)U^\top$ is still feasible, since
\[
    \|\widetilde C\|_F^2=\sum_i x_i^2\leq\|D\|_F^2=\|C\|_F^2,
    \qquad
    \|\widetilde C\|_{\mathrm{op}}=\max_i x_i
    \leq\max_i|D_{ii}|\leq\|D\|_{\mathrm{op}}=\|C\|_{\mathrm{op}}.
\]
Moreover, since $s_i\geq0$, replacing $D_{ii}$ by $x_i$ can only decrease its distance to $s_i$. Unitary invariance of the Frobenius norm therefore gives
\[
    \|A_0-\widetilde C\|_F^2
    =\sum_i(s_i-x_i)^2
    \leq\sum_i|s_i-D_{ii}|^2+\sum_{i\ne j}|D_{ij}|^2
    =\|A_0-C\|_F^2.
\]
Thus an optimal matrix has the same Takagi factor $U$ as $A_0$.

\emph{Second, project the singular-value vector.}
Since the Takagi factors are the same, it remains to solve the ordinary Euclidean projection problem
\[
    s'=\mathop{\rm argmin}_{x\in\R^n}
    \left\{\tfrac12\|x-s\|_2^2:
    \|x\|_2\leq\sqrt h,\ \|x\|_\infty\leq r_h\right\}.
\]
This projection onto the intersection of an $\ell_2$ ball and an $\ell_\infty$ ball can be computed with polynomially many real-arithmetic operations. Since $s\geq0$, its projection also satisfies $s'\geq0$. The desired matrix projection is then
\begin{equation}\label{eq:lf-domain-projection}
    A'=U\operatorname{diag}(s')U^\top.
\end{equation}
\end{proof}

\end{document}